\documentclass{article}
\usepackage[utf8]{inputenc}
\usepackage{fullpage}
\usepackage{amsthm,mathtools,scalerel}
\usepackage{color,soul,latexsym,amsmath,amssymb,amsfonts,amsthm,dsfont,enumitem,xcolor,bbm,threeparttable,graphicx,float,subfigure}
\usepackage{authblk}
\usepackage[round]{natbib}
\usepackage{listings}
\RequirePackage[colorlinks=true,allcolors=blue,hypertexnames=false]{hyperref}%
\usepackage{graphicx,todonotes}
\usepackage{outlines}
\usepackage{array}
\usepackage{selectp}

\usepackage[ruled,linesnumbered, vlined, noend]{algorithm2e}
\mathtoolsset{showonlyrefs}
\usepackage{titlesec}
\newcommand{\E}{\mathbb{E}}
\newcommand{\PP}{\mathbb{P}}
\renewcommand{\hat}{\widehat}
\newtheorem{thm}{Theorem}
\newtheorem{lem}{Lemma}

\newtheorem{prop}{Proposition}
\newtheorem{definition}{Definition}
\newtheorem*{remark}{Remark}
\newcounter{myalgctr}
\numberwithin{myalgctr}{section}
\DeclareMathOperator*{\argmin}{arg\,min}

\DeclareRobustCommand{\rchi}{{\mathpalette\irchi\relax}}
\newcommand{\irchi}[2]{\raisebox{\depth}{$#1\chi$}}
\usepackage{graphicx} 

\makeatletter
\newcommand*{\rom}[1]{\expandafter\@slowromancap\romannumeral #1@}
\def\namedlabel#1#2{\begingroup
    #2%
    \def\@currentlabel{#2}%
    \phantomsection\label{#1}\endgroup
}
\makeatother

\title{Doubly Robust Calibration of Prediction Sets under Covariate Shift}
\author[1]{
    Yachong Yang\thanks{Email and address: {\tt yachong@wharton.upenn.edu}, Academic Research Building, 265 S 37th St, Philadalphia, PA, US.}}
\author[2]{
    Arun Kumar Kuchibhotla\thanks{Email and address: {\tt arunku@cmu.edu}, Baker Hall, 4909 Frew St, Pittsburgh, PA, US.}}
\author[1]{
   Eric Tchetgen Tchetgen\thanks{Email and address: {\tt ett@wharton.upenn.edu}, Academic Research Building, 265 S 37th St, Philadalphia, PA, US.}}    
\affil[1]{Department of Statistics, University of Pennsylvania}
\affil[2]{Department of Statistics \& Data Science, Carnegie Mellon University}
\begin{document}
\maketitle

\begin{abstract}
Conformal prediction has received tremendous attention in recent years and has offered new solutions to problems in missing data and causal inference; yet these advances have not leveraged modern semiparametric efficiency theory for more robust and efficient uncertainty quantification. In this paper, we consider the problem of obtaining distribution-free prediction regions accounting for a shift in the distribution of the covariates between the training and test data. Under an \textit{explainable covariate shift} assumption analogous to the standard missing at random assumption, we propose three variants of a general framework to construct well-calibrated prediction regions for the unobserved outcome in the test sample. Our approach is based on the efficient influence function for the quantile of the unobserved outcome in the test population combined with an arbitrary machine learning prediction algorithm, without compromising asymptotic coverage. We establish that the resulting prediction sets eventually attain nominal coverage in large samples. This guarantee is a consequence of the product bias form of our proposal which implies correct coverage if either the propensity score or the conditional distribution of the response is estimated sufficiently well. Our results also provide a framework for construction of doubly robust prediction sets of individual treatment effects, under the unconfoundedness condition. We further discuss aggregation of prediction sets from different machine learning algorithms for optimal prediction and illustrate the performance of our methods in both synthetic and real data. Finally, inspired by sensitivity analysis in missing data, we briefly discuss how our proposal could be extended to account for departures from the explainable covariate shift setting.
\end{abstract}

\section{Introduction}

Prediction is a major focus of modern machine learning literature. Most machine learning methods are designed for point prediction, but accurately quantifying the uncertainty associated with a given point prediction algorithm remains an important challenge in many applications. Given independent and identically distributed (i.i.d.) pairs $\left(X_{i}, Y_{i}\right)$, $i=1, \ldots, N$, from a distribution $P = P_{X}\otimes P_{Y|X}$ supported on $\mathcal{X}\times\mathbb{R}$ (e.g., $\mathcal{X} = \mathbb{R}^d, d\ge1$),
and given a desired nominal coverage rate $1-\alpha \in (0,1)$, the goal of prediction with well-calibrated uncertainty quantification is to build a prediction set $\widehat{C}_{N,\alpha}$, such that
\begin{equation}\label{eq:general-goal}
\mathbb{P}\bigl(Y_f \in \widehat{C}_{N,\alpha} (X_f)\bigr) \ge 1 - \alpha,
\end{equation}
where the probability is taken over the marginal distribution of all the training data along with $(X_f, Y_f)$.
Note that \eqref{eq:general-goal} does not imply conditional coverage $\mathbb{P}\bigl(Y_f \in \widehat{C}_{N,\alpha} (X_f) | X_f = x_f \bigr) \geq 1-\alpha$, which is known to be impossible without assumptions over the underlying distribution as shown in  \cite{barber2019limits}. {\color{black}This goal, however, can be \emph{approximately} achieved where \emph{approximately} is meant either asymptotically or by conditioning on $X_f$ belonging to a set $A$ rather than $X_f$ being equal to, say $x_f$.}
Conformal prediction introduced by~\cite{vovk2005algorithmic} provides a simple and finite-sample valid solution to~\eqref{eq:general-goal} without any assumption on the distribution $P$, requiring only that the training data and $(X_f, Y_f)$ to be exchangeable (jointly).  It provides a valid prediction set by wrapping around any point prediction algorithm, irrespective of what the point prediction algorithm is. 

Recently several works have considered the extension of conformal prediction methodology to the case of non-exchangeable data; see Section~\ref{sec:literature-review} for a review. Our work lies in this space. One important distinction from the exchangeable case is that finite sample coverage guarantees are generally impossible for non-exchangeable data without very restrictive assumptions. Formally, we consider the problem of prediction with uncertainty quantification under covariate shift where the completely observed training data is drawn i.i.d. from $P_X\otimes P_{Y|X}$ but the test point that needs to be covered by a prediction set is drawn from $Q_X\otimes P_{Y|X}$. This problem, which we term \textit{explainable covariate shift} problem, was first posed 
by~\cite{tibshirani2019conformal} who provides a solution called ``Weighted Conformal Prediction" (WCP) which assumes that the covariate shift is known via $dQ_X/dP_X$. This method has been extended by~\cite{lei2020conformal} to allow for unknown covariate shift. 

The prediction problem under covariate shift can be equivalently stated as predicting the label for a given feature vector observed from a different covariate distribution. In terms of real-world applications, prediction under covariate shift is important in semi-supervised and transfer learning settings. In health care and related problems, it is often the case that the amount of labeled data is limited in comparison to unlabeled data. This is, particularly, true with electronic health record (EHR) data where labeling the response can be costly and/or laborious~\citep{chakrabortty2018efficient}. Instead of assuming the data are missing completely at random (i.e., observations are chosen to be labeled completely at random), it may be preferable to allow for the possibility that they are labeled based on observed covariate features. Prediction in this case is same as prediction under covariate shift as proved in Section~\ref{sec:notation-covariate} and can be useful either for imputation or for understanding the spread in the response distribution. Recently,~\cite{lei2020conformal} and~\cite{jin2021sensitivity} showed how to use prediction under covariate shift to construct prediction intervals for individual treatment effect under traditional causal inference assumptions. This can potentially be more useful in understanding the impact of a treatment at the individual level than the standard average treatment effect; please see appendix \ref{sec:notation-causal} for more details.

Therefore, prediction under covariate shift is a building block of several important prediction problems and an improvement upon the weighted conformal method of~\cite{tibshirani2019conformal} would lead to advancements in many other directions as well, which is exactly what we aim to do in this work. We reconsider the problem of prediction under covariate shift from a missing data point of view and provide a novel solution that is more robust and computationally efficient using modern semiparametric efficiency theory.

\section{Conformal prediction: literature review}\label{sec:literature-review}

There are different forms of guarantees that one might consider for the validity of a prediction set. A test point $(X, Y)$ is covered by a set $C$ if and only if $\mathbbm{1}\{Y\in C(X)\} = 1$. Often in practice, the set $C$ is constructed based on a training data and the test point $(X, Y)$ is independent of $C$. If the prediction set denoted as $\widehat{C}$ is computed from a training data and the test point $(X, Y)$ is drawn from a distribution $P$, then the (prediction) miscoverage loss with respect to $P$ is given by
\[
L_P(\widehat{C}) := \mathbb{E}_{(X, Y)\sim P}[\mathbbm{1}\{Y\notin \widehat{C}(X)\}\big|\widehat{C}] = \int \mathbbm{1}\{y\notin\widehat{C}(x)\}dP(x, y).
\]
Observe that $L_P(\widehat{C})$ is a random variable if $\widehat{C}$ is a random set, which often it is. Further, note that there is no requirement that the training data used to construct $\widehat{C}$ comes from the same distribution $P$ as the test data $(X, Y)$ in defining the loss $L_P(\widehat{C})$.

By a $(1-\alpha)$-prediction set, one might expect/ask for $L_P(\widehat{C}) \le \alpha$. But this is too much to expect/ask in general without very restrictive assumptions on $P$ or $\widehat{C}$. More actionable goals are to require $L_P(\widehat{C})$ to be less than $\alpha$ either in expectation or with some specific probability with respect to the randomness of $\widehat{C}$. We refer to the goal of $\mathbb{E}[L_P(\widehat{C})] \le \alpha$ as {\bf joint coverage}. This is the problem solved by conformal prediction (as introduced by~\cite{vovk2005algorithmic}) and many of its variants in the literature. The goal of $\mathbb{P}(L_P(\widehat{C}) \le \alpha) \ge 1 - \delta$ is referred to as {\bf$(\alpha, \delta)$ probably approximately correct (PAC)} in the literature, which belongs to the class of classical tolerance regions --- sets that cover a pre-specified fraction of the population distribution, see e.g. \cite{guttman1967statistical} and  \cite{krishnamoorthy2009statistical}. Notably, the PAC prediction set depends on an additional parameter $\delta$ which is not required by the conformal approach. We note that, in general, neither will imply the other. \cite{vovk2012conditional} shows that the split conformal prediction sets with $(1-\alpha)$ joint coverage also satisfy $(\alpha + f_N(\delta), \delta)$-PAC guarantee, for an explicitly computable $f_N(\delta)$.

\subsection{Under Exchangeability}\label{sec:exchangeable}
We will now provide a brief introduction to the literature on conformal prediction with some emphasis on split conformal prediction that will be important to understanding the current work. \cite{vovk2005algorithmic} first introduced a version of conformal prediction called the transductive conformal method (later described in~\cite{lei2013distribution} as the full conformal method) that requires fitting the learning algorithm to samples $(X_i,Y_i), 1\le i\le N$ and $X_{N+1} = x$ for all $x \in \mathcal{X}$. While this method makes full use of the data for prediction, it is computationally intensive in practice. \cite{papadopoulos2002inductive} proposed an alternative method called the inductive conformal method (or the split conformal method in~\cite{lei2014distribution}) which splits the data into two different parts, and the learning algorithm is trained only on the first part and the prediction set is constructed using conformal ``scores'' on the second part of the data. For concreteness, we describe the split conformal algorithm here in the regression setting with $(X_i, Y_i)\in\mathcal{X}\times\mathbb{R}$. Let $\mathcal{A}:\mathcal{X}\to\mathbb{R}$ be any prediction algorithm trained on the first split of the data, i.e., for any $x\in\mathcal{X}$, $\mathcal{A}(x)$ is the point prediction for $Y$. If $m$ is the number of observations in the second split, $R_i = |Y_i - \mathcal{A}(X_i)|, i=1,\dots, m$ are residuals computed on the second split of the data, and $\widehat{Q}_{\alpha}$ is the $\lceil(m+1)(1-\alpha)\rceil$-th largest element of $R_i$'s, then for any $(X_{N+1}, Y_{N+1})$ that is exchangeable with $(X_i, Y_i), i=1, \dots, N$, it holds that $\mathbb{P}(|Y_{N+1} - \mathcal{A}(X_{N+1})| \le \widehat{Q}_{\alpha}) \ge 1 - \alpha.$ In particular, for $\widehat{C}_{\alpha}(x) = \{ y:\,|y - \mathcal{A}(x)| \le \widehat{Q}_{\alpha}\}$, we have $\mathbb{P}( Y_{N+1} \in \widehat{C}_{\alpha} (X_{N+1} ) ) \ge 1 - \alpha$. {\color{black}Both full and split conformal algorithms require only the assumption of exchangeable data for coverage validity.} 
To overcome potential statistical inefficiency due to data splitting, several papers including \cite{barber2021predictive}, \cite{romano2020classification} and \cite{kim2020predictive} introduced aggregation techniques such as jackknife+, CV+, bootstrap after jackknife+  to make better use of the data. All aforementioned prediction set constructions are valid under exchangeability of the training data as well as the test point $(X, Y)$ to be covered, as discussed in \cite{kuchibhotla2020exchangeability}. Also, see~\cite{solari2021multi}.


The above works consider joint coverage as the criterion for valid prediction set and focus on exchangeability for coverage validity.
There also is a separate line of research including, but not limited to, \cite{gyofi2020nearest} and \cite{yang2021finite} on prediction sets based on i.i.d. data assumption and concentration inequalities. The advantage is that one attains coverage guarantees conditional on the training data used to construct the set, which may be more informative and automatically implies unconditional (PAC) coverage. To elaborate on this line of thought and as an initial glimpse into the current work, consider the split conformal prediction procedure described above. In that description, under i.i.d. data assumption, the residuals computed on the second split of the data $R_i = |Y_i - \mathcal{A}(X_i)|, i=1,\dots,m$ are i.i.d. along with $R = |Y - \mathcal{A}(X)|$, irrespective of what $\mathcal{A}(\cdot)$ is. We aim to find $\widehat{r}_{\alpha}$ such that $\mathbb{P}(|Y - \mathcal{A}(X)| \le \widehat{r}_{\alpha}|\mathcal{A}) \ge 1 - \alpha$. The concentration inequalities approach to prediction sets can now proceed as follows: the well-known DKW inequality (see e.g. ~\cite{massart1990tight}) implies that
    \[
    \E \biggl[ \sup_{t} \Bigl|\frac{1}{m}\sum_{i=1}^m \mathbbm{1}\{|Y_i - \mathcal{A}(X_i)| \le t\} - \PP \bigl(|Y - \mathcal{A}(X)| \le t|\mathcal{A} \bigr) \Bigr| |\mathcal{A}\biggr] \le \sqrt{\frac{2}{m}}.
    \]
     Hence, an $\widehat{r}_{\alpha}$ at which the empirical CDF of the residuals is at least $(1 - \alpha) + \sqrt{2/m}$ yields a valid prediction set of joint coverage of at least $1-\alpha$. This also implies that if we just use the sample $(1-\alpha)-$quantile of the residuals, then this would give a prediction set that has an approximate coverage of $1-\alpha$ with a slack at most $\sqrt{2/m}.$ Moreover, the application of DKW-type inequalities here also readily yields a PAC guarantee. The well-developed theory of efficient estimation in semiparametric theory shows that the sample quantile is an optimal estimator of the population quantile~\citep{pfanzagl1976investigating}. This optimality theory is leveraged throughout the manuscript in deriving well-calibrated prediction sets under covariate shift.

\subsection{Under Non-exchangeability}

Several authors have considered the problem of prediction sets for non-exchangeable data. The literature can be divided into three parts: (1) works that consider independent data but with potential non-identical distributions, (2) works that consider prediction set construction for dependent data including time series, and (3) works that are agnostic to the randomness structure in the data. Our current work belongs to the first category. In the following, we only review the works directly related to our work and leave the review of the latter two categories to Section \ref{sec:non-exchangeable} of the apppendix.

\cite{tibshirani2019conformal} introduced the problem of prediction set construction under covariate shift. Building on \cite{tibshirani2019conformal}, \cite{lei2020conformal} worked in a covariate shift setting, establishing asymptotic joint coverage under certain conditions, accounting for estimation of the covariate density ratio.
\cite{cauchois2020robust} provided a method that produces approximate valid prediction set for any test distribution in an $f$-divergence ball around the training population and also discussed how to estimate the expected data shift and build robustness to it.
\cite{lei2020conformal} and~\cite{kivaranovic2020conformal} constructed prediction sets for counterfactuals and individual treatment effects (ITE). The former work uses the classical SUTVA assumption and rewrites the problem in terms of covariate shift assumption. The latter work assumes covariates are independent of treatment assignment and therefore implicitly rules out any covariate shift. These works represent the first connection of conformal prediction to causal inference. 

In this paper, we will focus on the specific form of non-exchangeability called \emph{covariate shift} considered in~\cite{tibshirani2019conformal} and~\cite{lei2020conformal} where the training data $\mathcal{D}^{\mbox{tr}}$ is composed of two parts $\mathcal{D}^{\mbox{tr}}_P$ and $\mathcal{D}^{\mbox{tr}}_Q$, and  
random variables $(X_i,Y_i)$'s in $\mathcal{D}^{\mbox{tr}}_P$ are  i.i.d. from $P_X\otimes P_{Y|X}$, while random variables $X_i$'s in $\mathcal{D}^{\mbox{tr}}_Q$ are i.i.d. from $Q_X$ with missing response. The goal is to obtain a prediction set such that the probability a new data pair $(X_f,Y_f)$ falls into this prediction set is larger than some nominal level $1-\alpha$ where $(X_f, Y_f)$ is from $Q_X\otimes P_{Y|X}$. As noted, all prior works focused primarily on achieving this goal asymptotically with the exception of~\cite{tibshirani2019conformal} who assumed the covariate shift is known (i.e., known $dQ_X/dP_X$), in which case they achieved finite sample coverage guarantee. We prove in this paper that it is impossible to construct a \emph{non-trivial} finite sample valid prediction set without complete a priori knowledge of either the covariate shift or the conditional distribution of $Y_f$ given $X_f$; see Theorem~\ref{thm:impossibility-finite}. We thus resort to the goal of constructing a prediction set with an asymptotic coverage of at least $1-\alpha$.

Our approach to construction of prediction sets under covariate shift requires estimation of the quantile of a univariate function of $(X_f, Y_f)$ akin to the conformal score. This task involves nuisance parameters that must be estimated sufficiently well to ensure the asymptotic coverage guarantee of our prediction sets. Fortunately, as we establish using modern semiparametric theory, this task can be accomplished in a robust fashion by using the efficient influence function as an estimating equation for the quantile. As we show, the efficient influence function is endowed with a double robustness property (see e.g. \cite{scharfstein1999adjusting}, \cite{robins2000robust} and \cite{bang2005doubly}) which ensures that the coverage bias of our prediction sets can be made negligible even if the nuisance parameters are estimated at nonparametric rates using highly adaptive machine learning algorithms.   

Before completing the review of the relevant literature, we mention a concurrent work~\cite{qiu2022distribution} that uses similar connection to semiparametric statistics as ours does but targets a different type of asymptotic PAC guarantee than we are able to provide in Theorem~\ref{thm:coverage-final}, a detailed discussion comparing the two goals and corresponding methods is given at the end of Section~\ref{subsec:covariate-shift-problem}.

\paragraph{Organization} The remainder of the paper is organized as follows. In Section \ref{sec:problem-notation}, we formally introduce the explainable covariate shift problem, providing the missing data formulations of the problem that gives notation used throughout the paper and also makes its connection to causal inference. 
In Section \ref{sec:impossibility-finite}, we formally establish that it is impossible to construct a well-calibrated prediction set that is informative without a priori knowledge about the covariate shift such as knowledge of the ratio of covariate densities, in the sense that any valid prediction set would with high probability have infinite Lebesgue measure.
  In Section \ref{sec:method-split}, we provide our first doubly robust algorithm for the explainable covariate shift problem, using a sample splitting strategy which attains nominal asymptotic prediction coverage.
  In Section \ref{sec:method-full}, we provide a second doubly robust algorithm that makes more efficient use of the observed data by avoiding sample splitting, with  guaranteed validity under certain regularity conditions we establish. Section \ref{sec:simulations} reports simulation studies (both synthetic and real data) validating the theoretical results of the proposed doubly robust methods, and comparing them to the weighted conformal prediction method of \cite{tibshirani2019conformal}. In Section \ref{sec:efcp}, we consider the problem of aggregating a collection of prediction sets by providing an explicit algorithm adapted from  \cite{yang2021finite} in an effort to optimize prediction accuracy; we illustrate the efficiency of the proposed algorithm through simulations. 
  In Section \ref{sec:sensitivity}, we relax the explainable covariate shift assumption by allowing for the presence of latent covariate shift encoded in a sensitivity parameter and discuss the efficient influence function. This is the first step in extending the framework of the current paper to account for a departure from explainable covariate shift as well as constructing prediction sets for ITE under unmeasured confounding. Further elaboration of Section~\ref{sec:sensitivity} will appear elsewhere.
  Finally in Section \ref{sec:discussion}, we conclude the paper with a brief discussion. 

Proofs of all results and supporting lemmas are provided in the supplementary where for convenience, the sections and equations are prefixed with ``S.'' and ``E.'', respectively.

\section{Our Problem and Notation}\label{sec:problem-notation}
In this section, we provide a formal description of the \emph{explainable} covariate shift problem, and introduce both missing data and counterfactual formulations of the problem. 

\subsection{The Covariate Shift Problem}\label{subsec:covariate-shift-problem}
The common assumption that the training and test data follow a common probability distribution can fail in practice, for example training data may be collected under stringent laboratory conditions that cannot be met when deployed in clinical practice; likewise, image training data may be obtained in one region whereas the test data may be collected in another. Therefore, it is important to consider situations where training and test distributions are different, also known as \textit{covariate shift}, which we formalize below. Assuming we have training data $\mathcal{D}^{\mbox{tr}}$ composed of two parts $\mathcal{D}^{\mbox{tr}}_P$ and $\mathcal{D}^{\mbox{tr}}_Q$, where
\begin{equation}\label{eq:training-data}
\mathcal{D}^{\mbox{tr}}_P ~:=~ \{Z_i = (X_i, Y_i):\, 1\le i\le n\}\quad\text{ and }\quad \mathcal{D}^{\mbox{tr}}_Q ~:=~ \{Z_i = X_i:\, n+1 \le i\le N\},
\end{equation}
and random variables in $\mathcal{D}^{\mbox{tr}}_P$ are  i.i.d. from $P_X\otimes P_{Y|X}$, while random variables in $\mathcal{D}^{\mbox{tr}}_Q$ are i.i.d. from $Q_X$.
Note that in $\mathcal{D}^{\mbox{tr}}_Q$ only data on covariates are available and the outcome/response is missing, and thus the subject of prediction. In such setting, a covariate shift problem is said to be present as the covariates in $\mathcal{D}^{\mbox{tr}}_Q$ are sampled from $Q_X$ which may be different from $P_X$ (the distribution of covariates in $\mathcal{D}^{\mbox{tr}}_P$). This setting readily extends to the more general case where one also observes samples from $Q_X \otimes P_{Y|X}$, and/or samples from $Q_X \otimes Q_{Y|X}$; settings that are closely related to \textit{transfer learning} (see e.g. \cite{kpotufe2018marginal} and \cite{reeve2021adaptive}) and in the special case when $P_X$ and $Q_X$ are identical, this reduces exactly to the setting of \textit{semi-supervised learning} (see e.g. \cite{zhu2009introduction} and \cite{zhang2019semi}). Though several variants of the covariate shift problem have been of interest in statistics and ML literatures, the specific setup considered  has recently generated renewed interest in ML literature. While the focus has  usually been on approaches to account for covariate shift while conducting model selection such as regression, see e.g. \cite{sugiyama2007covariate}, \cite{quinonero2008dataset}, \cite{bickel2009discriminative}, \cite{reddi2015doubly}, \cite{chen2016robust}, the goal here is to obtain a prediction set such that
\begin{equation}\label{eq:goal}
\mathbb{P}\bigl(Y_f \in \widehat{C}_{N,\alpha}(X_f)\bigr) \ge 1 - \alpha,\quad\mbox{whenever}\quad (X_f, Y_f)\sim Q_X\otimes P_{Y|X}.
\end{equation}
This problem was first posed in \cite{tibshirani2019conformal} where they developed a weighted version of conformal prediction that produces valid prediction sets when the likelihood ratio between training and test distributions is known. The idea of their construction is similar to importance sampling Monte Carlo. If, as in most practical settings, the likelihood ratio between the two distributions is unknown and therefore must be estimated, they empirically demonstrate in a low-dimensional setting that approximate coverage might still be possible via simulation studies. However, \cite{tibshirani2019conformal} do not formally consider the extent to which bias in estimating the likelihood ratio propagates to impact coverage. {\color{black}In addition, as noted by the authors of that paper, for every new test point $(x_f,y_f)$, their prediction set that involves a weighted quantile has to be recalculated, and this would be computationally intensive.}  
Specifically, their approach requires that the
test point $x_f$  is specified in advance. 
  
In this work we will construct prediction regions $\widehat{C}_N$ that are determined by one-dimensional functions of $(X_f, Y_f)$, i.e., we take an arbitrary function $(x,y)\mapsto R(x, y)\in\mathbb{R}$ and estimate the quantile of $R(X_f, Y_f)$. For now, we will think of $R(\cdot, \cdot)$ as a fixed non-stochastic function. In practice, $R(\cdot, \cdot)$ is a conformal score computed from an independent sample; examples include $R(x, y) = |y - \mathcal{A}(x)|$~\citep[regression residual,][]{lei2018distribution} or $R(x, y) = \max\{\widehat{q}_{\alpha/2}(x) - y, y - \widehat{q}_{1-\alpha/2}(x)\}$~\citep[conformalized quantile residual with estimated conditional quantiles $\widehat{q}_{\alpha/2}(\cdot)$ and $\widehat{q}_{1-\alpha/2}(\cdot)$,][]{romano2019conformalized}. If $r_{\alpha}$ is the smallest $(1 - \alpha)$-quantile of $R(X_f, Y_f)$ in the target population $Q_X\otimes P_{Y|X}$, then
\[
\mathbb{P}\bigl( R(X_f, Y_f) \le r_{\alpha} \bigr) \ge 1 - \alpha,
\]
and hence, for $C_{\alpha} = \{(x, y):\, R(x, y) \le r_{\alpha}\}$, we have $\mathbb{P}\bigl( (X_f, Y_f) \in C_{\alpha} \bigr) \ge 1 - \alpha$. This result holds, irrespective of the choice of function $R(\cdot, \cdot)$. Note that the assumption that $Y_i$ conditional on $X_i$ for $1\le i\le n$ has the same distribution as $Y_f$ conditional on $X_f$ essentially means that the covariate shift problem can be completely accounted for by conditioning on observed covariates $X$, hence the reference to this setting as  ``explainable covariate shift" problem. It implies that the conditional distribution of $R(X_i, Y_i)$ given $X_i$ coincides with the conditional distribution of $R(X_f, Y_f)$ given $X_f$, i.e.,
\begin{equation}\label{eq:conditional-equal-Residuals}
\mathbb{P} \bigl(R(X_i, Y_i)\in B\big|X_i = x \bigr) = \mathbb{P} \bigl(R(X_f, Y_f) \in B\big|X_f = x \bigr), 1 \le i \le n
\end{equation}
for all Borel sets $B\subseteq\mathbb{R}$.
We borrow results from semiparametric theory and estimate $r_\alpha$ based on its efficient influence function, which is intimately related to efficient influence function of the average treatment effect among the treated (ATT) functional; we then combine this influence function with arbitrary training map $R(\cdot, \cdot)$ and establish that the resulting prediction set  has asymptotic nominal coverage with coverage bias of a product form which implies correct coverage if either the likelihood ratio between the two distributions or the conditional distribution of $R$ given $X$ can be estimated sufficiently well, also known as double robustness. 
 Note that in our construction of the prediction set $\widehat{C}_{N,\alpha}$, the mapping $R$ does not depend on the test point $x_f$ at which prediction is needed. 
\cite{tibshirani2019conformal} approached~\eqref{eq:goal} using the equation
\begin{equation}\label{eq:Weighted-conformal-equation}
\mathbb{P}_{(X_f, Y_f)\sim Q_X\otimes P_{Y|X}}(R(X_f, Y_f) \le \theta) = \mathbb{E}_{(X, Y)\sim P_X\otimes P_{Y|X}}\left[\mathbbm{1}\{R(X, Y) \le \theta\}\frac{dQ_X}{dP_X}(X)\right].
\end{equation}
This implies that consistent estimation of $dQ_X/dP_X$ allows for consistent estimation of $r_{\alpha}$. But, given~\eqref{eq:conditional-equal-Residuals}, consistent estimation of the conditional distribution of $R(X, Y)$ given $X$ also allows for consistent estimation of $r_{\alpha}$. This gives a hint at double robustness in estimating $r_{\alpha}$ which we will formalize later in the paper.

In \cite{lei2020conformal}, the authors proposed a method that targets the covariate shift problem under the framework of counterfactual prediction in a causal inference setting, which we consider in Section \ref{sec:notation-causal}. Interestingly, the coverage of their prediction set has bias of the order of the minimum of two errors, that of the prediction ML algorithm and that of the estimated covariate likelihood ratio, a property that appears to hold under the so-called conformal quantile regression (CQR) which restricts the choice of conformal score to a quantile regression function for the outcome in view; it is unclear whether similar robustness extends beyond CQR.  In contrast, while our approach equally applies to the counterfactual prediction framework considered by \cite{lei2020conformal}, as we establish, its coverage bias is guaranteed to be of the order of the product of two errors, that of an estimated quantile function for $R$ with that of the covariate likelihood ratio, an immediate consequence of double robustness. Therefore, the bias of our coverage error rate can be substantially smaller relative to that of \cite{lei2020conformal}. In addition, the product bias property of the proposed method is guaranteed to hold for any ML technique used to empirically construct $R$, therefore making our approach potentially more general than theirs. 
\cite{park2021pac} uses probably approximately correct (PAC) prediction sets (tolerance regions that cover a pre-specified fraction of the population distribution) for deep learning models including in  the presence of covariate shift that also requires prior knowledge on the shift of the distributions. Notably, the PAC prediction set depends on an additional parameter $\delta$ which is not required by our approach.


Concurrently to our paper, \cite{qiu2022distribution} consider a related prediction setting, also drawing from modern semiparametric theory, but mainly focusing on a particular form of asymptotic PAC guarantee with covariate shift. We view these contributions as complementary. Specifically, in our covariate shift setting, asymptotic joint coverage is interpreted as: $\mathbb{P}(Y_f\in\widehat{C}_{N,\alpha}(X_f)) \ge 1 - \alpha - o(1)$ as $N\to\infty$. In contrast, the $(\alpha, \delta)$-PAC guarantee can be met approximately with negligible errors in either $\alpha$, $\delta$ or both. The target of~\cite{qiu2022distribution} is to develop $(\alpha, \delta + o(1))$-PAC prediction set. The methodology developed in the current paper provides $(\alpha+o_p(1), \delta)$-PAC guarantee along with the joint coverage guarantee. Whichever version of asymptotic PAC guarantee is more useful depends on the application. Notably, in order to minimize the impact of bias due to nuisance parameter estimation, \cite{qiu2022distribution} leverage the product bias property of the efficient influence function for the coverage probability of their prediction set, while we leverage the product bias structure of the efficient influence function for the $(1-\alpha)$-quantile of the statistic generating the prediction set; these two influence functions are equal up to a multiplicative constant for a given coverage guarantee, although the difference in their use leads to a different guarantee. Importantly, their PAC approach involves the construction of a valid confidence interval for the coverage probability which in turn requires a regular asymptotic linear estimator of coverage, in which case, their product bias must be of order smaller than root-$n$. Our proposed approach does not have this requirement, and therefore our guarantee is attainable even if the product bias is of order larger than root-$n$ provided that at least one of our nuisance functions is consistent.    

\subsection{Reformulation as a Missing Data Problem} \label{sec:notation-covariate}
In this section, we formulate the covariate shift problem in a missing data framework, and introduce notation used throughout the remainder of the paper. This reformulation allows us to make use of the modern theory of semiparametric statistics. Recall the training data~\eqref{eq:training-data}. 
For each $(X_i, Y_i)$ contained in $\mathcal{D}^{\mbox{tr}}_{P}$, define $R_i = R(X_i, Y_i)$ and set $T_i = 0$. For each $(X_i, Y_i)$ contained in $\mathcal{D}_Q^{\mbox{tr}}$, $R_i$ is unobserved because the corresponding $Y_i$ is unobserved. Hence the observed data $\mathcal{D} = \mathcal{D}_P^{\mbox{tr}}\cup\mathcal{D}_Q^{\mbox{tr}}$ can be succinctly written as $Z_i = \bigl(X_i, T_i, (1 - T_i)R_i \bigr), 1\le i\le N$ such that
\begin{equation}\label{eq:covariate-shift}
\mathbb{P}(X_i\in A|T_i = 0) =: P_X(A)\quad\mbox{and}\quad \mathbb{P}(X_i\in A|T_i = 1) = :Q_X(A),
\end{equation}
while
\begin{equation}\label{eq:conditional-dist-same}
\mathbb{P}(R_i \in B|T_i = 0, X_i = x) = \mathbb{P}(R_i \in B|T_i = 1, X_i = x) =: P_{Y|X=x}(B)\quad\mbox{a.e.}\quad x.
\end{equation}
Equations in~\eqref{eq:covariate-shift} signify that the covariates in $\mathcal{D}_P^{\mbox{tr}}$ are distributed as $P_X$ and that the covariates in $\mathcal{D}_{Q}^{\mbox{tr}}$ are distributed as $Q_X$. Condition~\eqref{eq:conditional-dist-same}, on the other hand, signifies that the conditional distribution of the response $Y$ given $X$ is the same for $\mathcal{D}_P^{\mbox{tr}}$, $\mathcal{D}_Q^{\mbox{tr}}$, and the future data $(X_f, Y_f)$. Condition~\eqref{eq:conditional-dist-same} restates condition~\eqref{eq:conditional-equal-Residuals} in terms of $T_i$. 

Condition
~\eqref{eq:conditional-dist-same} implies that $R_i$ is independent of $T_i$ conditional on $X_i$ for all $1\le i\le N$; this is denoted by $R_i\perp T_i|X_i$ and is equivalent to the missing at random (MAR) assumption in missing data literature.  This assumption, which is not testable without an additional condition, essentially states that there is no unmeasured factor that is related with both $R$ and $T$. 
For identification, we further assume that for any Borel set $B$,
\begin{equation}\label{eq:measure-dominating}
    P_X(B)=0 \text{ implies } Q_X(B)=0.
\end{equation}
This is the same as assuming the measure $Q_X$ is absolutely continuous w.r.t $P_X$. In other words, the support of $X|T=1$ is contained in the support of $X|T=0$. Assumption~\eqref{eq:measure-dominating} is needed for~\eqref{eq:Weighted-conformal-equation}, which is a crucial component of the double robustness of our methodology. 

{\color{black}Summarizing the above discussion, the covariate shift assumption which states that only the covariate distributions between $\mathcal{D}_P^{\mbox{tr}}$ and $\mathcal{D}_{Q}^{\mbox{tr}}$ can be different but not the conditional distributions of the response given the covariates is equivalent to the MAR assumption.}

For any $\theta\in\mathbb{R}$, define sets
\begin{equation}
\widehat{C}(\theta; x):= \bigl\{ y: R(x,y) \leq \theta \bigr\}.   
\end{equation}
In this notation, 
under condition \eqref{eq:measure-dominating}, we aim to find a data-dependent random variable $\widehat{r}_{\alpha}$ such that for any $(X_f, Y_f)\sim Q_X\otimes P_{Y|X}$,
\begin{equation}\label{eq:desired-coverage}
\begin{aligned}
\mathbb{P}\bigl(  Y_f \in \widehat{C}(\widehat{r}_{\alpha} ; X_f) \bigr)  = \PP_{(X, Y) \sim Q_{X} \times P_{Y \mid X}}( R(X,Y) \leq \widehat{r}_{\alpha} )  = \PP ( R(X,Y) \leq \widehat{r}_{\alpha}|T = 1 ) \ge 1-\alpha,
\end{aligned}
\end{equation}
while $r_\alpha$, the ``target'' of $\widehat{r}_{\alpha}$ is defined to be the smallest real number such that
\begin{equation}\label{eq:definition-quantile}
\PP \bigl( R(X,Y) \leq r_{\alpha} |T = 1 \bigr) \ge  1-\alpha.
\end{equation}
Because we do not observe random variables $R$ when $T=1$, this goal is not achievable in finite samples without restrictive assumptions such as a known $ \PP(T=1|X)/ \PP(T=0|X)$~\citep{tibshirani2019conformal}.
We provide a random variable $\widehat{r}_{\alpha}$ so that~\eqref{eq:desired-coverage} is achieved with a slack that converges to zero as $N$ tends to $\infty.$

\section{Impossibility of finite sample coverage}\label{sec:impossibility-finite}
Recall our aim from~\eqref{eq:desired-coverage}. Resorting to semiparametric theory, in most cases, implies that the resulting coverage guarantee is only asymptotic. In our problem of covariate shift, one can prove that it is impossible to construct a finite sample valid \emph{non-trivial} prediction set without the knowledge of either the covariate shift or the conditional distribution of $Y$ given $X$. Here, by a non-trivial prediction set, we mean a set with a finite Lebesgue measure. Lemma S1 of Section S6.1 of~\cite{qiu2022distribution} prove an analogous but weaker result for PAC guarantee as they establish a result similar to \eqref{eq:finite-marginal-set-infinite} of the following Theorem, however \eqref{eq:infinite-new} appears to be an entirely novel contribution.

\begin{thm}\label{thm:impossibility-finite}
Suppose the observed data consists of $n$ i.i.d. tuples $(X_i, T_i, (1 - T_i)Y_i)$. Further assume that $\mathcal{X}\subseteq\mathbb{R}^d$ is the support of $X_i$ and $\mathcal{Y}\subseteq\mathbb{R}$ is the support of $Y_i$. Let $\bar{\mathcal{P}}^0$ be the set of all distributions $\bar{P}^0$ on the random vector $\bar{O} = (X,T,Y)$ such that $T$ is independent of $Y$ given $X$ ($T\perp Y|X$), and the joint distribution of $(X, Y)$ is absolutely continuous with respect to the Lebesgue measure on $\mathcal{X} \times \mathcal{Y}$. 

Suppose that a (possibly randomized) prediction set $\hat{C}_{\alpha}$ has finite-sample joint coverage guarantee in the target population, that is,
\begin{equation}\label{eq:finite-marginal-set}
\sup_{\bar{P}^0\in\bar{\mathcal{P}}^0}\,\PP_{\bar{P}^0}(Y \notin \hat{C}_{\alpha}(X) \mid T=1 ) \leq \alpha,\quad\mbox{for some}\quad \alpha\in(0, 1).
\end{equation}
Then, for any $\bar{P}^0 \in \bar{\mathcal{P}}^0$ and a.e. $y \in \mathcal{Y}$ with respect to the Lebesgue measure,
\begin{equation}\label{eq:finite-marginal-set-infinite}
\PP_{\bar{P}^0} \bigl( y \notin \hat{C}_{\alpha}(X) \bigr) \leq \alpha .
\end{equation}
Furthermore, $\hat{C}_{\alpha}(X)$ would at least cover one of the end points $\mathcal{Y}$ with probability at least $1-\alpha$, and hence if $\mathcal{Y} = \mathbb{R}$, then 
\begin{equation}\label{eq:infinite-new}
\mathbb{E}_{\bar{P}^0}[\mbox{Leb}(\widehat{C}_{\alpha}(X))] = \infty.
\end{equation}
\end{thm} 
We defer the proof to the Section \ref{sec:impossibility-appendix} of the appendix. The proof is based on the lack of a non-trivial test for the problem of conditional independence hypothesis testing proved in~\cite{shah2020hardness}. The connection to conditional independence testing can be seen from the fact that any prediction set that is valid under the conditional independence of $T$ and $Y$ given $X$ can be rewritten as a valid test for the hypothesis that $T$ is conditionally independent of $Y$ given $X$; see the proof in Section~\ref{sec:impossibility-appendix} for more details.

Given the lack of finite-sample valid non-trivial prediction, we resort to finding an efficient asymptotically valid prediction set based on semi-parametric theory in the following sections.
\section{Methodology with split data}\label{sec:method-split}
In this section, we discuss our novel prediction set construction under the covariate shift setting using semiparametric theory. 

Recall our notation $Z = (X, T, (1 - T)R)$. Suppose that one is interested in the $q$th-quantile of a random variable $R|T=1$, denoted $\theta_0 := \inf\{r: F(r) \geq q\}$, where $F$ is the CDF of $R |T=1$.  
An estimator 
$\widehat{\theta}$ is said to be asymptotically linear 
if it satisfies
$$
\sqrt{N}\bigl(\widehat{\theta}-\theta_0 \bigr)=\frac{1}{\sqrt{N}} \sum_{i=1}^{N} \psi (Z_{i} )+o_{p}(1),\quad \E [\psi(Z)]=0, \E \bigl[\psi(Z)^{\top} \psi(Z) \bigr]<\infty,
$$
as $N\to\infty$.
Clearly, the asymptotic variance of $\widehat{\theta}$ is then $\E \bigl[\psi(Z) \psi(Z)^{\top} \bigr]/N$. The function $\psi(z)$ is referred to as the influence function, following terminology of \cite{hampel1974influence}. Furthermore, as the model is nonparametric in the sense that the observed data distribution is not restricted and all such distributions are regular, (see Chapter 2 of \cite{ bickel1993efficient} on regularity), any estimator satisfying the above expansion is said to attain the semiparametric efficiency bound, and $ \psi (\cdot)$ is said to be the efficient influence function of $\theta_{0}$ in the nonparametric model. For more on influence functions and semiparametric theory, see e.g. \cite{newey1990semiparametric}, Chapter 25 of \cite{van2000asymptotic} and~\cite{van2002semiparametric}.

We now state the efficient influence function for the $(1-\alpha)$-th quantile of $R|T = 1$ under the MAR assumption. For every $x\in\rchi$ and $r\in\mathbb{R}$, define
\begin{equation}\label{eq:true-functions}
\begin{split}
\pi^{\star}(x) ~&:=~ \mathbb{P}(T = 1|X = x)/\mathbb{P}(T = 0|X = x),\\
m^{\star}(r, x) ~&:=~ \mathbb{E}[\mathbbm{1}\{R \le r\}|X = x].
\end{split}
\end{equation}
The function $\pi^{\star}(\cdot)$ represents the true density ratio of the covariates among labeled and unlabeled data. The function $m^{\star}(\cdot, \cdot)$ represents the true conditional mean function that by assumption~\eqref{eq:conditional-equal-Residuals} is common for both labeled and unlabeled data.

{\color{black}We now state the efficient influence function for the quantile of interest under regularity conditions for the data distribution that will motivate our proposed method, where it should be noted that strictly speaking, the regularity conditions are actually not needed for the theoretical guarantees of our proposed methods. }
\begin{lem}\label{lem:if}
Suppose $\E [\mathbbm{1}\{T=0\} \pi^{\star2} (X)] = \E {\PP^2(T=1|X )}/{\PP(T=0|X)}$ is finite and that the density of the conditional distribution of $R|T=1$ at $r_\alpha$ is bounded away from zero. Then the efficient influence function of the $(1-\alpha)$-quantile of $R|T=1$ in the nonparametric model for $Z$ which allows the distribution of $Z$ to remain unrestricted under the condition that it is regular, is given up to a proportionality constant by
\begin{equation}\label{eq:if}
\begin{aligned}
\psi(z) = \mathrm{IF}(r_\alpha, x, r, t; \pi^\star, m^\star) &= \mathbbm{1}\{t = 0\}\pi^\star(x)\Big[\mathbbm{1}\{r \le r_\alpha\} - m^\star(r_\alpha, x)\Big] \\ &\quad+ \mathbbm{1}\{t = 1\}\Big[m^\star(r_\alpha, x) - (1 - \alpha)\Big]. 
\end{aligned}
\end{equation}
\end{lem}
\begin{proof}
 \cite{hahn1998role}
gave a derivation of the influence function for the average treatment effect among the treated (ATT). We adapt their proof to that of the conditional quantile and give the complete derivation of \eqref{eq:if} along with a basic introduction to semiparametric theory needed to derive the result in Section \ref{sec:if} of the supplementary.
\end{proof}


For any two functions $\pi(\cdot)$ and $m(\cdot, \cdot)$, let 
\begin{align}
 \mathrm{IF}(\theta, x, r, t; \pi, m):= \mathbbm{1}\{t = 0\}\pi(x)\Big[\mathbbm{1}\{r \le \theta\} - m(\theta, x)\Big] + \mathbbm{1}\{t = 1\}\Big[m(\theta, x) - (1 - \alpha)\Big],
\end{align}
ignoring the scaling factor. Note that $\mathrm{IF}(\theta, x, r, t; \pi, m)$ is only a function of $(x, t, (1-t)r)$ because the term that depends on $r$ has a multiplicative factor of $\mathbbm{1}\{t = 0\}$. Let $P [f]$ denote integration conditional on the training sample. For example, for any $\theta, \pi, m$ that are potentially data-dependent, 
\[
P[\mathrm{IF}(\theta, x, r, t; \pi, m)] = \int \mathrm{IF}(\theta, x, r, t; \pi, m)dP_{R|X = x}(r|x)dP_{T|X = x}(t|x)dP_X(x).
\]
This is a random variable if $\theta$ or $\pi$ or $m$ are random.

First, we draw a key connection between the desired coverage and the aforementioned influence function.
\begin{lem}\label{lem:connection-if-coverage}
Let $\pi:\rchi\to\mathbb{R}_+$ and $m:\mathbb{R}\times\rchi\to[0,1]$ be any two functions. Then for every (potentially) data-dependent $\theta\in\mathbb{R}$, the representation 
\begin{align}
\mathbb{P}_{(X, Y)\sim Q_X\otimes P_{Y|X}}\bigl(Y \in \widehat{C}(\theta ; X) \mid\theta\bigr) ~&=~ \mathbb{P}(Y \in \widehat{C}(\theta; X)|\theta, T = 1)\\
~&=~ 1-\alpha +\frac{P[\mathrm{IF}(\theta, X,R,T;\pi,m)]}{\PP(T=1)}, \label{eq:connection-if-coverage}
\end{align}
holds true, whenever either of the following holds true:
\begin{enumerate}
    \item $\pi(x) = \pi^{\star}(x)$ for all $x$; or
    \item $m(\gamma, x) = m^{\star}(\gamma, x)$ for all $\gamma$ and $x$.
\end{enumerate}
\end{lem}
\begin{proof}
See Section~\ref{sec:proof-of-connection-if-coverage} for a proof.
\end{proof}
Note that even if data-dependent, an estimate of $\theta$ remains independent of a future observation $(X, Y)$.
Lemma~\ref{lem:connection-if-coverage} has a key implication that $\mathrm{IF}(\cdots)$ is a doubly robust influence function. Because
\[
\mathbb{P}\bigl(Y \in \widehat{C}(\theta ; X) \mid\theta, T = 1\bigr) = \mathbb{P}\bigl(R(X, Y) \le \theta \mid \theta, T = 1 \bigr),
\]
taking $\theta = r_{\alpha}$, the quantile of $R(X, Y)$ conditional on $T = 1$,\footnote{We assume here that $\mathbb{P}_{(X,Y)\sim Q_X\otimes P_{Y|X}}(R(X, Y) \le r_{\alpha}) = 1 - \alpha$, which is mild as one can add small Gaussian noise to $R(X, Y)$.} Lemma~\ref{lem:connection-if-coverage} implies that
\begin{equation}\label{eq:double-robustness}
P[\mathrm{IF}(r_{\alpha}, X, R, T; \pi, m)] = 0,\quad\mbox{if either }\quad \pi \equiv \pi^{\star}\mbox{ or }m \equiv m^{\star}.
\end{equation}
Because $r_{\alpha}$ is a constant, $P[\mathrm{IF}(r_{\alpha}, X, R, T; \pi, m)] = \mathbb{E}[\mathrm{IF}(r_{\alpha}, X, R, T; \pi, m)]$. 

\begin{remark}
Note that our results do not actually require  uniqueness of a solution to $\E \big[ \mathrm{IF}(\theta, x, r, t; \pi^\star, m^\star) \bigl]=0$. This could arise for example in settings when $R$ is discrete. In principle, our result would continue to hold for any element $\theta$ of a solution set. 
\end{remark}

Property~\eqref{eq:double-robustness} is a double robustness property in that the expectation is zero, as long as one of $\pi$ and $m$ is the true function. This property implies that when, as would generally be the case in practice, $\pi^\star$ and $m^\star$ are estimated, the resulting bias is of the following product form, where we introduce the notation $\|f(\cdot)\|_2$ as the $L_2$-norm of $f(\cdot)$, where $\|f(\cdot)\|_2:= [ \int f^2(x)M(\mathrm{d}x) ]^{1/2} $, and $M(\cdot)$ is the probability measure of $X$ that for any Borel set $B$, $M(B)$ is given by
\begin{align*}
M(B) &= \mathbb{P}(X \in B)\\ 
&= \mathbb{P}(X \in B|T = 1)\mathbb{P}(T = 1) + \mathbb{P}(X \in B|T = 0)\mathbb{P}(T = 0)\\ 
&= Q_X(B)\mathbb{P}(T = 1) + P_X(B)\mathbb{P}(T = 0).
\end{align*}

\begin{thm}\label{thm:product-bias}
For any functions $\widehat\pi(\cdot)$ and $\widehat{m}(\cdot,\cdot) $,  it holds that
\begin{align}\label{eq:product bias}
    &\sup_{\gamma \in \mathbb{R}} \Bigl| P \bigl[ \mathrm{IF}(\gamma, X, R, T; \widehat{\pi},\widehat{m}) - \mathrm{IF}(\gamma, X, R, T; {\pi}^\star,m^\star ) \bigr] 
    \Bigr|\\ 
    &\qquad\quad\leq \| \widehat\pi - \pi^\star \|_2 \sup_{\gamma} \| \widehat{m}(\gamma,\cdot) - m^\star(\gamma,\cdot)\|_2.
\end{align}
\end{thm}
\begin{proof}
See Section \ref{appsec:proof-product-bias} for a proof.
\end{proof}
Theorem~\ref{thm:product-bias} implies that $P[\mathrm{IF}(\gamma, X, R, T; \widehat{\pi}, \widehat{m})]$ converges to $P[\mathrm{IF}(\gamma, X, R, T; \pi^{\star}, m^{\star})]$ as long as one of $\pi^{\star}$ and $m^{\star}$ is estimated consistently. 

Lemma~\ref{lem:connection-if-coverage} is the main building block for our methodology. In order to ensure approximately correct coverage of $1 - \alpha$, we need to find $\theta$ such that $P[\mathrm{IF}(\theta, X, R, T; \pi, m)]$ is approximately zero, with either $\pi\equiv \pi^{\star}$ or $m \equiv m^{\star}$. In practice where we do not have access to either $\pi^{\star}$ or $m^{\star}$ and even if we know either of them, one cannot compute $P[\mathrm{IF}(\theta, X, R, T; \pi, m)]$ without access to the true distribution of $(X, (1-T)R, T)$. Our methodology, hence, is as follows. We construct estimators $\widehat{\pi}(\cdot)$ and $\widehat{m}(\cdot, \cdot)$ such that $\|\widehat{\pi} - \pi\|_2 = o_p(1)$ and $\|\widehat{m} - m\|_2 = o_p(1)$ for some $\pi(\cdot)$ and $m(\cdot, \cdot)$, and either $\pi \equiv \pi^{\star}$ or $m \equiv m^{\star}$. Then we find the smallest $\hat{\theta}$ such that
\begin{equation}\label{eq:definition-theta}
\mathbb{P}_{N} \bigl[\mathrm{IF}(\hat{\theta}, X, R, T; \widehat{\pi}, \widehat{m}) \bigr] := \frac{1}{N}\sum_{1 \le i \le N } \mathrm{IF}(\hat{\theta}, X_i, R_i, T_i; \widehat{\pi}, \widehat{m}) \ge 0.
\end{equation}
We can prove under certain regularity conditions on $\widehat{\pi}$ and $\widehat{m}$ that 
\begin{equation}\label{eq:empirical-to-expectation-IF}
\mathbb{P}_{N} \bigl[\mathrm{IF}(\hat{\theta}, X, R, T; \widehat{\pi}, \widehat{m}) \bigr] - P \bigl[\mathrm{IF}(\hat{\theta}, X, R, T; \pi, m) \bigr] = o_p(1),
\end{equation} 
even for a data-dependent $\hat{\theta}$. Then, Lemma~\ref{lem:connection-if-coverage} implies that $\widehat{\theta}$ satisfying~\eqref{eq:definition-theta} also satisfies
\begin{equation}\label{eq:expectation-desired}
\mathbb{P}_{(X,Y)\sim Q_X\otimes P_{Y|X}}\bigl( Y \in\widehat{C}(\hat{\theta} ; X) \mid \hat{\theta} \bigr) \ge (1 - \alpha)
+ o_p(1).
\end{equation} 
This yields the desired coverage guarantee~\eqref{eq:desired-coverage}. In finding $\hat{\theta}$ and proving~\eqref{eq:empirical-to-expectation-IF}, one can avoid restrictive regularity conditions (such as smoothness or Donsker class) on $\widehat{\pi}, \widehat{m}$ by splitting the data into two parts, using the first part to determine  $\widehat{\pi}, \widehat{m}$ and using the second part to compute $\mathbb{P}_{\mathcal{I}_2}[\mathrm{IF}(\cdots)]$. The detailed split sample procedure is succinctly described in Algorithm~\ref{alg:influence-conformal-split}. 
\begin{algorithm}[h]
    \SetAlgoLined
    \SetEndCharOfAlgoLine{}
   \KwIn{Training data $\mathcal{D}^{\text{tr}} = \mathcal{D}^{\text{tr}}_P \cup \mathcal{D}^{\text{tr}}_Q$; Coverage probability $1 - \alpha$, a training method $\mathcal{A}$ and estimators $\widehat{\pi}, \widehat{m}$, the point for prediction $x$.}
    \KwOut{A valid prediction set $\widehat{C}_{\alpha}(x)$.}
    Split training data $\mathcal{D}^{\text{tr}}$ randomly into $\mathcal{D}_1$ and $\mathcal{D}_2$, where $\mathcal{D}_1 = \{Z_i \in \mathcal{D}^{\text{tr}}, i \in \mathcal{I}_1\}$ and $\mathcal{D}_2 = \{Z_i \in \mathcal{D}^{\text{tr}}, i \in \mathcal{I}_2\}$.\;
    Fit the training method $\mathcal{A}$ on $\mathcal{D}_1$ and using fitted method $\mathcal{A}$, construct an increasing (nested) sequence of sets $\{\mathcal{F}_t\}_{t\in\mathcal{T}}$. Here $\mathcal{T}$ is a subset of $\mathbb{R}$. The nested sets $\{\mathcal{F}_t\}_{t\in\mathcal{T}}$ can depend arbitrarily on $\mathcal{D}_1$.\;
  For each $i\in\mathcal{I}_2$ that satisfies $T_i=0$, define the conformal score
    \[
    r_i ~=~ r(Z_i) := \inf\{t\in\mathcal{T}:\,Z_i\in\mathcal{F}_t\}.
    \]\;
    \vspace{-1em}
    Fit estimators $\widehat{\pi},\widehat{m}$ on $\mathcal{D}_1$ and find the smallest $\hat{\theta} = \widehat{r}_{\alpha}$ such that $\mathbb{P}_{\mathcal{I}_2}[\mathrm{IF}(\hat{\theta}, X, R, T; \widehat{\pi}, \widehat{m})] \ge 0$, where
    \begin{align*}
    \mathbb{P}_{\mathcal{I}_2}[\mathrm{IF}(\hat{\theta}, X, R, T; \widehat{\pi}, \widehat{m})] ~&=~ \frac{1}{|\mathcal{I}_2|}\sum_{i \in \mathcal{I}_2} \mathbbm{1}\{t_i=0\} \widehat{\pi}(x_i) \Big[ \mathbbm{1}\{r_i \leq \hat{\theta} \} - \widehat{m}(\hat{\theta},x_i) \Bigr]\\
    ~&\quad+~ \frac{1}{|\mathcal{I}_2|} \sum_{i \in \mathcal{I}_2} \mathbbm{1}\{t_i=1\} \Bigl[\widehat{m}(\hat{\theta},x_i) - (1-\alpha) \Bigr].
    \end{align*}
  
    \Return the prediction set $\widehat{C}_{\alpha}(x):= \bigl\{ y: R(x,y) \leq \widehat{r}_{\alpha} \bigr\}$.
    \caption{Split doubly robust prediction}
    \label{alg:influence-conformal-split}
\end{algorithm}

In step 2 of Algorithm~\ref{alg:influence-conformal-split}, the training method $\mathcal{A}$ can be a regression estimator of $Y$ on $X$ leading to $\widehat{\mu}$ and the nested sets, for example, could be $\mathcal{F}_t = \{(x, y):\,|y - \widehat{\mu}(x)| \le t, t \ge 0 \}$. This corresponds to using $R = R(x, y) = |y - \widehat{\mu}(x)|$. Alternatively, one can also consider a training method that leads to conditional quantile estimator $\widehat{q}_{\alpha/2}(\cdot)$ and $\widehat{q}_{1-\alpha/2}(\cdot)$. The nested sets, for example, could be $\mathcal{F}_t = \{(x, y):\, y\in[\widehat{q}_{\alpha/2}(x) - t, \widehat{q}_{1-\alpha/2}(x) + t]\}.$ This corresponds to the map $R = R(x, y) = \max\{\widehat{q}_{\alpha/2}(x) - y, y - \widehat{q}_{1-\alpha/2}(x)\}$ which is the conformal score of the conformalized quantile regression (CQR) method of~\cite{romano2019conformalized}. For discrete/categorial response $Y$, conditional probability of each class $\PP(Y = j|X = x)$ could be used to construct the nested sets and the map $R$, see for example, Section 4 of \cite{kuchibhotla2021nested}.

Now we list some assumptions that will be used in the following theorems.

\begin{enumerate}[label=\bf(A\arabic*)]
\setcounter{enumi}{0}
\item \label{assump:DGP}$(X_i, T_i, (1 - T_i)R_i), i\in\mathcal{I}_2$ are independent and identically distributed random vectors satisfying condition~\eqref{eq:conditional-dist-same}.
\item \label{assump:bounded} The functions $(\theta, x)\mapsto \widehat{m}(\theta, x)$ and $x\mapsto \widehat{\pi}(x)$ are bounded, i.e., there exist $m_0$ and $\pi_0$ such that for all $\theta \in \mathbb{R}$ and $x \in \mathbb{R}^d$, $| \widehat{m}(\theta, x) | \leq m_0$ and $| \widehat{\pi}(x) | \leq \pi_0$. 
\item \label{assump:monotone} The estimator $\widehat{m}(\theta, x)$ is a nondecreasing function of $\theta$. 
\end{enumerate}

Assumptions \ref{assump:DGP} and \ref{assump:bounded} are both standard conditions where we note that $m^\star(\cdot,\cdot)$ is a conditional CDF contained in the unit interval [0,1].
For assumption \ref{assump:monotone}, because $m^\star(\theta, x)$ is a conditional CDF which must be monotonically nondecreasing in $\theta$, any given estimator $\widetilde{m}$ can be improved upon by transforming it into a monotone estimator $\widehat{m}^*$ such that $\|\widehat{m}^* - m \| \leq \| \widetilde{m} - m\|$, see e.g. the first two properties of Proposition 2 of \cite{chernozhukov2009improving}. We state their result in Proposition \ref{prop:chernozhukov2009improving} in \ref{sec:supporting} of the supplementary for completeness. Given this proposition, it is natural that we restrict the estimator $\widehat{m}(\cdot,\cdot)$ to the class of  functions that are non-decreasing in their first argument. Hence we impose assumption \ref{assump:monotone}.

Under assumptions~\ref{assump:DGP}--\ref{assump:monotone}, we now provide a coverage guarantee for the prediction set $\widehat{C}_{\alpha}$ returned by Algorithm~\ref{alg:influence-conformal-split}. Observe that following Lemma~\ref{lem:connection-if-coverage} and the discussion surrounding~\eqref{eq:definition-theta} and~\eqref{eq:empirical-to-expectation-IF}, we obtain
\begin{equation}\label{eq:coverage-guarantee-split-IF}
    \begin{split}
    \mathbb{P}\Bigl( Y \in\widehat{C}_{\alpha} (X) \mid \mathcal{D}^{\mathrm{tr}}, T = 1\Bigr) - (1 - \alpha)
    ~&=~ \frac{\mathbb{P}_{\mathcal{I}_2}[\mathrm{IF}(\widehat{r}_{\alpha}, X, R, T; \widehat{\pi}, \widehat{m})]}{\mathbb{P}(T = 1)}\\
    ~&\quad+ \frac{ 
    P[\mathrm{IF}(\widehat{r}_{\alpha}, X, R, T; \widehat{\pi}, \widehat{m})] - \mathbb{P}_{\mathcal{I}_2}[\mathrm{IF}(\widehat{r}_{\alpha}, X, R, T; \widehat{\pi}, \widehat{m})] 
    }{\mathbb{P}(T = 1)}\\
    ~& \quad+ \frac{ 
    P[\mathrm{IF}(\widehat{r}_{\alpha}, X, R, T; \pi^{\star}, m^{\star})] - P[\mathrm{IF}(\widehat{r}_{\alpha}, X, R, T; \widehat{\pi}, \widehat{m})] 
    }{\mathbb{P}(T = 1)}\\
    ~&\ge~ 0 + \mathbf{I} + \mathbf{II}.
    \end{split}
\end{equation}
Here we use the fact that $\widehat{r}_{\alpha}$ satisfies $\mathbb{P}_{\mathcal{I}_2}[\mathrm{IF}(\widehat{r}_{\alpha}, X, R, T; \widehat{\pi}, \widehat{m})] \ge 0$. The term $\mathbf{II}$ in~\eqref{eq:coverage-guarantee-split-IF} can be bounded in absolute value using Theorem~\ref{thm:product-bias}.
We now provide a bound on $\mathbf{I}$ in Theorem~\ref{thm:convergence-if-function} below under assumptions~\ref{assump:DGP}--\ref{assump:monotone}. Theorem~\ref{thm:convergence-if-function} actually proves the tail and expectation bound for
\[
\sup_{\theta\in\mathbb{R}}|\mathbb{P}_{\mathcal{I}_2}[\mathrm{IF}(\theta, X, R, T; \widehat{\pi}, \widehat{m})] - P[\mathrm{IF}(\theta, X, R, T; \widehat{\pi}, \widehat{m})]|.
\]

\begin{thm}\label{thm:convergence-if-function}
Under assumption \ref{assump:DGP}, for any estimators $\widehat{\pi}, \widehat{m}$ satisfying assumptions \ref{assump:bounded} and \ref{assump:monotone}, there exists a universal constant $\mathfrak{C}$ such that for any $\delta>0$,
\begin{align}
    \mathbb{P}\left(|\mathbf{I}| \le \frac{\mathfrak{C}}{\mathbb{P}(T = 1)}\sqrt{\frac{ (m_0+\pi_0 +1- \alpha)^2 \log \bigl( {1}/{\delta}  \bigr) + (m_0+\pi_0)^2 }{ |\mathcal{I}_2|}}\;\bigg|\mathcal{D}_1\right) ~\ge~ 1 - \delta.
\end{align}
Moreover, there exists a universal constant $\mathfrak{C}'$ such that
\[
\mathbb{E}\left[|\mathbf{I}|\big|\mathcal{D}_1\right] ~\le~ \frac{\mathfrak{C}'}{\mathbb{P}(T = 1)}\sqrt{\frac{(m_0 + \pi_0 +1 - \alpha)^2   +(m_0+\pi_0)^2 }{|\mathcal{I}_2|}} \leq \frac{\mathfrak{C}'}{\mathbb{P}(T = 1)}\frac{m_0 + \pi_0   +1}{\sqrt{|\mathcal{I}_2|}}.
\]
\end{thm}
\begin{proof}
See Section \ref{appsec:proof-convergence-if-funcion} of the appendix.
\end{proof}
This theorem is proved using techniques from empirical process theory by bounding $\sup_\theta \bigl|\PP_N[\mathrm{IF}(\widehat{r}_\alpha,\dots)] - \PP[\mathrm{IF}(\widehat{r}_\alpha,\dots)] \bigr|$. It gives a convergence rate of $\PP_N[\mathrm{IF}(\widehat{r}_\alpha,\dots)]$ to $\PP[\mathrm{IF}(\widehat{r}_\alpha,\dots)]$ that scales as $O(N^{-1/2})$, if $|\mathcal{I}_1| \asymp |\mathcal{I}_2|\asymp N$.

Using the definition of $\widehat{r}_\alpha$ and then combining Theorems  \ref{thm:product-bias} and \ref{thm:convergence-if-function} together with~\eqref{eq:coverage-guarantee-split-IF} yields the following main result.

\begin{thm}\label{thm:coverage-final}
Under assumption \ref{assump:DGP}, 
for any estimators $\widehat{\pi}, \widehat{m}$ satisfying assumptions \ref{assump:bounded} and \ref{assump:monotone}, there exists a universal constant $\mathfrak{C}$ such that for any $\delta > 0$ with probability at least $1-\delta$,
\begin{equation}\label{eq:main-coverage}
\begin{aligned}
\mathbb{P}_{(X, Y) \sim Q_{X} \times P_{Y \mid X}} \Bigl(  Y \in \hat{C}(\widehat{r}_{\alpha} ; X) \mid \mathcal{D}^{\mathrm{tr}} \Bigr) &\geq  1-\alpha\\
&\quad-   \frac{\|\widehat{\pi} -  \pi^\star \|_2}{\mathbb{P}(T = 1)}    \sup_{\theta} \| \widehat{m}(\theta,\cdot) - m^\star(\theta,\cdot)\|_2\\ 
&\quad - \mathfrak{C} \frac{(m_0+\pi_0+1)}{\PP(T=1)} \sqrt{\frac{  \log \bigl( {1}/{\delta}  \bigr) +1 }{ |\mathcal{I}_2|}} .
\end{aligned}
\end{equation}
Moreover, 
\begin{equation}\label{eq:final-coverage-unconditional}
    \begin{split}
        \mathbb{P}_{(X, Y)\sim Q_X\otimes P_{Y|X}}\Bigl(Y \in \hat{C}(\widehat{r}_{\alpha} ; X) \Bigr) ~&\ge~ (1 - \alpha)
        - \mathbb{E}\left[\frac{\|\widehat{\pi} - \pi^{\star}\|_2}{\mathbb{P}(T = 1)}\sup_{\theta}\|\widehat{m}(\theta, \cdot) - m^{\star}(\theta, \cdot)\|_2\right]
        \\ 
        &\qquad- \frac{\mathfrak{C}}{\mathbb{P}(T = 1)}\frac{m_0 + \pi_0   +1}{\sqrt{|\mathcal{I}_2|}} .
    \end{split}
\end{equation}

\end{thm}
\begin{proof}
This is a direct result of Lemma \ref{lem:connection-if-coverage},~\eqref{eq:coverage-guarantee-split-IF}, and Theorem~\ref{thm:convergence-if-function}.
\end{proof}

Equation \eqref{eq:main-coverage} of Theorem \ref{thm:coverage-final} provides a coverage guarantee conditional on the training data and \eqref{eq:final-coverage-unconditional} provides a bound on the unconditional coverage. Note that the slack for the coverage is the sum of two terms: the product bias from the estimation of $\pi^\star$ and $m^\star$ and a term of order $O(N^{-1/2})$ if we assume the two splits $\mathcal{I}_1$ and $\mathcal{I}_2$ are of similar size.  Proposition 2a of~\cite{vovk2012conditional} establishes a conditional prediction coverage guarantee which also involves a slack analogous to the last term of~\eqref{eq:main-coverage}. 
The miscoverage error slacks in~\eqref{eq:main-coverage} and~\eqref{eq:final-coverage-unconditional} have clear meaning. The first slack (product of errors) comes from the double robustness property of $\mathrm{IF}$ and the second slack (of order $|\mathcal{I}_2|^{-1/2}$) comes from approximating $P[\mathrm{IF}(\cdots)]$ with $\mathbb{P}_{\mathcal{I}_2}[\mathrm{IF}(\cdots)]$.
\begin{remark}
Theorem~\ref{thm:coverage-final} only provides lower bounds on the (conditional or unconditional on $\mathcal{D}^{\mathrm{tr}}$) coverage probability. Without a continuity assumption on the distribution of $R(X, Y)$ when $(X, Y) \sim Q_X\otimes P_{Y|X}$, it is not possible to provide an upper bound. If there is an $r_{\alpha}$ such that $P[\mathrm{IF}(r_{\alpha}, X, R, T; \pi^\star, m^\star)] = 0$, then the conclusions of Theorem~\ref{thm:coverage-final} can be made two-sided. The condition of existence of $r_{\alpha}$ which makes $P[\mathrm{IF}(\cdots)] = 0$ is same as saying that there exists an $r_{\alpha}$ such that $\mathbb{P}(R(X, Y) \le r_{\alpha}|T = 1) = 1 - \alpha$, i.e., there are no jumps in the distribution of $R(X, Y)|T = 1$ at $(1 - \alpha)$-th quantile, or equivalently, the distribution function takes the value of $1-\alpha$. Under this condition, inequality~\eqref{eq:final-coverage-unconditional}, for instance, can be strengthened to
\begin{equation}\label{eq:final-coverage-unconditional-2}
    \begin{split}
        \Bigl|
        \mathbb{P}_{(X, Y)\sim Q_X\otimes P_{Y|X}}\Bigl(Y \in \hat{C}(\widehat{r}_{\alpha} ; X) \Bigr) - (1 - \alpha)
        \Bigr|
        &\le \mathbb{E}\left[\frac{\|\widehat{\pi} - \pi^{\star}\|_2}{\mathbb{P}(T = 1)}\sup_{\theta}\|\widehat{m}(\theta, \cdot) - m^{\star}(\theta, \cdot)\|_2\right]
        \\ 
        &\qquad
        + \frac{\mathfrak{C}}{\mathbb{P}(T = 1)}\frac{m_0 + \pi_0   +1}{\sqrt{|\mathcal{I}_2|}} .
    \end{split}
\end{equation}
Similar strengthening also holds for~\eqref{eq:main-coverage}. See the proof of Theorem~\ref{thm:coverage-final} for details.
\end{remark}

\section{Methodology without sample splitting}\label{sec:method-full}
In this section, we provide an alternative methodology that builds upon Algorithm \ref{alg:influence-conformal-split} but is potentially more efficient by avoiding sample splitting. 
The procedure is summarized in Algorithm \ref{alg:influence-conformal-full}.
\begin{algorithm}[h]
    \SetAlgoLined
    \SetEndCharOfAlgoLine{}
   \KwIn{Training data $\mathcal{D}^{\text{tr}} = \mathcal{D}^{\text{tr}}_P \cup \mathcal{D}^{\text{tr}}_Q$; Coverage probability $1 - \alpha$, a training method $\mathcal{A}$ and estimators $\widehat{\pi}, \widehat{m}$, the point for prediction $x$.}
    \KwOut{A valid prediction set $\widehat{C}_{\alpha}(x)$.}
    Fit the training method $\mathcal{A}$ on $\mathcal{D}^{\text{tr}}$ and using fitted method $\mathcal{A}$, construct an increasing (nested) sequence of sets $\{\mathcal{F}_t\}_{t\in\mathcal{T}}$. Here $\mathcal{T}$ is a subset of $\mathbb{R}$. The nested sets $\{\mathcal{F}_t\}_{t\in\mathcal{T}}$ can depend arbitrarily on $\mathcal{D}^{\text{tr}}$.\;
  For each $i\in [N]$ that satisfies $T_i=0$, define the conformal score
    \[
    r_i ~=~ r(Z_i) := \inf\{t\in\mathcal{T}:\,Z_i\in\mathcal{F}_t\}.
    \]\;
    \vspace{-1em}
    Fit estimators $\widehat{\pi},\widehat{m}$ on $\mathcal{D}^{\text{tr}}$ and find the smallest $\hat{\theta} = \widehat{r}_{\alpha}$ such that $\mathbb{P}_{N}[\mathrm{IF}({\theta}, X,R,T; \widehat{\pi}, \widehat{m})] \ge 0$, where
    \begin{align*}
    \mathbb{P}_{N} \bigl[ \mathrm{IF}({\theta},X,R,T; \widehat{\pi},\widehat{m}) \bigr] ~&=~ \frac{1}{N}\sum_{i=1} ^N\mathbbm{1}\{t_i=0\} \widehat{\pi}(x_i) \big[ \mathbbm{1}\{r_i \leq {\theta} \} - \widehat{m}({\theta},x_i) \bigr]\\
    ~&\quad+~ \frac{1}{N} \sum_{i=1}^N \mathbbm{1}\{t_i=1\}[\widehat{m}({\theta},x_i) - (1-\alpha)].
    \end{align*}
  
    \Return the prediction set $\widehat{C}_{\alpha}(x):= \{ y: R(x,y) \leq \widehat{r}_{\alpha} \}$.
    \caption{Full doubly robust prediction}
    \label{alg:influence-conformal-full}
\end{algorithm}

Similar as in \eqref{eq:coverage-guarantee-split-IF}, let $\mathbf{I}:= \{P[\mathrm{IF}(\widehat{r}_{\alpha}, X, R, T; \widehat{\pi}, \widehat{m})] - \mathbb{P}_{N}[\mathrm{IF}(\widehat{r}_{\alpha}, X, R, T; \widehat{\pi}, \widehat{m})]\}/ \mathbb{P}(T = 1)$. We now provide a bound on $\mathbf{I}$ in Theorem \ref{thm:convergence-if-function-whole}. Note that in addition to assumptions~\ref{assump:DGP}--\ref{assump:monotone}, because we are doing the training and evaluating on the same dataset, we need some additional assumptions on the classes of estimators in order to apply results from empirical processes to ensure $\mathbf{I}$ converges to zero.

\begin{enumerate}[label=\bf(A\arabic*)]
\setcounter{enumi}{3}
\item \label{assump:covering} Assume that $\widehat{\pi}(\cdot)$ and its limit $\pi(\cdot)$ are in function classes $\mathcal{F}_{\pi}$ and $\widehat{m}(\cdot,\cdot)$ and its limit ${m}(\cdot,\cdot)$ are in $\mathcal{F}_{m}$, such that for some $\alpha_\pi, \alpha_m \geq 0$, the covering numbers satisfy $\forall \varepsilon >0,$
\[
\log N(\varepsilon, \mathcal{F}_{\pi}, L_2(Q)) \le C\varepsilon^{-\alpha_{\pi}}, \text{  and }\log N(\varepsilon, \mathcal{F}_{m}, L_2(Q)) \le C\varepsilon^{-\alpha_{m}},
\]
where $C$ is some constant and $Q$ is any discrete probability measure, and the covering number $N \bigl(\varepsilon, \mathcal{F}, L_r(Q) \bigr)$ is defined in the same way as in \cite{kearns1994introduction}.
\end{enumerate}

\begin{thm}\label{thm:convergence-if-function-whole}
Under assumption \ref{assump:DGP}, for any estimators $\widehat{\pi}, \widehat{m}$ satisfying assumptions \ref{assump:bounded},  \ref{assump:monotone} and \ref{assump:covering}, there exists a universal constant $\mathfrak{C}$ such that for any $\delta>0$, 
\begin{equation}
\begin{aligned}
     \PP \Biggl( |\mathbf{I}| \leq \mathfrak{C}  \biggl\{  &N^{- 1/(\alpha_m \vee 2)} \bigl(1+\mathbbm{1}{\{\alpha_m=2\}}\log N \bigr) + N^{- 1/(\alpha_\pi \vee 2)} \bigl(1+\mathbbm{1}{\{\alpha_\pi=2\}}\log N \bigr) +\\
     &\sqrt{\frac{ (m_0+\pi_0 +1- \alpha)^2 \log \bigl( \frac{1}{\delta}  \bigr) }{N}} \biggr\}/ \mathbb{P}(T = 1) \Biggr) \geq 1-\delta,
\end{aligned}
\end{equation}
where $\vee$ is the maximum operator. Moreover,  there exists a universal constant $\mathfrak{C}'$ such that
\[
\mathbb{E}\left[|\mathbf{I}|\right] ~\le~ \frac{\mathfrak{C}'}{\mathbb{P}(T = 1)} \biggl( N^{- 1/(\alpha_m \vee 2)} \bigl(1+\mathbbm{1}{\{\alpha_m=2\}}\log N \bigr) + N^{- 1/(\alpha_\pi \vee 2)} \bigl(1+\mathbbm{1}{\{\alpha_\pi=2\}}\log N \bigr) + \frac{m_0 + \pi_0 +1 - \alpha }{\sqrt{N}} \biggr).
\]
\end{thm}
The proof of this theorem is in Section \ref{appsec:proof-convergence-if-funcion-whole} of the appendix.

Using the definition of $\widehat{r}_\alpha$ and combining Theorems  \ref{thm:product-bias} and \ref{thm:convergence-if-function} together with~\eqref{eq:coverage-guarantee-split-IF} yields the following result for the final coverage of the prediction region from Algorithm \ref{alg:influence-conformal-full}.

\begin{thm}\label{thm:coverage-final-full}
Under assumption \ref{assump:DGP}, 
for any estimators $\widehat{\pi}, \widehat{m}$ satisfying assumptions \ref{assump:bounded}, \ref{assump:monotone} and \ref{assump:covering}, there exists a universal constant $\mathfrak{C}$ such that for any $\delta > 0$ with probability at least $1-\delta$,
\begin{equation}\label{eq:main-coverage-full}
\begin{aligned}
&\mathbb{P}_{(X, Y) \sim Q_{X} \times P_{Y \mid X}}\Bigl( Y \in \hat{C}(\widehat{r}_{\alpha} ; X) \mid \mathcal{D}^{\mathrm{tr}} \Bigr) \geq  1-\alpha
- \frac{1}{\PP(T=1)}\biggl(  \|\widehat{\pi} -  \pi^\star \|_2    \sup_{\theta} \| \widehat{m}(\theta,\cdot) - m^\star(\theta,\cdot)\|_2 \\&\quad+ N^{- 1/(\alpha_m \vee 2)} \bigl(1+\mathbbm{1}{\{\alpha_m=2\}}\log N \bigr) + N^{- 1/(\alpha_\pi \vee 2)} \bigl(1+\mathbbm{1}{\{\alpha_\pi=2\}}\log N \bigr) + \mathfrak{C} (m_0+\pi_0+1) \sqrt{\frac{  \log \bigl( {1}/{\delta}  \bigr) +1 }{ N}} \biggr).
\end{aligned}
\end{equation}
Moreover, 
\begin{equation}\label{eq:final-coverage-unconditional-full}
    \begin{split}
        &\mathbb{P}_{(X, Y)\sim Q_X\otimes P_{Y|X}}\Bigl(Y \in \hat{C}(\widehat{r}_{\alpha} ; X) \Bigr) - (1 - \alpha)\\ 
        &\qquad\ge -\frac{\mathfrak{C}}{\mathbb{P}(T = 1)}\biggl( N^{- 1/(\alpha_m \vee 2)} \bigl(1+\mathbbm{1}{\{\alpha_m=2\}}\log N \bigr) + N^{- 1/(\alpha_\pi \vee 2)} \bigl(1+\mathbbm{1}{\{\alpha_\pi=2\}}\log N \bigr) + \frac{m_0 + \pi_0 +1 }{\sqrt{N}} \biggr)\\
        &\qquad \quad - \mathbb{E}\left[\frac{\|\widehat{\pi} - \pi^{\star}\|_2}{\mathbb{P}(T = 1)}\sup_{\theta}\|\widehat{m}(\theta, \cdot) - m^{\star}(\theta, \cdot)\|_2\right].
    \end{split}
\end{equation}
\end{thm}
\begin{proof}
This is a direct result of Lemma \ref{lem:connection-if-coverage},~\eqref{eq:coverage-guarantee-split-IF}, and Theorem~\ref{thm:convergence-if-function-whole}.
\end{proof}

Comparing the results of Theorem \ref{thm:coverage-final-full} with Theorem \ref{thm:coverage-final} in terms of the slack in the miscoverage probability, we notice that the full data based prediction set can have larger miscoverage error than the split data version. However, in terms of how close $\widehat{r}_{\alpha}$ is to $r_{\alpha}$, we expect the full data version to perform better compared to the split data version. This is expected because with full data version the quantile estimator is based on $N$ observations rather than $|\mathcal{I}_2|$ observations which is smaller than $N$. The variance of the quantile estimator can be better up to a constant of $N/|\mathcal{I}_2|$. We do not pursue these variance comparisons for $\widehat{r}_{\alpha}$ as that is not our goal.   

\section{Simulation Studies}\label{sec:simulations}
In practice, unless $R(\cdot,\cdot)$ is a map that is independent of data, we further split $\mathcal{D}_1$ of Algorithm \ref{alg:influence-conformal-split} and use the first split to train $R(\cdot,\cdot)$, while the second split is used to estimate the two nuisance parameters $\pi^\star(\cdot)$ and $m^\star(\cdot,\cdot)$. We include the two algorithms in Algorithm \ref{alg:influence-conformal-split} and Algorithm \ref{alg:influence-conformal-full} under the name  ``DRP w. three splits" and ``DRP w. full data" respectively. We also include simulation results where $R$, $\widehat{m}$ and $\widehat{\pi}$ are trained on the same split and the remainder of the data is used for validation under the name ``DRP w. two splits" in Section \ref{sec:absolute-residuals} of the appendix. For both synthetic data and the real data set we use two kinds of score functions to estimate $R(\cdot,\cdot)$, 
\begin{enumerate}
    \item absolute residual score $|y - \widehat{\mu}(x)|$ where $\mu(x) := \E(Y|X=x)$ is estimated by ridge regression,
    \item conditional quantile regression \footnote{the cutting edge algorithm employed in \cite{lei2020conformal} paper};
\end{enumerate}
The nuisance parameters $\pi^\star$ and $m^\star(\cdot, \cdot)$ are estimated through SuperLearner\footnote{SuperLearner uses cross-validation to estimate the performance of multiple machine learning models and then creates an optimal weighted average of those models using the test data. This approach has been proven to be asymptotically as accurate as the best possible prediction algorithm that is tested. For details please refer to \cite{Polley2010SuperLI}.} that includes both RandomForest and generalized linear model (GLM).  To avoid numerical issues, 
we clip the propensity score at 0.99 to prevent $\widehat{\pi}$ from becoming unbounded. The proposed methods are compared against the weighted conformal prediction (WCP) method proposed in \cite{tibshirani2019conformal}. Note that in the weighted conformal prediction method, the prediction interval is given by
\begin{equation}
\widehat{C}_{n}(x)=\mu_{0}(x) \pm \text { Quantile }\biggl(1-\alpha ; \sum_{i=1}^{n} p_{i}^{w}(x) \delta_{\left|Y_{i}-\mu_{0}\left(X_{i}\right)\right|}+p_{n+1}^{w}(x) \delta_{\infty}\biggr),
\end{equation}
where $p^w(x)$ is a function which depends on the likelihood ratio between the two covariate distributions, or $\pi^\star(x)$.
Therefore, when the distribution shift is too ``large", i.e. $p^w_{n+1}(x)$ is larger than $\alpha$, the width becomes $\infty$. And indeed we observe infinite widths for this method over $50\%$ of the time across all the predicted points and the Monte Carlo replications on synthetic data and over $90\%$ of the time on real data. For illustrative purposes, we truncate the width for WCP at 10 and 50 on synthetic data and real data respectively, and demonstrate the mean width from 500 Monte Carlo simulations. 

It is shown in \cite{lei2020conformal} that the CQR score would guarantee asymptotic conditional coverage $\mathbb{P}\bigl(Y_f \in \widehat{C}_{N,\alpha} (X_f) | X_f = x_f \bigr)$,  we also conduct an experiment using this score for our method under the setting of Section \ref{sec:synthetic}, where we specify 200 test points of $X_f$ (generated from standard normal distributions $N(0,I_4)$) and for each test point, 100 $Y_f$'s following the distribution \eqref{eq:synthetic} are generated to test if they fall into the prediction sets trained by our method and WCP, where we report the average (reflected points) for each test point in Figure \ref{fig:cqr-conditional},which shows the coverage and width for each test point that is represented by the $L_2$ norm (the X axis). We also fit a smoothing spline (with default parameters of the R function 
\textbf{smooth.spline}) with these points.

\subsection{Real data}\label{sec:real}
We demonstrate the use of conformal prediction in the covariate shift setting in an empirical example. We re-analyze the data set used in \cite{tibshirani2019conformal} which is the airfoil data set from the UCI Machine Learning Repository which contains $N=1503$ observations of a response $Y$ (scaled sound pressure level of NASA airfoils), and a vector of covariates $X$ with $d=5$ dimension (log frequency, angle of attack, chord length, free-stream velocity, and suction side log displacement thickness). Label missingness was then generated as in \cite{tibshirani2019conformal} using a propensity score model analogous to the one specified in Section 7.2 below. 
\begin{figure}[!h]
    \centering
    \includegraphics[width=\textwidth]{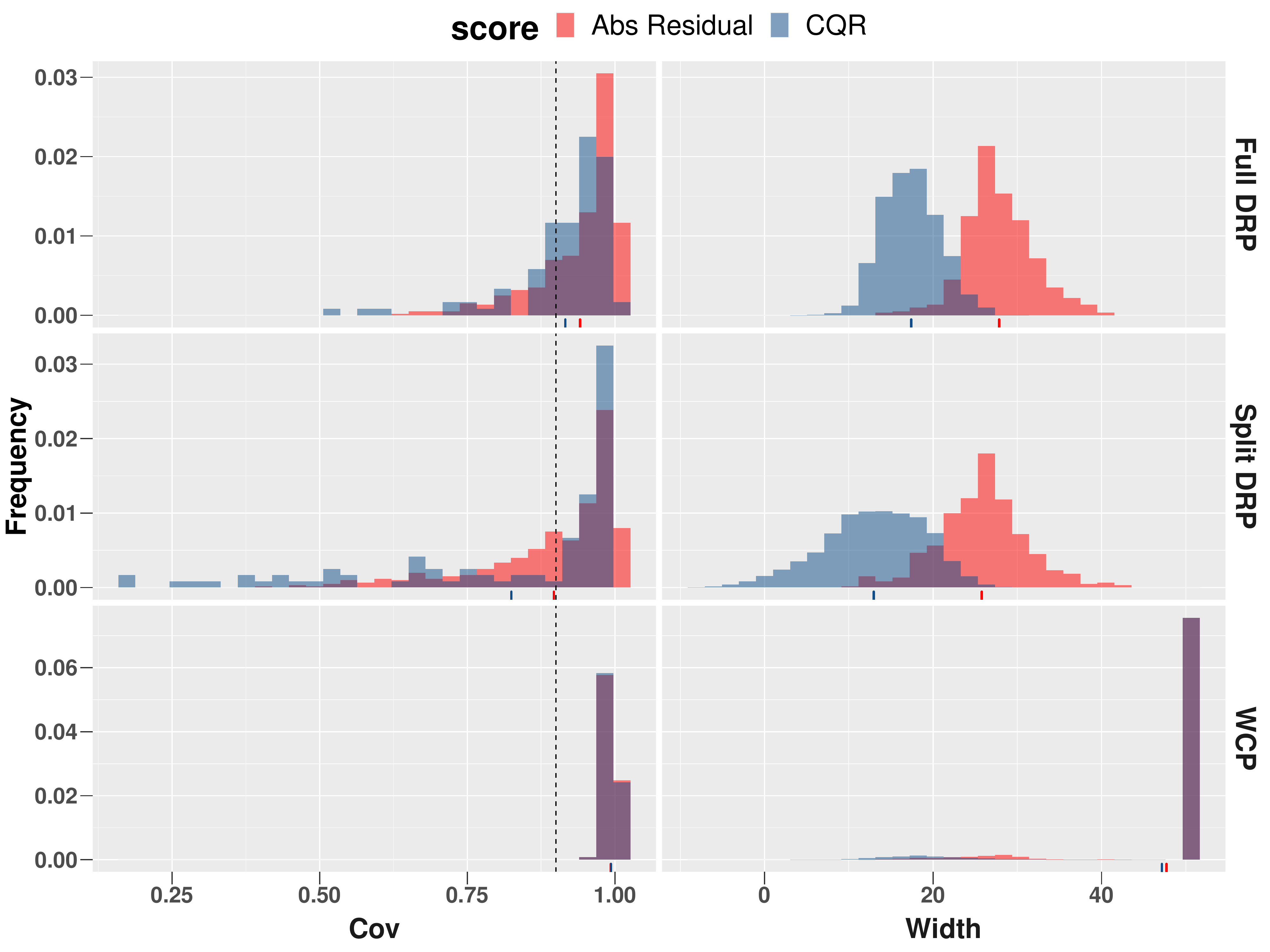}
    \caption{Histograms of coverage and width of Double Robust Prediction (DRP) and Weighted Conformal Prediction (WCP) on real data through either absolute residual score or the CQR score, where the width is truncated at 50 when WCP produces infinity width.}
    \label{fig:cqr-real-all}
\end{figure}


\begin{figure}[!h]
    \centering
    \includegraphics[width=\textwidth]{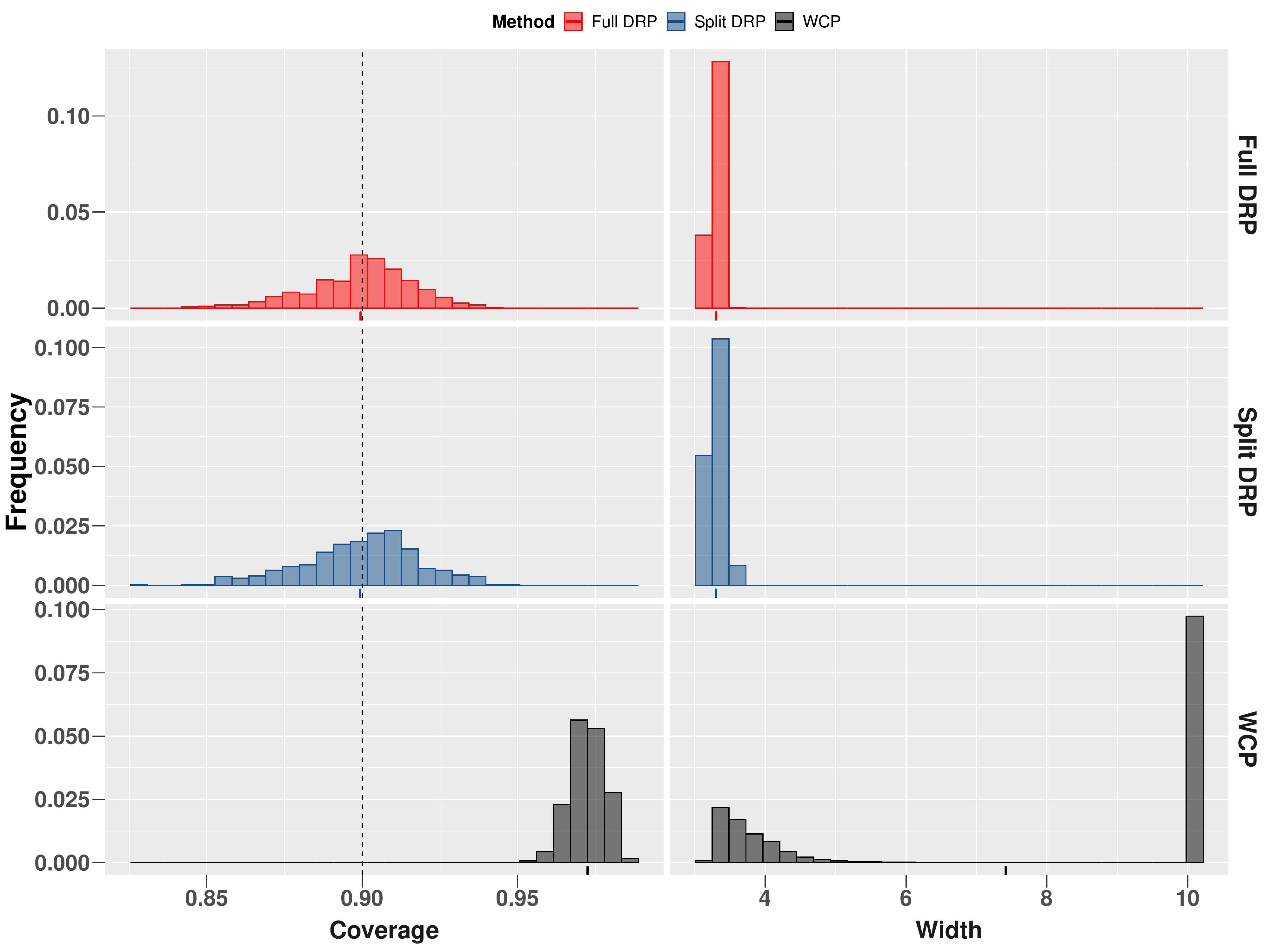}
    \caption{Histograms of coverage and width of Doubly Robust Prediction (DRP) and Weighted Conformal Prediction (WCP) on synthetic data using the absolute residual score. The width is truncated at 10 for WCP.}
    \label{fig:synthetic_new}
\end{figure}


\subsection{Synthetic data}\label{sec:synthetic}

Here we use the setting from \cite{kang2007demystifying} where for each unit $i=1, \ldots, n=2000$, suppose that $\left(x_{i 1}, x_{i 2}, x_{i 3}, x_{i 4}\right)^{\top}$ is independently distributed as $N(0, I_4)$ where $I_4$ is the $4 \times 4$ identity matrix. The $y_{i}$'s are generated as
\begin{equation}\label{eq:synthetic}
y_{i}=210+27.4 x_{i 1}+13.7 x_{i 2}+13.7 x_{i 3}+13.7 x_{i 4}+\varepsilon_{i},
\end{equation}
where $\varepsilon_{i} \sim N(0,1)$, and the true propensity scores are
$$
\PP(T=1 |x_i)=\operatorname{expit}\left(-x_{i 1}+0.5 x_{i 2}-0.25 x_{i 3}-0.1 x_{i 4}\right), \text{ where } \mathrm{expit}(x_i^\top \alpha) = \frac{\exp(x_i^\top \alpha)}{1+\exp(x_i^\top \alpha)}.
$$


\subsection{Simulation results}
The mean coverage and width from 500 Monte Carlo simulations using the absolute residual score are shown in Table \ref{tab:synthetic-real}, where the middle column corresponds to synthetic data results and the rightmost column to real data results. Figure \ref{fig:cqr-real-all}  display histograms of coverage and width from 500 runs on real data using either the absolute residual score and the CQR score while \ref{fig:synthetic_new} shows the use of absolute residual score. For synthetic data, we keep the results using the two scores separate, with the absolute residual score in Figure \ref{fig:synthetic_new} and the CQR score to Appendix \ref{sec:cqr-synthetic} for easier comparison between our DRP methods and the WCP method. 
Below we make a few observations on simulation results:
\begin{itemize}
 \item WCP produces wider width and therefore, tends to over-cover by a considerable amount (by more than $7\%$ over the nominal coverage of $90\%$).
 \item Doubly robust prediction with full data and multiple splits have similar performance with valid coverage. In practice, because using three splits guaranteed nominal coverage in sufficiently large sample with fewest assumptions, we recommend this approach as it does not appear to suffer much efficiency loss.
 \item For the conditional coverage simulation results described in the last paragraph at the start of Section \ref{sec:simulations}, we see that DRP has similar coverage and much smaller width compared to WCP. Note that they both attain desired coverage when the norm of the test data is less than 2, which is what the $L_2$ norm of a standard 4-dimensional normal r.v. concentrates to. As the norm gets past 2, there are less data points and hence we get fewer observations.
\end{itemize}

\begin{table}[H]
\centering
\begin{tabular}{c|cc|cc}
    Mean coverage and width& \multicolumn{2}{c|}{ Synthetic data} & \multicolumn{2}{c}{Real data} \\
    from 500 monte carlo runs & Coverage          & \multicolumn{1}{c|}{Width}                 & Coverage     & \multicolumn{1}{c}{Width}         \\ \hline\hline
DRP w. full data  & 0.90       & 3.29       & 0.94   & 27.85\\
DRP w. three splits  & 0.90       & 3.30       & 0.90   & 25.79\\
DRP w. two splits  & 0.90       & 3.29       & 0.88   & 25.19\\
WCP & 0.97     & 7.41       & 0.99   & 47.71 \\
\hline
\end{tabular}
\caption{Coverage and width of DRP and WCP on synthetic and real data. Clearly, DRP improves on WCP in terms of width while maintaining coverage close to the nominal level of 0.9.}\label{tab:synthetic-real}
\end{table}



\begin{figure}[!h]
    \centering
    \includegraphics[width=\textwidth]{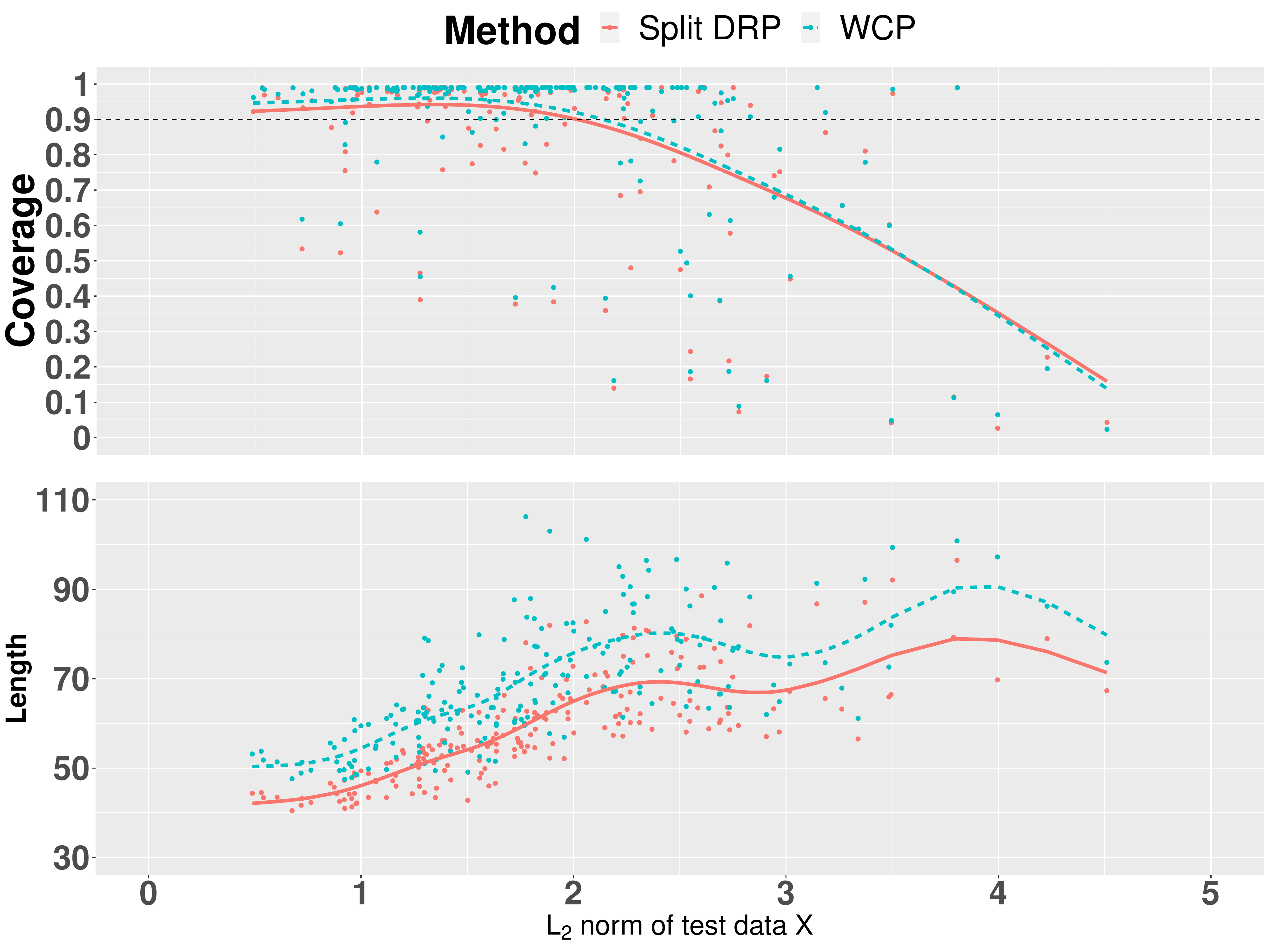}
    \caption{Estimated conditional coverage and length at 200 test points with Split DRP and WCP using the CQR score. The setting is described in the last paragraph at the start of Section \ref{sec:simulations}. The points are the average of the coverage and width from 100 Monte Carlo simulations and the lines are drawn by fitting a smoothing splines for those points.}
    \label{fig:cqr-conditional}
\end{figure}

\section{Aggregation of prediction sets}\label{sec:efcp}
When, as typically the case in practice, multiple training methods are available, we propose to combine the proposed approach with that of  \cite{yang2021finite} in order to construct a better prediction set with smaller width. The detailed procedure is given in Algorithm \ref{alg:influence-conformal-efcp} and can be shown by arguments given in \cite{yang2021finite}  to retain coverage validity while attaining the smallest width in large samples with high probability.   
\begin{algorithm}[h]\label{alg:efcp}
    \SetAlgoLined
    \SetEndCharOfAlgoLine{}
    \KwIn{Training data $\mathcal{D}^{\text{tr}} = \mathcal{D}^{\text{tr}}_P \cup \mathcal{D}^{\text{tr}}_Q$ split into $\mathcal{D}_1$ and $\mathcal{D}_2$, where $\mathcal{D}_1 = \{Z_i \in \mathcal{D}^{\text{tr}}, i \in \mathcal{I}_1\}$ and $\mathcal{D}_2 = \{Z_i \in \mathcal{D}^{\text{tr}}, i \in \mathcal{I}_2\}$; Estimators $\widehat{\pi}, \widehat{m}$,  and training methods $\mathcal{A}_k, k\in [K]$; The prediction point $x$.}
    \KwOut{A valid prediction set $\widehat{C}^{\mathrm{EFCP}}_{\alpha}(x)$ with smallest width.}
    Fit training methods $\mathcal{A}_1, \ldots, \mathcal{A}_K$ on $\mathcal{D}_1$ and for each fitted method $\mathcal{A}_k$, construct an increasing (nested) sequence of sets $\{\mathcal{F}_t^{(k)}\}_{t\in\mathcal{T}}$. Here $\mathcal{T}$ is a subset of $\mathbb{R}$.\;
     For each $i\in\mathcal{I}_2$ that satisfies $T_i=0$, define the conformal score
    \[
    r_k(Z_i) := \inf\{t\in\mathcal{T}:\,Z_i\in\mathcal{F}_t^{(k)}\}.
    \]\;
    \vspace{-1em}
    Fit estimators $\widehat{\pi},\widehat{m}$ on $\mathcal{D}_1$ and solve for $\hat{\theta} = \widehat{r}_{\alpha,k}$ as the solution to $\mathbb{P}_{\mathcal{I}_2} \big[\mathrm{IF}(\hat{\theta}, \mathcal{D}_2,r_k; \widehat{\pi}, \widehat{m}) \bigr] = 0$, where
    \begin{align*}
    \mathbb{P}_{\mathcal{I}_2}[\mathrm{IF}(\hat{\theta}, X, R, T; \widehat{\pi}, \widehat{m})] ~&=~ \frac{1}{|\mathcal{I}_2|}\sum_{i \in \mathcal{I}_2} \mathbbm{1}\{t_i=0\} \widehat{\pi}(x_i) \big[ \mathbbm{1}\{r_i \leq \hat{\theta} \} - \widehat{m}(\hat{\theta},x_i) \bigr]\\
    ~&\quad+~ \frac{1}{|\mathcal{I}_2|} \sum_{i \in \mathcal{I}_2} \mathbbm{1}\{t_i=1\}[\widehat{m}(\hat{\theta},x_i) - (1-\alpha)].
    \end{align*}
    Compute the corresponding conformal prediction set as
    $\widehat{C}_k(x) ~:=~ \{y:\, r_k(x,y) \le \widehat{r}_{\alpha,k}\}.$\;
    Set $$\widehat{k} := \argmin_{1\le k\le K}\,\mbox{Width}(\widehat{C}_k (x)).$$
    Here $\mbox{Width}(\cdot)$ can be any measure of width or volume of a prediction set. The quantity $\widehat{k}$ need not be unique and any minimizer can be chosen.\;
    \Return the prediction set $\widehat{C}_{\widehat{k}}$ as $\widehat{C}^{\mathrm{EFCP}}_{\alpha}$.
    \caption{Efficient doubly robust prediction}
    \label{alg:influence-conformal-efcp}
\end{algorithm}

Table \ref{tab:efcp} reports results from a simulation study comparing the proposed doubly robust prediction algorithm to select an optimal tuning parameter for ridge regression-based conformal score, with cross-validation aimed at minimizing the ridge regression MSE. The simulation results confirm the proposed algorithm's ability to preserve marginal prediction coverage while optimizing prediction interval width, both in synthetic and real data sets.  

\begin{table}[H]
\centering
\begin{tabular}{c|cc|cc}
    Mean coverage and width& \multicolumn{2}{c|}{ Synthetic data} & \multicolumn{2}{c}{Real data} \\
    from 500 monte carlo runs & Coverage          & \multicolumn{1}{c|}{Width}                 & Coverage     & \multicolumn{1}{c}{Width}         \\ \hline\hline
Efficient DRP  & 0.90       & 3.30       & 0.89   & 27.30\\
DRP w. CV  & 0.89       & 3.32       & 0.84   & 18.61\\
\hline
\end{tabular}
\caption{Coverage and width of efficient doubly robust prediction and doubly robust prediction with cross-validation on synthetic and real data. Efficient DRP improves on CV in terms of width while maintaining coverage close to the nominal level of 0.9.}
\label{tab:efcp}
\end{table}

\section{Sensitivity analysis for latent covariate shift}\label{sec:sensitivity}
Thus far, we have assumed that the covariate shift problem is primarily due to observed covariates, which we have denoted explainable covariate shift (or equivalently that outcomes for the target population are missing at random), an assumption that cannot be confirmed empirically without invoking a different non-testable assumption. In this section, we relax this assumption $P_{Y|X} = Q_{Y|X}$ and propose a sensitivity analysis for obtaining doubly robust calibrated prediction sets accounting for a latent covariate shift problem encoded in a sensitivity parameter. We define a sensitivity function as
\begin{align*}
\gamma^\star(x,y) ~=~ \log  \left(\frac{ P_{Y=y| X }   }{ Q_{Y=y| X }  } \bigg/ \frac{ P_{Y=y_0| X }  }{ Q_{Y=y_0| X }  }\right) , 
\end{align*}
where $y_0$ is any baseline value for $Y$. Here $\gamma^\star(x,y)$ is the sensitivity analysis function encoding a hypothetical departure from the assumption that $X$ suffices to account for the covariate shift problem with $\gamma^\star(x,y_0)=0$, and $\gamma^\star(x,y) = 0$ for all $y$ recovers the standard assumption of explainable covariate shift. For simplicity we take $y_0=0$. Our sensitivity analysis is inspired by a semiparametric approach for accounting for data missing not at random (MNAR), due to Chapter 5 of \cite{robins2000sensitivity} in the missing data literature.  Formally, the sensitivity function $\gamma^\star$ can be represented in the missing data notation as
\begin{align*}
\gamma^\star(x,y) ~=~ \log \frac{\PP(T=0|X=x,Y=y)\PP(T=1|X,Y=0)}{\PP(T=0|X,Y=0)\PP(T=1|X=x,Y=y)}.
\end{align*}
For any three functions $\eta(\cdot)$, $m(\cdot, \cdot)$, and $\gamma(\cdot, \cdot)$, let  
\begin{equation}\label{eq:if-sensitivity}
\begin{aligned}
\mathrm{IF}(\theta, x,y,r,t;\eta,m, \gamma) :
&= \mathbbm{1}\{t = 0\} \exp \bigl\{  - \eta(x) - \gamma(x,y)\bigr\}\Big[\mathbbm{1}\{r \le \theta\} -m(\theta,x) \Big] \\ 
&\quad+ \mathbbm{1}\{t = 1\}\Big[ m(\theta,x) - (1 - \alpha)\Big],
\end{aligned}
\end{equation}
where the two nuisance functions are
\begin{equation}
\begin{split}
\eta^\star(x) ~&:=~ \log \frac{\PP(T=0|X=x,Y=0)}{\PP(T=1|X=x,Y=0)}\\
m^\star(\theta,x) ~&:=~  \PP(R \leq \theta| X=x,T=1).
\end{split}
\end{equation}
Per the sensitivity analysis framework, we assume $\gamma^\star(x, y) = \gamma(x, y)$ is known.
See also Example 2 of \cite{tsiatis2014sensitivity} for some explanations on why this is a nonparametric identified model.
In this framework we have the following theorem.
\begin{thm}\label{thm:sensitivity-if}
Under the assumption that $\E  \frac{\PP^2(T=1|X,Y )}{\PP(T=0|X,Y)}$ is finite and that the density of the conditional distribution of $R|T=1$ at $r_\alpha$ is bounded away from zero, the efficient influence function of  $r_\alpha$, the $(1-\alpha)-$quantile of $R|T=1$ in the  nonparametric model which allows the distribution of $Z$ to remain unrestricted, and the conditional log odds ratio function relating $T$ to $R$ given $X$ is known to equal $\gamma^\star$ is given up to a proportionality constant by
\begin{equation}\label{eq:if-sensitivity-exact}
\begin{aligned}
\psi(z) = \mathrm{IF}(r_\alpha, x,y,r,t;\eta^\star,m^\star, \gamma^\star)&= \mathbbm{1}\{t = 0\} \exp \bigl\{  - \eta^\star(x) - \gamma^\star(x,y)\bigr\}\Big[\mathbbm{1}\{r \le r_\alpha\} -m^\star(r_\alpha,x) \Big] \\ 
&\quad+ \mathbbm{1}\{t = 1\}\Big[ m^\star(r_\alpha,x) - (1 - \alpha)\Big]
\end{aligned}
\end{equation}
Furthermore, the moment function $\mathrm{IF}(r_\alpha, x,y,r,t;\eta,m,\gamma^\star)$ satisfies the double robustness property that
\begin{equation}
    \E \bigl[ \mathrm{IF}(r_\alpha, x,y,r,t;\eta,m, \gamma^\star) \bigr] = 0, 
\end{equation}
if either $\eta = \eta^\star$ or $m=m^\star$. 
\end{thm}
\begin{proof}
See Section~\ref{sec:proof-of-sensitivity-if} for a proof.
\end{proof}

Theorem \ref{thm:sensitivity-if} gives the efficient influence function of $r_{\alpha}$, which provides a moment equation for $r_{\alpha}$. And this can be leveraged as an estimating equation similar to Lemma \ref{lem:if} and Theorem \ref{thm:coverage-final} to yield a valid prediction set with the product bias from estimating the nuisance functions. In this case, estimation of nuisance functions is not straightforward and not further pursued in this paper. 
We also include a review of existing literature on this topic. The recent paper \cite{jin2021sensitivity} generalized the covariate shift setting to distributional shift where the joint distribution of covariates and response can be different between the test and training data. The motivation is sensitivity analysis for individual treatment effects. 
They proposed a robust conformal prediction algorithm that builds upon the weighted conformal inference method from \cite{tibshirani2019conformal}. The paper derives prediction sets that achieve marginal coverage if the propensity score is known exactly, allowing for latent covariate shift of a magnitude bounded by a known sensitivity parameter. 
They also proposed a second algorithm 
where the coverage is attained with high probability $1-\delta$ conditional on the training data, and showed their results provide tight prediction sets in some cases that cannot be improved upon under their assumptions. The coverage bias of both methods is first order.
Both their methods require the test point to be specified in advance so that they can compute a different quantile for every new $x$ and thus may be computationally intensive.  
In terms of sensitivity analysis for unmeasured confounding, we note that their paper considers a different sensitivity framework than ours,  as they posit the existence of an unobserved confounder $U$ which together with $X$ completely accounts for confounding. They encode the magnitude of unmeasured confounding in terms of upper and lower bounds 
for the likelihood ratio of density of $U$ in the treated and untreated samples over the support of $U$. Thus, their sensitivity parameter appears to capture both the extent of residual confounding, but also reflects aspects of the density of $U$ that may not be of scientific interest, and for which an investigator may not have any prior information. An implication of this choice of parameterization is that for a given sensitivity bound, the bound may reflect small amount of confounding over a wide support of $U$, or a large amount of confounding, over a narrow support of $U$, rendering such sensitivity analysis difficult to interpret. In contrast, we 
 favor the approach of \cite{robins2000sensitivity} with a more direct sensitivity analysis, which in the counterfactual setting, encodes the departure from unconfoundedness in terms of a likelihood ratio for the counterfactual outcome in view of the treated and control arms conditional on observed covariates. 
\cite{yin2021conformal} also studied the sensitivity analysis of ITE using a similar approach as the first method proposed in \cite{jin2021sensitivity} while their method of analysis offers a different perspective.

\section{Discussion }\label{sec:discussion}
This paper has proposed three separate algorithms to construct prediction regions which are adaptive to unknown covariate shift between a population from which labeled data are available and an unlabeled population for which outcome prediction is in view. Our three proposed methods have been described as ``Split doubly robust prediction", ``Full doubly robust prediction" and ``Efficient doubly robust prediction". The paper provided a rigorous analysis of the coverage properties of these algorithms, notably establishing that all have coverage bias of a product form and providing formal conditions under which all are asymptotically well-calibrated, in the sense of attaining the nominal coverage rate in large samples. ``Split doubly robust prediction"  has coverage guarantees in large samples under minimal conditions, but requires one to use a non-negligible subset of the data for training; in contrast, ``Full doubly robust prediction" uses the entire data set both for training and prediction, and attains the nominal coverage in large samples under relatively stronger conditions. ``Efficient doubly robust prediction" combines ``Split doubly robust prediction" and the EFCP algorithm from \cite{yang2021finite}, which is empirically shown to potentially outperform a standard cross-validation approach. We conjecture that the proposed efficient doubly robust prediction algorithm is nearly as efficient as an oracle with a priori knowledge of the optimal prediction interval, although formally proving this result is left to future work. 
An important advantage of our framework is that its large sample efficiency and validity guarantees hold for any collection of machine learning techniques and their respective tuning parameters, under the relatively mild requirement that at least one of two estimated nuisance functions is consistent, without necessarily requiring fast convergence rates for the latter. 

Another important contribution of the paper is to draw upon a key equivalence between the explainable covariate shift problem, the MAR assumption in the missing data literature, and the notion of unconfoundedness in the causal inference literature, to develop a sensitivity analysis approach to evaluate the extent to which prediction regions might be impacted by hypothetical departures from this assumption. Notably, the proposed methods readily extend to accommodate such sensitivity analysis via a slight modification of our procedure to incorporate a sensitivity parameter, without compromising the product bias or double robustness property of the approach. Fully developing prediction inference for this sensitivity analysis framework however requires care in estimation of nuisance functions which, due to space limitation, we plan to consider in future work.
Overall, this paper reveals and leverages deep connections between modern literatures of semiparametric theory, missing data and causal inference, and emerging methods for well-calibrated prediction inference. To the best of our knowledge, such connections have previously not been drawn upon as deliberately as shown to be possible in this work which we hope will generate both interest and further developments towards even more robust and efficient well-calibrated prediction. One possible line of future investigation might be to build on recent theory of higher order influence functions due to Robins and colleagues, see e.g. \cite{robins2008higher} and \cite{robins2017minimax}, which in principle could be used to reduce the second order product bias obtained in this paper to a higher order product bias, therefore potentially improving finite sample coverage over a wider range of regimes.

\bibliographystyle{plainnat}
\bibliography{ref}

\newpage
\setcounter{section}{0}
\setcounter{equation}{0}
\setcounter{figure}{0}
\renewcommand{\thesection}{S.\arabic{section}}
\renewcommand{\theequation}{E.\arabic{equation}}
\renewcommand{\thefigure}{A.\arabic{figure}}
\setcounter{page}{1}
  \begin{center}
  \Large {\bf Supplement to ``Doubly robust calibration of prediction sets under covariate shift''}
  \end{center}
       
\begin{abstract}
This supplement contains the proofs to all the main results in the paper and some supporting lemmas. 
\end{abstract}

\section{A review of conformal inference with non-exchangeable data}\label{sec:non-exchangeable}
\cite{politis2015model} considered a transformation based approach that transforms non-i.i.d. data to i.i.d. data and applies i.i.d. data prediction sets on the transformations. For example, in an AR(1) model, transform the observed data to obtain estimated innovations which are assumed i.i.d. and predict the future innovation, then observed data forecast is obtained using the AR(1) model estimates and the prediction of future innovation. \cite{chernozhukov2018exact} dealt with time series data where they developed a randomization method by including block structure in the permutation scheme and showed asymptotic validity under some modeling assumptions on the conformity score when exchangeability fails, see also \cite{chernozhukov2021distributional} and \cite{chernozhukov2021exact} from the same authors under more general settings. \cite{cattaneo2021prediction} recently described a different approach to obtain prediction regions with time series data in a synthetic control framework. Recently,~\cite{oliveira2022split} proved that the split conformal prediction methodology retains asymptotic coverage guarantee for several dependent data settings.

In a separate strand of work, \cite{gibbs2021adaptive} developed an adaptive approach based on ideas from conformal inference that builds  predictions sets in an online setting where the data generating distribution is allowed to vary over time and established coverage validity over the long term; see also  follow-up work in \cite{zaffran2022adaptive} for an adaptive, tuning-free method. These works partly combine ideas from online learning and sequential prediction literature.

 \cite{candes2021conformalized} developed a method for survival analysis subject to administrative censoring that has approximate marginal coverage if the censoring mechanism or the conditional survival function is estimated well, and \cite{teng2021t} focused on a similar censoring scenario under a Cox proportional hazards model under the strong ignorability condition. 
A recent paper \cite{barber2022conformal} designed a new technique for non-exchangeable data that does not treat data points symmetrically and is robust against distribution shift. In almost all of these cases, the coverage guarantee is attained asymptotically as the number of training samples diverges to infinity.

\section{Connection to Causal Inference}\label{sec:notation-causal}
In this section, we briefly discuss how the goal of prediction in the covariate shift setting connects with that of prediction of counterfactuals and individual treatment effects in a potential outcome framework~\citep{rubin1974estimating,splawa1990application}. \cite{lei2020conformal} were the  first to formally draw this connection. Given $N$ subjects, let $A_{i} \in\{0,1\}$ denote a binary treatment indicator, 
$(Y_{i}(1), Y_{i}(0) )$ be the pair of potential outcomes for unit $i$, and $X_{i}$ be the corresponding vector of measured covariates needed to control of confounding. We assume that
$$
\bigl( Y_{i}(1), Y_{i}(0), A_{i}, X_{i} \bigr) \stackrel{\text { i.i.d. }}{\sim}(Y(1), Y(0), A, X),
$$
where $(Y(1), Y(0), A, X)$ is a random vector. Under the stable unit treatment value assumption (SUTVA) commonly assumed in the literature (see e.g. \cite{rubin1990formal}), the observed dataset consists of triples $(Y_{i}^{\text {obs}}, A_{i}, X_{i} )$ where
$$
Y_{i}^{\mathrm{obs}}=
\begin{cases}
Y_i(1), & A_i = 1,\\
Y_i(0), & A_i = 0.
\end{cases}
$$
For regularity, we assume the distributions of $X | A=1$ and $X| A=0$ are absolutely continuous with respect to each other.
In this counterfactual setting, we wish to predict the individual treatment effect (ITE) $\tau_{i}:= Y_{i}(1)-Y_{i}(0)$, which cannot be observed because for every unit only one potential outcome is observed while the other is missing. In order to obtain such prediction, we make the standard unconfoundedness assumption (also known as strong ignorability condition) that $(Y(1), Y(0) ) \perp A \,|\, X$, the counterfactual analog of MAR. Our approach thus yields prediction intervals for ITE with valid coverage for subjects in the study, for whom one potential outcome is observed; below, we briefly also discuss how one might obtain prediction intervals for subjects not in the study, for whom both potential outcomes are missing, but a covariate $X$ is available.

For any treated unit $i$ in the study, i.e. with $A_i = 1$,  we construct a prediction interval $\widehat{C}^{\mathrm{ITE}}_i$ for $\tau_i$ such that $\widehat{C}^{\mathrm{ITE}}_i  = Y_{i}^{\mathrm{obs}} - \widehat{C}_0(X_i)$, where $\widehat{C}_0(x)$ satisfies 
\begin{equation}\label{eq:treated}
\mathbb{P}\bigl(Y(0) \in \hat{C}_{0}(X) \mid A=1\bigr) \geq 1-\alpha;
\end{equation}
Similarly, for any untreated unit, we construct $\widehat{C}^{\mathrm{ITE}}_i  = \widehat{C}_1(X_i) - Y_{i}^{\mathrm{obs}}$, where $\widehat{C}_1(x)$ satisfies
\begin{equation}\label{eq:untreated}
\mathbb{P}\bigl(Y(1) \in \hat{C}_
{1}(X) \mid A=0\bigr) \geq 1-\alpha.
\end{equation}
Thus, such construction has a guaranteed coverage for $\tau_i$ because 
\begin{align*}
    \PP \bigl(Y_i (1) - Y_i (0) \in   \widehat{C}^{\mathrm{ITE}}_i  \bigr) & = \PP(A_i =1) \PP(Y_i (0) \in \widehat{C}_0 (X_i) | A_i = 1 ) + \PP(A_i =0) \PP(Y_i (1) \in \widehat{C}_1 (X_i) | A_i = 0 ) \\
    & \geq (1-\alpha) \bigl(\PP(A_i = 1) + \PP(A_i = 0) \bigr) = 1- \alpha, \quad \text{ if } \eqref{eq:treated} \text{ and } \eqref{eq:untreated} \text{ both hold.}
\end{align*}
A more general goal than \eqref{eq:treated} would be to build a prediction set $\widehat{C}_0$ such that
\begin{equation}\label{eq:causal-general-q}
\mathbb{P}_{(X, Y(0)) \sim Q_{X} \times P_{Y(0) \mid X}}\left(Y(0) \in \hat{C}_{0}(X)\right) \geq 1-\alpha,
\end{equation}
where $Q_X$ is some distribution for $X$. Note that \eqref{eq:treated} can be seen as a special case of \eqref{eq:causal-general-q} with $Q_X = P_{X|A=1}$. Based on the untreated samples, we learn the distribution of $Y(0) \mid X, A= 0$. And by the unconfoundedness assumption, this has the same distribution as $Y(0) \mid X, A= 1$, and also $Y(0) \mid X$.

A more challenging goal potentially of interest in many settings, might be to obtain prediction sets for future subjects not in the study, for whom neither $Y_i(1)$ nor $Y_i(0)$ is observed, but $X_i$ is available. Without any knowledge of the relationship between the distributions of $Y(1)$ and $Y(0)$, a simple approach would be to obtain a pair of prediction intervals at level $1- \alpha/2$, namely $(\widehat{Y}^L_1(x), \widehat{Y}^R_1(x))$ for $Y(1)$ and $(\widehat{Y}^L_0(x), \widehat{Y}^R_0(x))$ for $Y(0)$. Then taking the difference of the two sets such that $\widehat{C}^\mathrm{ITE} = \bigl( \widehat{Y}^L_1(x) - \widehat{Y}^R_0(x), \widehat{Y}^R_1(x) - \widehat{Y}^L_0(x) \bigr)$ would in principle yield a valid $(1-\alpha)$--coverage for $Y(1) - Y(0)$. See, e.g. Section 4.1 of \cite{lei2020conformal}. We also note that \cite{fan2007usharp} has provided sharp bounds on the distribution of the treatment effect which might aid in building a more precise prediction set. 


\section{Semiparametric theory and influence functions}\label{sec:if}
In this section, we provide the definition of an influence function in the literature of semi-parametrics theory and give the derivation of \eqref{eq:if}.

\begin{definition}\label{def:if} Given a semiparametric model $\mathcal{F}$, a law ${F}^{*}$ in $\mathcal{F}$, and a class $\mathcal{A}$ of reg. parametric submodels of $\mathcal{F}$, a real valued functional
$$
\theta: \mathcal{F} \rightarrow \mathbb{R}
$$
is said to be a {\bf pathwise differentiable} or regular parameter at F* wrt $\mathcal{A}$ in model $\mathcal{F}$ iff there exists $\psi_{{F}^ *}({x})$ in $\mathcal{L}_{2}(\dot{F}^{*})$ such that for each submodel in $\mathcal{A}$, say indexed by $t$ and with ${F}^{*}={F}_{t^{*}}$, and score, say $S_{t}\left(t^{*}\right)=s_{t}\left(X ; t^{*}\right)$ at $t^{*}$, it holds that
$$
\left.\frac{\partial}{\partial t} \theta\left(F_{t}\right)\right|_{t=t^{*}}=\E_{F^{*}}\left[\psi_{F^{*}}(X) S_{t}\left(t^{*}\right)\right]
$$
$\psi_{{F}^*}(.)$ is called a {\bf{gradient}} of $\theta$ at ${F}^{*}$ (wrt $\left.\mathcal{A}\right)$.
If, in addition, $\psi_{{F}^*}( {X})$ has mean zero under $\mathrm{F}^{*}, \psi_{{F}^*}( {X})$ is most commonly referred to as an {\bf{influence function}} of the functional $\theta$ at $F^{*}$.
\end{definition}

In our case of finding the $(1-\alpha)$-quantile for $R|T=1$, let $t$ denote the index for the parametric submodels, $\theta(t)$ be the desired $(1-\alpha)$-quantile for $R|T=1$ with $\theta(t^*):=r_\alpha$. Let $u(R; \theta):= \mathbbm{1}(R \leq \theta | T=1)-(1-\alpha)$ with $\theta(t^*)$ satisfies that $\E_{t}[u(R; \theta(t))] =0$.
Then, 
\begin{align}
    0 &= \frac{\partial}{\partial t} \E_{t}\bigl[u(R;\theta(t)) \bigr] \bigr|_{t=t^*}\\
    &=\frac{\partial}{\partial t} \E_{t}\bigl[u(R;\theta(t^*))\bigr] \bigr|_{t=t^*} + \frac{\partial}{\partial \theta} \E_{t^*}[u(R;\theta)] |_{\theta=\theta(t^*)} \frac{\partial \theta(t)}{\partial t} \biggr|_{t=t^*}.
\end{align}
From this we conclude that 
\begin{align}
 \frac{\partial  \theta(t)}{\partial t} \, \biggr|_{t = t^*} &= -\biggl\{ \frac{\partial}{\partial \theta} \E_{t^*}  \bigl[ u(R;\theta) \bigr] \bigr|_{\theta=\theta(t^*)} \biggr\}^{-1} \frac{\partial}{\partial t} \E_{t}\bigl[u(R; \theta(t^*))\bigr] \bigr|_{t=t^*}\\
  &\propto  \frac{\partial}{\partial t} \E_t \biggl\{ \E_t \Big( \mathbbm{1}\{ R \leq \theta \bigr\} \mid T=1,X \Bigr) \Bigr| T=1 \biggr\} \, \biggr|_{t^*}  \\
  &= \frac{\partial}{\partial t} \E_t \biggl\{ \E_t \Big( \mathbbm{1}\{ R \leq \theta \bigr\} \mid T=0,X \Bigr) \Bigr| T=1 \biggr\} \, \biggr|_{t^*}  \\
     &=  \frac{\partial}{\partial t} \E_t \biggl\{ \E_{t^*} \Bigl(   \mathbbm{1}\{ R \leq \theta \bigr\} \mid T=0,X \Bigr) \Bigm| T=1  \biggr\} \, \biggr|_{t^*} \\
    & \qquad  +\frac{\partial}{\partial t} \E_{t^*} \biggl\{ \E_{t} \Bigl(   \mathbbm{1}\{ R \leq \theta \bigr\} \mid T=0,X \Bigr) \Bigm| T=1  \biggr\} \, \biggr|_{t^*} \\
    &=: \rom{1} + \rom{2}.
\end{align}

Let $S(Z):=S((1-T)Y, T,X)=\mathbbm{1}\{T=0\} S(Y|T=0,X)+S(T|X)+S(X)$ be the score vector for all the observed data. We analyze the two terms separately. For the first term, because $R \perp T | X$, we have that
\begin{align}
    \rom{1}&= \E_{t^*} \biggl[   E_{t^*} \bigl( \mathbbm{1}\{ R \leq \theta \bigr\} \mid T=0,X \bigr) S_{X|T=1 } \Bigm| T=1  \biggr]\\
    &= \E_{t^*} \biggl[  \frac{ \mathbbm{1}\{T=1\} }{\PP(T=1)}    E_{t^*} \bigl( \mathbbm{1}\{ R \leq \theta \bigr\} \mid T=0,X \bigr) S_{X|T=1 }  \biggr]\\
        &= \E_{t^*} \biggl[  \frac{ \mathbbm{1}\{T=1\} }{\PP(T=1)}    \PP( R \leq \theta \mid X ) S_{X|T }  \biggr]\\
                &= \E_{t^*} \biggl[  \frac{ \mathbbm{1}\{T=1\} }{\PP(T=1)}    \big\{ \PP( R \leq \theta \mid X ) - \E( \PP( R \leq \theta \mid X ) | T)  \bigr\} S_{X|T }  \biggr]\\
                &= \E_{t^*} \biggl[  \frac{ \mathbbm{1}\{T=1\} }{\PP(T=1)}    \big\{ \PP( R \leq \theta \mid X ) - \E( \PP( R \leq \theta \mid X ) | T)  \bigr\} S_{X,T }  \biggr]\\
                 &= \E_{t^*} \biggl[  \frac{ \mathbbm{1}\{T=1\} }{\PP(T=1)}    \big\{ \PP( R \leq \theta \mid X ) - \E( \PP( R \leq \theta \mid X ) | T)  \bigr\} S_{Z }  \biggr]\\
    &= \E_{t^*} \biggl[  \mathbbm{1}\{T=1\}  \Bigl\{ \E_{t^*} \bigl( \mathbbm{1}\{ R \leq \theta \bigr\} \mid X \bigr) - (1-\alpha)  \Bigr\} S(Z)  \biggr]/\PP(T=1),
\end{align}
For the second term, for brevity let $A:=\mathbbm{1}\{R \leq \theta\}$,
\begin{equation*}
\begin{aligned}
\rom{2} &= \frac{\partial}{\partial t} \E_{t^*} \biggl\{  \frac{ \mathbbm{1}\{T=1\} }{\PP(T=1)}   \E_{t} \Bigl(   \mathbbm{1}\{ R \leq \theta \bigr\} \mid T=0,X \Bigr)   \biggr\} \, \biggr|_{t^*}\\
&= \E_{t^*} \biggl\{  \frac{ \mathbbm{1}\{T=1\} }{\PP(T=1)}   \E_{t^*} \Bigl(  A S_{A | T=0,X}\mid T=0,X \Bigr)   \biggr\} \\
&= \E_{t^*} \biggl\{  \frac{ \mathbbm{1}\{T=1\} }{\PP(T=1)}   \E_{t^*} \Bigl(  A S_{A | T=0,X}(A) \Bigr| T=0,X \Bigr)   \biggr\} \, \\
&= \E_{t^*} \biggl[ \E_{t^*} \biggl\{ A  S_{A|T=0,X} (A,T=0,X) \Bigr| T=0,X  \biggr\}  \biggm| T=1      \biggr]\\
&= \E_{t^*} \biggl[ \E_{t^*} \biggl\{ \frac{A \mathbbm{1}\{T=0\} }{\PP(T=0 | X)}S_{A|T=0,X} (A,T=0,X) \Bigl| X  \biggr\}  \biggm| T=1     \biggr]\\
&= \E_{t^*} \biggl[ \E_{t^*} \biggl\{ \frac{A \mathbbm{1}\{T=0\} }{\PP(T=0 | X)}S_{A|T,X} (A,T,X) \Bigl| X  \biggr\}  \biggm| T=1     \biggr]\\
&= \int \int \int \frac{A \mathbbm{1}\{T=0\} }{\PP(T=0 | X=x)}S_{A|T,x} (r,T,X) f(r |T,x) f(T|x) f(x|T=1) \mathrm{d}r \mathrm{d}x \mathrm{d}T   \\
&=  \int \int \int \frac{A \mathbbm{1}\{T=0\} \PP(T=1 | X=x) }{\PP(T=1)\PP(T=0 | X=x)}S_{A|T,X} (r,T,X) f(r |T,X)  f(T|x) f(x) \mathrm{d}r \mathrm{d}x \mathrm{d}T  \\
&=  \E\Bigl[A \mathbbm{1}\{T=0\}\pi^{\star}(X) S_{A|T,X} (r,T,X) \Bigr] /\PP(T=1) \\
&=  \E\Bigl[ (A - \E(A|T=0,X)) \mathbbm{1}\{T=0\}\pi^{\star}(X)  S_{A|T,X} (R,T,X) \Bigr] /\PP(T=1) \\
&=\E\Bigl[ \bigl(\mathbbm{1} \{R \leq \theta \} - m^{\star}(t, X) \bigr) \mathbbm{1}\{T=0\}\pi^{\star}(X)  \bigl\{  \mathbbm{1}\{T=0\} S_{A|T,X} (R,T,X) + S_{T,X}(T,X)  \bigr\} \Bigr]/\PP(T=1) \\
&=\E\Bigl[ \bigl(\mathbbm{1} \{R \leq \theta \} - m^{\star}(t, X) \bigr) \mathbbm{1}\{T=0\} \pi^{\star}(X) S(Z) \Bigr] /\PP(T=1).
\end{aligned}
\end{equation*}
Combining the two terms together gives us 
\begin{align}
     \frac{\partial }{\partial t} \theta(t) \, \biggr|_{t^*}
     &\propto \E_{t^*} \biggl[ \Bigl\{ \mathbbm{1}\{T=1\}  \Bigl( m^{\star}(\theta, X) - (1-\alpha) \Bigr)  \\
     &\,\,+ \mathbbm{1}\{T=0\}\pi^\star(x)  \Bigl(\mathbbm{1} \{R \leq \theta \} - m^{\star}(\theta, X) \Bigr)\Bigr\}    S(Z)) \biggr].
\end{align}
We therefore conclude that the following function $\phi(\cdot)$ is proportional to a mean zero gradient of $\theta(t)$ at $t^*$,
\begin{equation}
\begin{aligned}
   \phi(\theta, X,R,T; \pi^{\star}, m^{\star})
   &= \mathbbm{1}\{T = 0\}\pi^{\star}(X)\Big[\mathbbm{1}\{R \le \theta\} - m^{\star}(\theta, X)\Big] + \mathbbm{1}\{T=1\}  \Big[m^{\star}(\theta, X) - (1 - \alpha)\Big]. 
\end{aligned}
\end{equation}

\section{Proof of Lemma~\ref{lem:connection-if-coverage}}\label{sec:proof-of-connection-if-coverage}
\begin{proof}[Proof of Lemma~\ref{lem:connection-if-coverage}]
Throughout the proof, we will write $R$ instead of $R(X, Y)$ for convenience.
Firstly, note that
\begin{equation}\label{eq:coverage}
\begin{split}
\mathbb{P}_{(X, Y) \sim Q_{X} \times P_{Y \mid X}}( R(X,Y) \leq \theta \big|\theta) &= \E_{X \sim Q_X} \bigl[ \PP(R \leq \theta \mid X, \theta)\big|\theta\bigr]\\
&=\E_{X \sim Q_X} \bigl[ \PP(R \leq \theta  \mid X,T=1,\theta)\big|\theta \bigr]\\
&=\int_{\chi} \int_{-\infty}^\theta p(x | T=1) p(r | x,T=1) \mathrm{d}r \mathrm{d}x\\
&=\int_{\chi} \int_{-\infty}^\theta p(x | T=1) p(r | x,T=0) \mathrm{d}r \mathrm{d}x \text{  (Using~\eqref{eq:conditional-dist-same})}\\
&=\int_{\chi} \int_{-\infty}^\theta \frac{p(x | T=1)}{p(x | T=0)} p(r | x,T=0) p(x | T=0) \mathrm{d}r \mathrm{d}x\\
&=\E \biggl[ \frac{p(X | T=1)}{p(X | T=0)}  \mathbbm{1}\{R \leq \theta  \} \,\Bigr|\, T=0, \theta  \biggr]. 
\end{split}
\end{equation}
On the other hand, we can prove that
\begin{equation}\label{eq:P-IF-simplification}
P[\mathrm{IF}(\theta, X,R,T;\pi,m)] = \PP(T=1) \biggl\{ \E \biggl[ \frac{p(X | T=1)}{p(X | T=0)}  \mathbbm{1}\{R \leq \theta  \} \,\Bigr|\, T=0, \theta  \biggr] - (1-\alpha )\biggr\}.
\end{equation}
This combined with~\eqref{eq:coverage} completes the proof of~\eqref{eq:connection-if-coverage}. We will prove~\eqref{eq:P-IF-simplification} in two steps.
\begin{equation}\label{eq:step-1-IF-simplify}
P[\mathrm{IF}(\theta, X,R,T;\pi,m)] = \E \Bigl[  \PP(T=1 |X) \bigl\{ \PP (R \leq \theta | X, \theta ) -  (1-\alpha) \bigr\}\big|\theta\Bigr],
\tag{Step 1}
\end{equation}
and
\begin{equation}\label{eq:step-2-IF-simplify}
\begin{split}
&\E \bigl[  \PP(T=1 |X) \bigl\{ \PP (R \leq \theta | X, \theta ) -  (1-\alpha) \bigr\}\big|\theta\bigr]\\ 
&\quad= \PP(T=1) \biggl\{ \E \biggl[ \frac{p(X | T=1)}{p(X | T=0)}  \mathbbm{1}\{R \leq \theta  \} \,\Bigr|\, T=0, \theta  \biggr] - (1-\alpha )\biggr\}.
\end{split}
\tag{Step 2}
\end{equation}
In the proof of~\eqref{eq:step-1-IF-simplify}, we will use the fact that either $\pi(\cdot)$ or $m(\cdot, \cdot)$ represents the correct density ratio or the correct conditional distribution function. The proof of~\eqref{eq:step-2-IF-simplify} follows essentially from Bayes rule.
\paragraph{Proof of~\eqref{eq:step-1-IF-simplify}.}
If $\pi(x) = \mathbb{P}(T = 1|X = x)/\mathbb{P}(T = 0|X = x)$ for all $x$ (i.e., density ratio is correct), then we have
\[
\mathbb{E}\left[\mathbbm{1}\{T = 0\}\pi(X)\big|X = x, R\right] = \mathbb{P}(T = 0|X = x, R)\frac{\mathbb{P}(T = 1|X = x)}{\mathbb{P}(T = 0|X = x)} = \mathbb{P}(T = 1|X = x).
\]
The second equality here follows because $T$ is independent of $R$ given $X$. This implies that
\begin{align*}
&\mathbb{E}\left[\mathbbm{1}\{T = 0\}\pi(X)\big\{\mathbbm{1}\{R \le \theta\} - m(\theta, X)\big\} \big| \theta \right]\\ 
&\quad= \mathbb{E}\Bigl[\mathbb{P}(T = 1|X)\big\{\mathbbm{1}\{R \le \theta \} - m(\theta, X)\big\} \big| \theta \Bigr]\\
&\quad= \mathbb{E}\Bigl[\mathbb{P}(T = 1|X) \bigl\{\mathbb{P}(R \le \theta|X, \theta) - m(\theta, X) \bigr\} \big| \theta \Bigr].
\end{align*}
Similarly, 
\[
\mathbb{E}\Bigl[\mathbbm{1}\{T = 1\}\left\{m(\theta, X) - (1 - \alpha)\right\} \big| \theta \Bigr] = \mathbb{E}\Bigl[\mathbb{P}(T = 1|X)\bigl\{m(\theta, X) - (1 - \alpha)\bigr\} \big| \theta \Bigr].
\]
Hence, if $\pi(\cdot)$ is the true density ratio, then
\begin{align*}
P [\mathrm{IF}(\theta, X, R, T; \pi, m)] &= \E \Bigl[\PP (T = 1|X) \bigl\{ \PP (R \le \theta |X,\theta) - (1 - \alpha) \bigr\} \big| \theta \Bigr].
\end{align*}
This completes the proof of~\eqref{eq:step-1-IF-simplify} when $\pi(\cdot)$ is the true density ratio.

{If $m(\gamma, x) = \mathbb{E}[\mathbbm{1}\{R \le \gamma\}|X = x]$ for all $\gamma\in\mathbb{R}$, $x\in\rchi$ (i.e., the conditional mean is correct)}, then using the conditional independence of $R$ and $T$ given $X$, we have 
\[
\mathbb{E} \Bigl[\mathbbm{1}\{T = 0\}\pi(X)\bigl\{ \mathbbm{1}\{R \le \theta\} - m(\theta, X) \bigr\} \big| \theta \Bigr] = 0.
\]
Hence,
\begin{align*}
P[\mathrm{IF}(\theta , X, R, T; \pi, m)] &= \mathbb{E}\Bigl[\mathbbm{1}\{T = 1\}\bigl\{ m(\theta , X) - (1 - \alpha)\bigr\} \big| \theta \Bigr]\\ 
&= \mathbb{E}\Bigl[\mathbb{P}(T = 1|X)\bigl\{m(\theta ,X) - (1 - \alpha)\bigr\} \big| \theta \Bigr]\\
&= \mathbb{E}\Bigl[\mathbb{P}(T = 1|X)\bigl\{\mathbb{P}(R \le \theta|X, \theta) - (1 - \alpha)\bigr\}\big|\theta\Bigr].
\end{align*}
This completes the proof of~\eqref{eq:step-1-IF-simplify} if $m(\cdot, \cdot)$ is the true conditional mean function.

\paragraph{Proof of~\eqref{eq:step-2-IF-simplify}.} 

\begin{equation}\label{eq:pif}
\begin{split}
&\E \bigl[  \PP(T=1 |X) \bigl\{ \PP (R \leq \theta | X, \theta ) -  (1-\alpha) \bigr\}\big|\theta\bigr]\\
    &\quad= \E \bigl[  \PP(T=1 |X)  \PP (R \leq \theta | X, \theta) \big|\theta \bigr] - \PP(T=1) (1-\alpha)  \\
    &\quad= \E \biggl[ \mathbbm{1}\{T=0\} \frac{\PP(T=1 |X)}{\PP(T=0 |X)}  \PP (R \leq \theta | X, \theta)\big|\theta \biggr] - \PP(T=1) (1-\alpha)  \\
    &\quad= \E \biggl[  \mathbbm{1}\{T=0\} \frac{\PP(T=1 |X)}{\PP(T=0 |X)}  \mathbbm{1}\{R \leq \theta  \} \big|\theta\biggr]  - \PP(T=1) (1-\alpha)  \\
    &\quad\stackrel{(b)}{=} \frac{\PP(T=1)}{\PP(T=0)} \E \biggl[  \mathbbm{1}\{T=0\} \frac{p(X | T=1)}{p(X | T=0)}  \mathbbm{1}\{R \leq \theta  \}  \bigl| \theta \biggr]  - \PP(T=1) (1-\alpha)  \\
    &\quad= \frac{\PP(T=1)}{\PP(T=0)} \E_{T} \biggl\{ \E \biggl[  \mathbbm{1}\{T=0\} \frac{p(X | T=1)}{p(X | T=0)}  \mathbbm{1}\{R \leq \theta  \}  \Big| T,\theta \biggr] \biggr\}  - \PP(T=1)(1-\alpha)  \\
    &\quad= \frac{\PP(T=1)}{\PP(T=0)} \PP(T=0) \E \biggl[ \frac{p(X | T=1)}{p(X | T=0)}  \mathbbm{1}\{R \leq \theta  \}  \Big| T=0,\theta \biggr] - \PP(T=1)(1-\alpha)  \\
    &\quad= \PP(T=1) \biggl\{ \E \biggl[ \frac{p(X | T=1)}{p(X | T=0)}  \mathbbm{1}\{R \leq \theta  \} \,\Bigr|\, T=0,\theta  \biggr] - (1-\alpha )\biggr\}, 
    \end{split}
\end{equation}
where equality (b) comes from Bayes rule. 
This completes the proof of~\eqref{eq:step-2-IF-simplify}.
\end{proof}

\section{Proof of Theorem~\ref{thm:product-bias}}\label{appsec:proof-product-bias}
\begin{proof}
By definition of the $\mathrm{IF}$ function, it holds $\forall \gamma \in \mathbb{R}$,
\begin{equation}\label{eq:pf-product-bias}
\begin{aligned}
    P \bigl[ \mathrm{IF}(\gamma, X, R,{\widehat\pi},\widehat{m} )\bigr]  &= P \Bigl[ \PP(T=0 |X) {\widehat\pi}(X) \bigr\{  \PP(R\leq \gamma |X) - \widehat{m}(\gamma, X) \bigr\}\\
    &\qquad + \PP (T=1 | X)  \bigl\{ \widehat{m}(\gamma,X)-(1-\alpha) \bigr\} \Bigr] \\
    &= P \Bigl[ \PP(T=0 |X) \bigl\{ \widehat\pi(X) - \pi^\star(X) \bigr\} \bigr\{ m^\star(\gamma,X) - \widehat{m}(\gamma,X) \bigr\} \Bigr]\\
    &\qquad + P \Bigl[ \PP (T=0 | X) \pi^\star(X)  \bigl\{ m^\star(\gamma, X) - \widehat{m}(\gamma, X)) \bigr\}  \\
    &\qquad + \PP (T=1 | X)  \bigl\{ \widehat{m}(\gamma, X)-(1-\alpha) \bigr\} \Bigr] \\
    &= P \Bigl[ \PP(T=0 |X) \bigl\{ \widehat\pi(X) - \pi^\star(X) \bigr\} \bigr\{ m^\star(\gamma,X) - \widehat{m}(\gamma, X) \bigr\} \Bigr]\\
    &\qquad + P \Bigl[ \PP (T=1 | X)   \bigl\{ m^\star(\gamma, X) - \widehat{m}(\gamma, X)) \bigr\}  \\
    &\qquad + \PP (T=1 | X)  \bigl\{ \widehat{m}(\gamma, X)-(1-\alpha) \bigr\} \Bigr] \\
    &= P \Bigl[ \PP(T=0 |X) \bigl\{\widehat\pi(X) - \pi^\star(X) \bigr\} \bigr\{ m^\star(\gamma,X) - \widehat{m}(\gamma, X) \bigr\} \Bigr]\\
    &\qquad + P \Bigl[ \PP (T=1 | X)   \bigl\{ m^\star(\gamma, X) - (1-\alpha) \bigr\} \Bigr].
\end{aligned}
\end{equation}

Repeating the same calculation in \eqref{eq:pf-product-bias} with either $\widehat\pi$ replaced by $\pi^\star$ and $\widehat{m}$ replaced by $m^\star$ yields $\forall \gamma \in \mathbb{R}$,
\begin{equation}
\begin{aligned}
    P \bigl[ \mathrm{IF}(\gamma, X, R,\pi,m )\bigr]  &=P \Bigl[ \PP (T=1 | X)   \bigl\{ m^\star(\gamma, X) - (1-\alpha) \bigr\} \Bigr].
\end{aligned}
\end{equation}
Therefore,
\begin{equation}
\begin{aligned}
    &\sup_{\gamma \in \mathbb{R}} \Bigl| P \bigl[  \mathrm{IF}(\gamma, X, R,\widehat{\pi},\widehat{m} ) - \mathrm{IF}(\gamma, X, R,{\pi^{\star}},{m^{\star}} ) \bigr] \Bigr|\\ &\quad= \sup_{\gamma \in \mathbb{R}}  \Bigl| P \bigl[ \PP(T=0 |X) \bigl\{ \pi^\star(X) - \widehat{\pi}(X) \bigr\} \bigr\{ m^\star(\gamma,X) - \widehat{m}(\gamma, X) \bigr\} \bigr] \Bigr| \\
    &\quad\leq \|\widehat{\pi} - \pi^\star \|_2 \sup_{\gamma} \| \widehat{m}(\gamma,\cdot) - m^\star(\gamma,\cdot)\|_2.
\end{aligned}
\end{equation}
The last inequality here follows from Cauchy--Schwarz inequality.  
\end{proof}

\section{Proof of Theorem~\ref{thm:convergence-if-function}}\label{appsec:proof-convergence-if-funcion}
\begin{proof}
Without loss of generality assume the indexes in $\mathcal{I}_2$ is $1, \dots, n$ with $n:= |\mathcal{I}_2| $, and we expand $\mathbb{P}_{\mathcal{I}_2} \bigl[ \mathrm{IF}(\theta,\mathcal{D}_2; \widehat{\pi},\widehat{m}) \bigr] -P \bigl[ \mathrm{IF}(\theta, X, R,T;\widehat\pi,\widehat{m} ) \bigr] 
$ into three parts,
\begin{align}
&\mathbb{P}_{\mathcal{I}_2} \bigl[ \mathrm{IF}(\theta,\mathcal{D}_2; \widehat{\pi},\widehat{m}) \bigr] - P \bigl[ \mathrm{IF}(\theta, X, R,T;\widehat\pi,\widehat{m} ) \bigr]\\
     = \;& \frac{1}{n}\sum_{i=1}^n \mathbbm{1}\{t_i=0\} \widehat{\pi}(x_i)  \mathbbm{1}\{r_i \leq \theta \} - P \bigl[ \mathbbm{1}\{t_i=0\} \widehat{\pi}(x_i)  \mathbbm{1}\{r_i \leq \theta \}  \bigr]\\
 & \quad +\frac{1}{n} \sum_{i=1}^n  \widehat{m}(\theta,x_i) \bigl( \mathbbm{1}\{t_i=1\} - \mathbbm{1}\{t_i=0\} \bigr) - P \bigl[\widehat{m}(\theta,x_i) (\mathbbm{1}\{t_i=1\} - \mathbbm{1}\{t_i=0\}) \bigr]\\
 & \quad - (1-\alpha) \biggl[ \frac{1}{n} \sum_{i=1}^n \mathbbm{1}\{t_i=1\} - \PP(T=1) \biggr] \\
=&:\; \mathcal{R}_1(\theta)+\mathcal{R}_2(\theta)+\mathcal{R}_3,
\label{eq:convergence-if-three-parts}
\end{align}
where the three terms will be controlled separately. In particular, $\sup_\theta |\mathcal{R}_1(\theta)| $ and $\sup_\theta |\mathcal{R}_2(\theta)|$ will be bounded using tools from the empirical processes theory. Also notice that conditional on training data $\mathcal{D}_1$, $\widehat{\pi}$ and $\widehat{m}$ are non-random functions and for ease of notation, we treat these functions as non-random and omit the conditioning part from this point onwards. 

For $W_i = (X_i, Y_i, T_i), i\in \mathcal{I}_2$ and any function $f:\mathbb{R}^{d+2}\to\mathbb{R}$,  for notation simplicity we define
\[
\mathbb{G}_{n} f ~:=~ \frac{1}{\sqrt{n}}\sum_{i=1}^n \bigl\{f(W_i) - \mathbb{E}[f(W_i)] \bigr\}.
\]

\paragraph{Bound on $\sup_\theta |\mathcal{R}_1(\theta)|$:}
We have a class of functions 
$$\mathcal{F}=\bigl\{f: f_{\widehat{\pi},\theta}(w)=\mathbbm{1}\{t=1 \} \widehat{\pi}(x) \mathbbm{1}{\{ R(x,y) \leq \theta\} } , \forall \widehat{\pi}(\cdot) \in \mathcal{F}_\pi \bigr\}.$$ Notice that $\forall \theta \in \mathbb{R}$ and $\widehat{\pi}(\cdot) \in \mathcal{F}_\pi$, we have $|f_\theta (w)| \leq \pi_0 \mathbbm{1}\{t=1\}$. Therefore, $F(w) := \pi_0 \mathbbm{1}\{t=1\}$ is an envelope function of $\{ f_\theta(\cdot): \theta \in \mathbb{R} \}$. 
Let $\|\cdot\|_{\mathcal{F}}$ denote the supremum norm $\|z\|_{\mathcal{F}}=\sup _{f \in \mathcal{F}}|z(f)|$.

Applying Lemma \ref{lem:expectation-empirical-proc-step} with $s(t,x) = \mathbbm{1}\{t=0\}\widehat{\pi}(x)$ and $h(x,y)=R(x,y)$ gives us
\begin{equation}\label{eq:expectation-r1}
\E  \| \mathbb{G}_n \|_{\mathcal{F} } \leq \mathfrak{C}\pi_0,
\end{equation}
where $\mathfrak{C}$ is a universal constant. Applying McDiarmid's inequality gives us 
\begin{equation}\label{eq:McDiarmid-application}
\PP \bigl( \| \mathbb{G}_n \|_{\mathcal{F}}- \E  \| \mathbb{G}_n \|_{\mathcal{F} } \geq u \bigr) \leq  \exp \biggl(-\frac{2u^2}{ \sum_{i=1}^n c_i^2}\biggr) \leq  \exp \biggl(-\frac{2u^2}{ \sum_{i=1}^n 4\pi_0^2/n}\biggr) = \exp(-u^2/2\pi_0^2),
\end{equation}
where 
\begin{align*}
c_i &:= \sup_{\substack{(x_i, y_i,t_i), (x_i', y_i',t_i'),\\ (x_1, y_1,t_1), \ldots, (x_n, y_n,t_n)}} \sup_{\theta} \sqrt{n} \biggl| \frac{1}{n}\sum_{j=1}^n \mathbbm{1}\{t_j=0\} \widehat{\pi}(x_j)  \mathbbm{1}\{r_j \leq \theta \}  - \frac{1}{n} \sum_{j=1, j\neq i}^n \mathbbm{1}\{t_j=0\} \widehat{\pi}(x_j)  \mathbbm{1}\{r_j \leq \theta \} \\
&\qquad \qquad \qquad -\frac{1}{n}\mathbbm{1}\{t_i'=0\} \widehat{\pi}(x_i')  \mathbbm{1}\{r_i' \leq \theta \}\biggr| \\
& \leq \sup_{(x_i, y_i,t_i), (x_i', y_i',t_i')} \sup_{\theta} \sqrt{n} \biggl| \frac{1}{n} \mathbbm{1}\{t_{i}=0\} \widehat{\pi}(x_{i})  \mathbbm{1}\{r_{i} \leq \theta \}- \frac{1}{n} \mathbbm{1}\{t_i'=0\} \widehat{\pi}(x_i')  \mathbbm{1}\{r_i' \leq \theta \}  \biggr| \le \frac{2 \pi_0}{\sqrt{n}}.
\end{align*}
Substituting the expectation bound~\eqref{eq:expectation-r1} in~\eqref{eq:McDiarmid-application} and setting the right hand side of~\eqref{eq:McDiarmid-application} to $\delta$ yields for another absolute constant $\mathfrak{C}^{\prime}$
\begin{equation}\label{eq:r1}
\PP \biggl(\| \mathbb{G}_n \|_{\mathcal{F}} \geq  \mathfrak{C}^\prime \sqrt{\pi_0^2 + \pi_0^2  \log \bigl(\frac{1}{\delta} \bigr)} \biggr) \leq \PP \biggl[ \| \mathbb{G}_n \|_{\mathcal{F}} \geq \mathfrak{C} \biggl\{ \pi_0 + \pi_0 \sqrt{2\log \bigl( \frac{1}{\delta} \bigr)} \biggr\} \biggr] \leq \delta.
\end{equation}

\paragraph{Bound on $\sup_\theta |\mathcal{R}_2(\theta)|$:}\label{pf:bound-r2} 
We define the class of functions $\mathcal{F}=\bigl\{f: f_\theta(w)=\widehat{m}(\theta,x_i) (\mathbbm{1}\{t_i=1\} - \mathbbm{1}\{t_i=0\}) \bigr\}$ with $F(w) = m_0 (\mathbbm{1}\{t_i=1\} - \mathbbm{1}\{t_i=0\} )$ as its envelope.

Note that 
\begin{equation}
    \begin{aligned}
    \sup_{\theta} | \mathbb{G}_n f | &= \sup_{\theta} \biggl| \mathbb{G}_n  \Bigl[ (\mathbbm{1}\{t_i=1\} - \mathbbm{1}\{t_i=0\}) \widehat{m}(\theta,x_i)   \Bigr] \biggr|\\
    &= \sup_{\theta} \biggl| \mathbb{G}_n  \Bigl[ (\mathbbm{1}\{t_i=1\} - \mathbbm{1}\{t_i=0\}) \int_{0}^{m_0} \mathbbm{1}\{\widehat{m}(\theta,x_i) \geq u \} \mathrm{d}u  \Bigr] \biggr|\\
    &= \sup_{\theta} \biggl|  \int_{0}^{m_0}  \mathbb{G}_n  \Bigl[ (\mathbbm{1}\{t_i=1\} - \mathbbm{1}\{t_i=0\})\mathbbm{1}\{\widehat{m}(\theta,x_i) \geq u \} \mathrm{d}u  \Bigr] \biggr|\\
    &\stackrel{(a)}{=} \sup_{\theta} \biggl|  \int_{0}^{m_0}  \mathbb{G}_n  \Bigl[ (\mathbbm{1}\{t_i=1\} - \mathbbm{1}\{t_i=0\})\mathbbm{1}\{h(x_i,u) \leq 
    \theta \} \mathrm{d}u  \Bigr] \biggr| \text{, for some function }h\\
    &\leq  \int_{0}^{m_0} \sup_{\theta} \biggl|  \mathbb{G}_n  \Bigl[ (\mathbbm{1}\{t_i=1\} - \mathbbm{1}\{t_i=0\})\mathbbm{1}\{h(x_i,u) \leq 
    \theta \} \Bigr] \biggr| \mathrm{d}u  ,
    \end{aligned}
\end{equation}
where equality (a) is from the monotonicity of $\widehat{m}(\theta, x)$ in $\theta$.
Taking the expectation on both sides gives us
\begin{equation}\label{eq:expectation-supremum-monotone}
\E \| \mathbb{G}_n \|_{\mathcal{F}} \leq \int_{0}^{m_0} \E \Bigl[ \sup_{\theta} \bigl|  \mathbb{G}_n  \bigl[ (\mathbbm{1}\{t_i=1\} - \mathbbm{1}\{t_i=0\})\mathbbm{1}\{h(x_i,u) \leq 
    \theta \} \bigr] \bigr| \Bigr] \mathrm{d}u.
\end{equation}

Applying Lemma \ref{lem:expectation-empirical-proc-step} for $s(t,x) = \mathbbm{1}\{t=1\} - \mathbbm{1}\{t=0\}$ gives us for any fixed $u$,
$$
\E \sup_{\theta} \Bigl|  \mathbb{G}_n  \bigl[ (\mathbbm{1}\{t_i=1\} - \mathbbm{1}\{t_i=0\})\mathbbm{1}\{h(x_i,u) \leq 
    \theta \} \bigr] \Bigr| \leq \mathfrak{C},
$$
where $\mathfrak{C}$ is a universal constant. Plugging this back to \eqref{eq:expectation-supremum-monotone} gives us 
\begin{equation}\label{eq:expectation-bound-r2}
\E \| \mathbb{G}_n \|_{\mathcal{F}} \leq \mathfrak{C}m_0.
\end{equation}

Using McDiarmid's inequality we have 
\begin{equation}\label{eq:McDiarmid-application-r2}
\PP ( \| \mathbb{G}_n \|_{\mathcal{F}}- \E  \| \mathbb{G}_n \|_{\mathcal{F} } \geq u ) \leq  \exp \biggl(-\frac{2u^2}{ \sum_{i=1}^n c_i^2}\biggr) \leq  \exp \biggl(-\frac{2u^2}{ \sum_{i=1}^n 4 m_0^2/n}\biggr) = \exp(-u^2/2m_0^2),
\end{equation}
where 
\begin{align*}
c_i &:= \sup_{\substack{(x_i, y_i,t_i), (x_i', y_i',t_i'),\\ (x_1, y_1,t_1), \ldots, (x_n, y_n,t_n)}} \sup_{\theta} \sqrt{n} \biggl| \frac{1}{n}\sum_{j=1}^n \widehat{m}(\theta,x_j) (\mathbbm{1}\{t_j=1\} - \mathbbm{1}\{t_j=0\} ) \\
&\qquad \qquad \qquad - \frac{1}{n} \sum_{j=1, j\neq i}^n \widehat{m}(\theta,x_j) (\mathbbm{1}\{t_j=1\} - \mathbbm{1}\{t_j=0\} ) -\frac{1}{n} \widehat{m}(\theta,x_i') (\mathbbm{1}\{t_i'=1\} - \mathbbm{1}\{t_i'=0\} ) \biggr|  \\
& \leq \sup_{(x_i, y_i,t_i), (x_i', y_i',t_i')} \sup_{\theta} \frac{1}{\sqrt{n} } \biggl|  \widehat{m}(\theta,x_i) (\mathbbm{1}\{t_i=1\} - \mathbbm{1}\{t_i=0\} )- \widehat{m}(\theta,x_i') (\mathbbm{1}\{t_i'=1\} - \mathbbm{1}\{t_i'=0\} )\biggr| \le \frac{2 m_0}{\sqrt{n}}.
\end{align*}
Substituting the expectation bound~\eqref{eq:expectation-bound-r2} in~\eqref{eq:McDiarmid-application-r2} and setting the right hand side of~\eqref{eq:McDiarmid-application-r2} to $\delta$ yields for another absolute constant $\mathfrak{C}^{\prime}$,
\begin{equation}\label{eq:r2}
\PP \biggl(\| \mathbb{G}_n \|_{\mathcal{F}} \geq  \mathfrak{C}^\prime m_0  \sqrt{1 + \log \bigl(\frac{1}{\delta} \bigr)} \biggr) \leq \PP \biggl[ \| \mathbb{G}_n \|_{\mathcal{F}} \geq \mathfrak{C} m_0  \biggl\{ 1 + \sqrt{2\log \bigl( \frac{1}{\delta} \bigr)} \biggr\} \biggr] \leq \delta.
\end{equation}

\paragraph{Bound on $\mathcal{R}_3$:}

Because the random variables in the averaging of $\mathcal{R}_3$ are i.i.d, applying Hoeffding's inequality yields
\begin{equation}\label{eq:r3-hoeffding}
\PP \biggl\{ \frac{1}{n} \sum_{i=1}^{n}\mathbbm{1}\{ t_i=1 \} - \PP(T=1) \geq t \biggr\} \leq \exp \biggl(-\frac{2 t^{2}}{n }\biggr).
\end{equation}
And this leads to
\begin{equation}\label{eq:r3}
\PP \biggl(R_3 \geq (1-\alpha) \sqrt{\frac{1}{2n} \log \bigl( \frac{1}{\delta} \bigr)} \biggr) \leq \delta.
\end{equation}

Combining \eqref{eq:r1}, \eqref{eq:r2} and \eqref{eq:r3} together using the union bound gives the result that for a universal constant $\mathfrak{C}$,
\begin{equation}\label{eq:r1-r2-r3}
\PP \biggl\{ \sup_\theta \bigl| \mathcal{R}_1(\theta)+\mathcal{R}_2(\theta) +\mathcal{R}_3 \bigr| \geq \mathfrak{C} \sqrt{\frac{ (m_0+\pi_0 +1- \alpha)^2 \log \bigl( \frac{1}{\delta}  \bigr) +(m_0+\pi_0)^2 }{n}} \biggr\} \leq \delta.
\end{equation}

\end{proof}

\section{Proof of Theorem~\ref{thm:convergence-if-function-whole}}\label{appsec:proof-convergence-if-funcion-whole}
\begin{proof}
Without loss of generality assume the indexes in $\mathcal{D}^{\mathrm{tr}}$ is $1, \dots, N$ with $N:= |\mathcal{D}^{\mathrm{tr}}| $, and we expand $\mathbb{P}_{N} \bigl[ \mathrm{IF}(\theta,\mathcal{D}^{\mathrm{tr}}; \widehat{\pi},\widehat{m}) \bigr] -P \bigl[ \mathrm{IF}(\theta, X, R,T,\widehat\pi,\widehat{m} ) \bigr] $ into three parts,
\begin{align}
&\mathbb{P}_{N} \bigl[ \mathrm{IF}(\theta,\mathcal{D}^{\mathrm{tr}}; \widehat{\pi},\widehat{m}) \bigr] - P \bigl[ \mathrm{IF}(\theta, X, R,T;\widehat\pi,\widehat{m} ) \bigr] \\
    = \;& \frac{1}{N}\sum_{i=1}^N \mathbbm{1}\{t_i=0\} \widehat{\pi}(x_i)  \mathbbm{1}\{r_i \leq \theta \} - \E \bigl[ \mathbbm{1}\{t_i=0\} \widehat{\pi}(x_i)  \mathbbm{1}\{r_i \leq \theta \}  \bigr]\\
 & \quad +\frac{1}{N} \sum_{i=1}^N \widehat{m}(\theta,x_i) \bigl( \mathbbm{1}\{t_i=1\} - \mathbbm{1}\{t_i=0\} \bigr) - \E \bigl[\widehat{m}(\theta,x_i) (\mathbbm{1}\{t_i=1\} - \mathbbm{1}\{t_i=0\}) \bigr]\\
 & \quad - (1-\alpha) \biggl[ \frac{1}{N} \sum_{i=1}^N \mathbbm{1}\{t_i=1\} - \PP(T=1) \biggr] \\
=&:\; \mathcal{R}_1(\widehat{\pi}, \theta)+\mathcal{R}_2(\widehat{m},\theta)+\mathcal{R}_3,
\label{eq:convergence-if-three-parts-whole}
\end{align}
where the three terms will be controlled separately. In particular, $\sup_{\theta} \mathcal{R}_1(\widehat{\pi},\theta) $ and $\sup_{\theta} \mathcal{R}_2(\widehat{m},\theta)$ will be bounded using tools from empirical processes. 


\paragraph{Bound on $\sup_{\theta} \mathcal{R}_1(\widehat{\pi},\theta)$:}
We have a class of functions $\mathcal{F}=\bigl\{f: f_{\widehat{\pi},\theta}(w)=\mathbbm{1}\{t=1 \} \widehat{\pi}(x) \mathbbm{1}{\{ R(x,y) \leq \theta\} } , \forall \widehat{\pi}(\cdot) \in \mathcal{F}_\pi, \theta \in \mathbb{R} \bigr\}$. Notice that $\forall \theta \in \mathbb{R}$ and $\widehat{\pi}(\cdot) \in \mathcal{F}_\pi$, we have $|f_\theta (w)| \leq \pi_0 \mathbbm{1}\{t=1\}$. Therefore, $F(w) := \pi_0 \mathbbm{1}\{t=1\}$ is an envelope function of $\{ f_{\widehat{\pi},\theta}(\cdot): \widehat{\pi}(\cdot) \in \mathcal{F}_\pi, \theta \in \mathbb{R} \}$. 

Note that because $\widehat{\pi}(\cdot)$ depends on the evaluation data, $\mathcal{R}_1(\widehat{\pi},\theta)$ is not an average of $N$ i.i.d random variables. However we can still bound this term by $\sup_{f \in \mathcal{F}} |\mathbb{G}_N f|$ where for any fixed $f \in \mathcal{F}, \mathbb{G}_N f/\sqrt{N}$ is an average of $N$ i.i.d. random variables with mean 0.

For any pair of functions $f_{\widehat{\pi}_1,\theta_1},f_{\widehat{\pi}_2,\theta_2} \in \mathcal{F}$,
\begin{align}
        \|f_{\widehat{\pi}_1,\theta_1} - f_{\widehat{\pi}_2,\theta_2} \|_{Q} &= \Bigl[ \sum_{i=1}^N  \mathbbm{1}^2 \{t_i=1\} \bigl(\widehat{\pi}_1(x_i)  \mathbbm{1}{\{ r_i \leq \theta_1 \}} - \widehat{\pi}_2(x_i) \mathbbm{1}{\{ r_i \leq \theta_2\} } \bigr)^2 Q(w_i) \Bigr]^{1/2}\\
        &= \frac{ \Bigl[ \sum_{i=1}^N  \mathbbm{1} \{t_i=1\} \bigl(\widehat{\pi}_1(x_i)  \mathbbm{1}{\{ r_i \leq \theta_1 \}} - \widehat{\pi}_2(x_i) \mathbbm{1}{\{ r_i \leq \theta_2\} } \bigr)^2 Q(w_i) \biggr]^{1/2} }{ \bigl[ \sum_{i=1}^N \mathbbm{1}^2\{t_i=1\} \ Q(w_i) \bigr]^{1/2}} \times\\
        &\qquad \qquad \qquad \qquad \qquad \qquad \qquad \qquad \bigl[ \sum_{i=1}^N \mathbbm{1}^2\{t_i=1\} ] Q(w_i) \bigr]^{1/2}\\
        &= \bigl\| \widehat{\pi}_1(x_i) \mathbbm{1}{\{ r_i \leq \theta_1\}} - \widehat{\pi}_2(x_i) \mathbbm{1}{\{ r_i \leq \theta_2\} } \bigr\|_{ \widetilde{Q}} \cdot \bigl[ \sum_{i=1}^N \mathbbm{1}\{t_i=1\} Q(w_i) \bigr]^{1/2}, \label{eq:change-of-measure-whole-1}
\end{align}
where $\widetilde{Q}$ is the new probability measure defined by 
\begin{equation}\label{eq:new-measure-whole}
\widetilde{Q}(w_i):=  \mathbbm{1}\{t_i=1\}  Q(w_i) /  \sum_{i=1}^N \mathbbm{1}\{t_i=1\} Q(w_i), i=1, \dots, N.
\end{equation}
Using $F(w)= \pi_0 \mathbbm{1}\{t=1\}$ as the envelope function of the class $\mathcal{F}$, \eqref{eq:change-of-measure-whole-1} becomes
\begin{equation}\label{eq:distance-relationship-whole}
        \| f_{\widehat{\pi}_1,\theta_1} - f_{\widehat{\pi}_2,\theta_2} \|_{Q} = \bigl\| \widehat{\pi}_1(x_i) \mathbbm{1}{\{ r_i \leq \theta_1\}} - \widehat{\pi}_2(x_i) \mathbbm{1}{\{ r_i \leq \theta_2\} } \bigr\|_{ \widetilde{Q}} \cdot \|F\|_Q / \pi_0.
\end{equation} 
Define a new class of functions $\mathcal{F}_1:= \bigl\{f: f_{\widehat{\pi},\theta}(z) = \widehat{\pi}(x) \mathbbm{1}\{R(x,y) \leq \theta \} \bigr\}$. Then   \eqref{eq:distance-relationship-whole} gives the  relationship between the covering numbers of the two classes $\mathcal{F}$ and $\mathcal{F}_1$,
\begin{equation}\label{eq:covering-change-of-measure-whole}
N\bigl(\varepsilon \| F \|_{Q}, \mathcal{F}, L_{2}(Q)\bigr) \leq    N \bigl(\pi_0 \varepsilon, \mathcal{F}_1, L_{2}(\widetilde{Q}) \bigr).
\end{equation}

We take the envelope function of $\mathcal{F}_1 $  to be $F_1 \equiv \pi_0$ and use Theorem 2.6.7 of \cite{van1996weak} with $r=2$ to get a bound on the covering number bound for $\mathcal{G}=\bigl\{g_\theta: g_\theta(z)= \mathbbm{1}\{R(x,y) \leq \theta \},\forall \theta \in \mathbb{R} \}$. Specifically, for any probability measure $Q$, there exists a universal constant $\mathfrak{C}$ such that
\begin{equation}
    N \bigl(\varepsilon, \mathcal{G}, L_{2}(Q)\bigr) \leq \frac{\mathfrak{C}}{\varepsilon}.
\end{equation}
Let $\mathcal{H}:= \{h:h=\widehat{\pi}(x), \forall \widehat{\pi} \in \mathcal{F}_\pi\}$. The relationship between covering and bracketing numbers gives us
\begin{equation}
N(\varepsilon, \mathcal{H},L_2(Q)) \leq N_{[\,]}(2 \varepsilon, \mathcal{H},L_2(Q)) \leq \exp \bigl( C(2\pi_0 \varepsilon)^{-\alpha_\pi} \bigr).
\end{equation}
Applying Lemma \ref{lem:covering-product} with $\mathcal{F}_1 = \mathcal{G}$, $\mathcal{F}_2 = \mathcal{H}$ and $C_1=1, C_2 =\pi_0$ gives 
\begin{equation}
N\bigl(\pi_0 \varepsilon, \mathcal{F}_1, L_{2}(Q)\bigr) \leq N \bigl(\varepsilon/2 , \mathcal{G}, L_{2}(Q)\bigr) N \bigl(\pi_0  \varepsilon/2 , \mathcal{H}, L_{2}(Q)\bigr) \leq \frac{2 \mathfrak{C}}{\varepsilon} \exp \bigl( C(\pi_0 \varepsilon)^{-\alpha_\pi} \bigr) . 
\end{equation}
Therefore, it holds for some constant $\mathfrak{C}'$ such that,
\begin{equation}\label{eq:covering1-whole}
\log N \bigl(\varepsilon \| F_1 \|_{Q}, \mathcal{F}_1, L_{2}(Q)\bigr) \leq C (\pi_0 \varepsilon)^{-\alpha_\pi} - \log(\varepsilon)+\mathfrak{C} \leq \mathfrak{C}'(\pi_0 \varepsilon)^{-\alpha_\pi} . 
\end{equation}
Because \eqref{eq:covering1-whole} holds for any probability measure $Q$, we choose $Q$ to be $\widetilde{Q}$ defined in \eqref{eq:new-measure-whole} to bound the right hand side of  \eqref{eq:covering-change-of-measure-whole}, 
\begin{equation}\label{eq:covering-whole}
\log N \bigl(\varepsilon \| F \|_{Q}, \mathcal{F}, L_{2}(Q) \bigr) \leq   \log N \bigl(\pi_0 \varepsilon, \mathcal{F}_1, L_{2}(\widetilde{Q})\bigr) \leq \mathfrak{C}' (\pi_0\varepsilon)^{-\alpha_\pi}.
\end{equation}
We apply Lemma A.1 of \cite{srebro2012optimistic} so that 
\begin{align}
\E  \| \mathbb{G}_N \|_{\mathcal{F} }/\pi_0 &\lesssim \inf_\eta \biggl\{ N^{1 / 2} \eta+\int_{\eta}^{1} \log _{+}^{1 / 2} N_{[]}\left(\pi_0 \varepsilon, \mathcal{F},\|\cdot\|_{L_{2}(P)}\right) \mathrm{d} \varepsilon \biggr\}+N^{-1 / 2} \log _{+} N_{[]}\left(\pi_0, \mathcal{F},\|\cdot\|_{L_{2}(P)}\right)\nonumber\\
&\leq \inf_\eta \biggl\{ N^{1 / 2} \eta+ \mathfrak{C} \int_{\eta}^{1}  (\pi_0\varepsilon)^{-\alpha_\pi/2} \mathrm{d} \varepsilon \biggr\} + N^{-1/2} \mathfrak{C} \pi_0^{-\alpha_\pi} \nonumber\\
&\leq  \frac{1}{\pi_0}N^{1 / 2 - 1/\alpha_\pi} \bigl(1+\mathbbm{1}\{\alpha_\pi=2\}\log N \bigr)+N^{-1/2} \mathfrak{C} \pi_0^{-\alpha_\pi} \nonumber\\
&\leq \frac{\mathfrak{C}}{\pi_0} N^{{(1 / 2 - 1/\alpha_\pi)}_{+}} \bigl(1+\mathbbm{1}\{\alpha_\pi=2\}\log N \bigr) \text{, when $N$ is large.} \label{eq:expectation-bound-whole}
\end{align}
Using McDiarmid's inequality we have 
\begin{equation}\label{eq:McDiarmid-application-whole}
\PP \bigl( \| \mathbb{G}_N \|_{\mathcal{F}}- \E  \| \mathbb{G}_N \|_{\mathcal{F} } \geq u \bigr) \leq  \exp \biggl(-\frac{2u^2}{ \sum_{i=1}^N c_i^2}\biggr) \leq  \exp \biggl(-\frac{2u^2}{ \sum_{i=1}^N 4\pi_0^2/N}\biggr) = \exp(-u^2/2\pi_0^2),
\end{equation}
where 
\begin{align*}
c_i &:= \sup_{\substack{(x_i, y_i,t_i), (x_i', y_i',t_i'),\\ (x_1, y_1,t_1), \ldots, (x_N, y_N,t_N)}} \sup_{\widehat{\pi},\theta} \sqrt{N} \biggl| \frac{1}{N}\sum_{j=1}^N \mathbbm{1}\{t_j=0\} \widehat{\pi}(x_j)  \mathbbm{1}\{r_j \leq \theta \}  - \frac{1}{N} \sum_{j=1, j\neq i}^N \mathbbm{1}\{t_j=0\} \widehat{\pi}(x_j)  \mathbbm{1}\{r_j \leq \theta \} \\
&\qquad \qquad \qquad -\frac{1}{N}\mathbbm{1}\{t_i'=0\} \widehat{\pi}(x_i')  \mathbbm{1}\{r_i' \leq \theta \}\biggr| \\
& \leq \sup_{(x_i, y_i,t_i), (x_i', y_i',t_i')} \sup_{\widehat{\pi},\theta} \sqrt{N} \biggl| \frac{1}{N} \mathbbm{1}\{t_{i}=0\} \widehat{\pi}(x_{i})  \mathbbm{1}\{r_{i} \leq \theta \}- \frac{1}{N} \mathbbm{1}\{t_i'=0\} \widehat{\pi}(x_i')  \mathbbm{1}\{r_i' \leq \theta \}  \biggr| \le \frac{2 \pi_0}{\sqrt{N}}.
\end{align*}
Substituting the expectation bound~\eqref{eq:expectation-bound-whole} in~\eqref{eq:McDiarmid-application-whole} and setting the right hand side of~\eqref{eq:McDiarmid-application-whole} to $\delta$ yields for another absolute constant $\mathfrak{C}^{\prime}$,
\begin{equation}\label{eq:r1-whole}
 \PP \biggl[ \| \mathbb{G}_n \|_{\mathcal{F}} \geq \mathfrak{C}^{\prime} \biggl\{  N^{{(1 / 2 - 1/\alpha_\pi)}_{+}} \bigl(1+\mathbbm{1}\{\alpha_\pi=2\}\log N \bigr) + \pi_0 \sqrt{2\log \bigl( \frac{1}{\delta} \bigr)} \biggr\} \biggr] \leq \delta.
\end{equation}
\paragraph{Bound on $\sup_{\theta} \mathcal{R}_2(\widehat{m},\theta)$:}\label{pf:bound-r2-whole}
We define the class of functions $\mathcal{F}=\{f: f_{\widehat{m},\theta}(w)=\widehat{m}(\theta,x_i) (\mathbbm{1}\{t_i=1\} - \mathbbm{1}\{t_i=0\} , \forall \widehat{m}(\cdot, \cdot) \in \mathcal{F}_m \text{ and } \theta \in \mathbb{R} \}$ with $F(w) = m_0 (\mathbbm{1}\{t_i=1\} - \mathbbm{1}\{t_i=0\} )$ as its envelope. Similar as before we bound this term by $\sup_{f \in \mathcal{F}} |\mathbb{G}_N f|$ and for any fixed $f \in \mathcal{F}, |\mathbb{G}_N f|/\sqrt{N}$ is an average of i.i.d. random variables.

For any pair of functions $f_{\widehat{m}_1,\theta_1},f_{\widehat{m}_2,\theta_2} \in \mathcal{F}$,
\begin{align}
        \|f_{\widehat{m}_1,\theta_1} - f_{\widehat{m}_2,\theta_2} \|_Q &= \Bigl[ \sum_{i=1}^N  \bigl( \mathbbm{1} \{t_i=1\} - \mathbbm{1} \{t_i=0\}  \bigr)^2 \bigl(\widehat{m}_1(\theta_1, x_i)  - \widehat{m}_2(\theta_2,x_i)   \bigr)^2 Q(w_i) \Bigr]^{1/2}\\
        &= \frac{ \Bigl[ \sum_{i=1}^N  \bigl( \mathbbm{1} \{t_i=1\} - \mathbbm{1} \{t_i=0\}  \bigr)^2 \bigl(\widehat{m}_1(\theta_1, x_i)  - \widehat{m}_2(\theta_2,x_i)   \bigr)^2 Q(w_i) \Bigr]^{1/2} }{ \bigl[ \sum_{i=1}^N  \bigl( \mathbbm{1} \{t_i=1\} - \mathbbm{1} \{t_i=0\}  \bigr)^2 Q(w_i) \bigr]^{1/2}} \times\\
        &\qquad \qquad \qquad \qquad \qquad \qquad \qquad \qquad \bigl[ \sum_{i=1}^N  \bigl( \mathbbm{1} \{t_i=1\} - \mathbbm{1} \{t_i=0\}  \bigr)^2  Q(w_i)  \bigr]^{1/2}\\
        &= \bigl\| \bigl(\widehat{m}_1(\theta_1, x_i)  - \widehat{m}_2(\theta_2,x_i) \bigr\|_{ \widetilde{Q}} \cdot \bigl[ \sum_{i=1}^n \bigl( \mathbbm{1} \{t_i=1\} - \mathbbm{1} \{t_i=0\}  \bigr)^2 Q(w_i) \bigr]^{1/2}, \label{eq:change-of-measure-whole}
\end{align}
where $\widetilde{Q}$ is the new probability measure defined by 
\begin{equation}\label{eq:new-measure-whole-r2}
\widetilde{Q}(w_i):=  \bigl( \mathbbm{1} \{t_i=1\} - \mathbbm{1} \{t_i=0\}  \bigr)^2  Q(w_i) /  \sum_{i=1}^N \bigl( \mathbbm{1} \{t_i=1\} - \mathbbm{1} \{t_i=0\}  \bigr)^2 Q(w_i), i=1,\dots, N.
\end{equation}
Using $F(w)= m_0 (\mathbbm{1}\{t_i=1\} - \mathbbm{1}\{t_i=0\} )$ as the envelope function of the class $\mathcal{F}$, \eqref{eq:change-of-measure-whole} becomes
\begin{equation}\label{eq:distance-relationship-whole-r2}
        \| f_{\widehat{m}_1,\theta_1} - f_{\widehat{m}_2,\theta_2} \|_Q = \bigl\| \widehat{m}_1(\theta_1, x_i)  - \widehat{m}_2(\theta_2,x_i) \bigr\|_{ \widetilde{Q}} \cdot \|F\|_Q / m_0.
\end{equation} 
Recall the class of functions $\mathcal{F}_m:= \bigl\{f: f_{\widehat{m},\theta}(z) = \widehat{m} (\theta, x) \bigr\}$ with its envelope $F_1 \equiv m_0$. Then  \eqref{eq:distance-relationship-whole-r2} gives the  relationship between the covering numbers of the two classes $\mathcal{F}$ and $\mathcal{F}_m$. Along with the bracketing number assumption on $\mathcal{F}_m$ it follows that
\begin{equation}\label{eq:covering-r2-whole}
\log N\bigl(\varepsilon \| F \|_{Q}, \mathcal{F}, L_{2}(Q)\bigr) \leq    \log N \bigl( m_0 \varepsilon, \mathcal{F}_m, L_{2}(\widetilde{Q}) \bigr) \leq \log N_{[]} \bigl( 2m_0 \varepsilon, \mathcal{F}_m, L_{2}(\widetilde{Q}) \bigr) \leq C (2 m_0\varepsilon)^{-\alpha_m}.
\end{equation}
Applying the same technique as we used in \eqref{eq:expectation-bound-whole}  when bounding $\sup_{\theta} \mathcal{R}_1$ gives us $\forall \alpha_m \geq 0$,
\begin{equation}\label{eq:expectation-bound-r2-whole}
\E  \| \mathbb{G}_N \|_{\mathcal{F} } \leq \mathfrak{C}  N^{{(1 / 2 - 1/\alpha_m)}_{+}} \bigl(1+\mathbbm{1}\{\alpha_m=2\}\log N \bigr).
\end{equation}
Using McDiarmid's inequality we have 
\begin{equation}\label{eq:McDiarmid-application-r2-whole}
\PP ( \| \mathbb{G}_N \|_{\mathcal{F}}- \E  \| \mathbb{G}_N \|_{\mathcal{F} } \geq u ) \leq  \exp \biggl(-\frac{2u^2}{ \sum_{i=1}^N c_i^2}\biggr) \leq  \exp \biggl(-\frac{2u^2}{ \sum_{i=1}^N 4 m_0^2/N}\biggr) = \exp(-u^2/2m_0^2),
\end{equation}
where 
\begin{align*}
c_i &:= \sup_{\substack{(x_i, y_i,t_i), (x_i', y_i',t_i'),\\ (x_1, y_1,t_1), \ldots, (x_N, y_N,t_N)}} \sup_{\widehat{m},\theta} \sqrt{N} \biggl| \frac{1}{n}\sum_{j=1}^N \widehat{m}(\theta,x_j) (\mathbbm{1}\{t_j=1\} - \mathbbm{1}\{t_j=0\} ) \\
&\qquad \qquad \qquad - \frac{1}{N} \sum_{j=1, j\neq i}^N \widehat{m}(\theta,x_j) (\mathbbm{1}\{t_j=1\} - \mathbbm{1}\{t_j=0\} ) -\frac{1}{N} \widehat{m}(\theta,x_i') (\mathbbm{1}\{t_i'=1\} - \mathbbm{1}\{t_i'=0\} ) \biggr|  \\
& \leq \sup_{(x_i, y_i,t_i), (x_i', y_i',t_i')} \sup_{\widehat{m},\theta} \frac{1}{\sqrt{N} } \biggl|  \widehat{m}(\theta,x_i) (\mathbbm{1}\{t_i=1\} - \mathbbm{1}\{t_i=0\} )- \widehat{m}(\theta,x_i') (\mathbbm{1}\{t_i'=1\} - \mathbbm{1}\{t_i'=0\} )\biggr| \le \frac{2 m_0}{\sqrt{N}}.
\end{align*}
Substituting the expectation bound~\eqref{eq:expectation-bound-r2-whole} in~\eqref{eq:McDiarmid-application-r2-whole} and setting the right hand side of~\eqref{eq:McDiarmid-application-r2-whole} to $\delta$ yields for another absolute constant $\mathfrak{C}^{\prime}$,
\begin{equation}\label{eq:r2-whole}
 \PP \biggl[ \| \mathbb{G}_N \|_{\mathcal{F}} \geq \mathfrak{C} \biggl\{  N^{{(1 / 2 - 1/\alpha_m)}_{+}} \bigl(1+\mathbbm{1}\{\alpha_m=2\}\log N \bigr) + m_0 \sqrt{2\log \bigl( \frac{1}{\delta} \bigr)} \biggr\} \biggr] \leq \delta.
\end{equation}
\paragraph{Bound on $\mathcal{R}_3$:}

Because the random variables in the averaging of $\mathcal{R}_3$ are i.i.d, applying Hoeffding's inequality yields
\begin{equation}\label{eq:r3-hoeffding-whole}
\PP \biggl\{ \frac{1}{N} \sum_{i=1}^{N}\mathbbm{1}\{ t_i=1 \} - \PP(T=1) \geq t \biggr\} \leq \exp \biggl(-\frac{2 t^{2}}{N }\biggr).
\end{equation}
And this leads to
\begin{equation}\label{eq:r3-whole}
\PP \biggl(R_3 \geq (1-\alpha)\sqrt{\frac{1}{2N} \log \bigl( \frac{1}{\delta} \bigr)} \biggr) \leq \delta.
\end{equation}

Combining \eqref{eq:r1-whole}, \eqref{eq:r2-whole} and \eqref{eq:r3-whole} together using the union bound gives the result that for a universal constant $\mathfrak{C}$,
\begin{equation}\label{eq:r1-r2-r3-whole}
\begin{aligned}
\PP \biggl\{ \sup_\theta \bigl( \mathcal{R}_1(\theta)+\mathcal{R}_2(\theta) +\mathcal{R}_3 \bigr) \geq &\mathfrak{C}  \biggl(  N^{- 1/(\alpha_m \vee 2)} \bigl(1+\mathbbm{1}\{\alpha_m=2\}\log N \bigr) + N^{- 1/(\alpha_\pi \vee 2)} \bigl(1+\mathbbm{1}\{\alpha_\pi=2\}\log N \bigr) \\
&\qquad+ \sqrt{\frac{ (m_0+\pi_0+1 - \alpha)^2 \log \bigl( \frac{1}{\delta}  \bigr) }{N}} \biggr) \biggr\} \leq \delta.
\end{aligned}
\end{equation}

\end{proof}

\section{Proof of Theorem \ref{thm:impossibility-finite}}\label{sec:impossibility-appendix}

The first half of the Theorem, specifically \eqref{eq:finite-marginal-set-infinite}, is adapted from Lemma S1 of \cite{qiu2022distribution}. We modify that proof a bit in our setting in the following and first state a Lemma that is adapted from Theorem 2 and Remark 4 from \cite{shah2020hardness}, which will contribute to the proof of Theorem \ref{thm:impossibility-finite}.
\begin{lem}\label{thm:conditional-independence-impossibility}(No-free-lunch for conditional independence testing). 
Assume that $X$ is continuous, given any $n \in \mathbb{N}, \alpha \in(0,1), M \in(0, \infty]$, and any potentially randomised test $\eta_n$ that has valid level $\alpha$ for the null hypothesis $\mathcal{P}_{0, M}$, we have that $\PP_Q\left(\eta_n=1\right) \leq \alpha$ for all $Q \in \mathcal{Q}_{0, M}$. Thus $\eta_n$ cannot have power against any alternative.
\end{lem}

Let $\bar{\mathcal{E}} \supseteq \bar{\mathcal{P}}^0$ be the space of distributions for the full data point $\bar{O}$ with the distribution of $(X, Y) \mid T = 0,1$ both dominated by the Lebesgue measure. These distributions may or may not satisfy the MAR assumption that $Y \perp T \mid X$. For any $x \in \mathcal{X}$, the prediction set $\hat{C}(x)$ is the short hand notation for $C\left(x ; O_1, \ldots, O_n\right)$. This notation is helpful to clarify the dependence of $\hat{C}$ on the observed training data $\left(O_1, \ldots, O_n\right)$. Let $\bar{O}_{n+1}$ to denote the full data point from a future draw.

The subtle difference between $\bar{O}$ and $O$ is that $\bar{O}$ represents the full but unobserved data $(X,Y,T)$ whereas $O$ corresponds to the observed data $(X,T,(1-T)Y)$.

Define the randomized test $\eta\left(\bar{O}_1, \ldots, \bar{O}_{n+1}\right)$ as follows: 
\begin{equation}\label{eq:test}
\eta\left(\bar{O}_1, \ldots, \bar{O}_{n+1}\right)=
\begin{cases}
 \mathbbm{1}\bigl\{ Y_{n+1} \notin C (X_{n+1} ; O_1, \ldots, O_n ) \bigr\}, \text{ if } T_{n+1}=1 ;\\
 \begin{cases}
1, \text{ w.p. } \alpha ,\\
0, \text{ w.p. } 1-\alpha,
\end{cases}
\text{ if } T_{n+1}=0.
\end{cases}
\end{equation}
 We note that although $\eta$ is a function of $n+1$ full data points $\left(\bar{O}_1, \ldots, \bar{O}_{n+1}\right)$, it only relies on $n$ observed training data points $\left(O_1, \ldots, O_n\right)$ and one future full data point $\bar{O}_{n+1}$. 
Furthermore, we can see that
\begin{align*}
    \PP_{\bar{P}^0}\left(\eta (\bar{O}_1, \ldots, \bar{O}_{n+1} )=1\right)  &=  \PP_{\bar{P}^0}\left(\eta (\bar{O}_1, \ldots, \bar{O}_{n+1} )=1 | T_{n+1} = 1 \right) \PP(T_{n+1} = 1)   \\
    &\qquad + \PP_{\bar{P}^0}\left(\eta (\bar{O}_1, \ldots, \bar{O}_{n+1} )=1 | T_{n+1} = 1 \right) \PP(T_{n+1} = 1)\\
    & = \PP_{\bar{P}^0} \bigl( Y \notin \widehat{C}(X) | T=1 \bigr) \PP(T = 1) + \PP_{\bar{P}^0} \bigl( \eta = 1 | T=0 \bigr) \PP(T = 0) \\
    &\leq \alpha  \PP(T = 1) + \alpha \PP(T=0) \text{, by \eqref{eq:finite-marginal-set} and the definition of }\eta\\
    &\leq \alpha.
\end{align*}
Therefore, $\eta$ can be viewed as a test with level $\alpha$ for the null hypothesis $Y \perp T \mid X$. Because Theorem \ref{thm:conditional-independence-impossibility} states that the power of $\eta$ against the alternative hypothesis is at most $\alpha$, we have that
\begin{align}\label{eq:power}
\PP_{\bar{Q}}\left(\eta\left(\bar{O}_1, \ldots, \bar{O}_{n+1}\right)=1\right) \leq \alpha \quad \text { for any distribution } \bar{Q} \in \bar{\mathcal{E}} \text {. }
\end{align}
Using the fact that $\eta$ is a completely randomized test when when $T_{n+1}) = 0$, this gives $\PP(\eta = 1 | T_{n+1} = 0) = \alpha$.
Therefore,
\begin{align*}
    \alpha \geq \PP_{\bar{Q}}\left(\eta\left(\bar{O}_1, \ldots, \bar{O}_{n+1}\right)=1\right) & = \PP_{\bar{Q}} (\eta = 1 | T_{n+1} = 0) \PP( T_{n+1} = 0) + \PP_{\bar{Q}} (\eta = 1 | T_{n+1} = 1) \PP( T_{n+1} = 1)\\
    &= \alpha \PP ( T_{n+1} = 0 ) + \PP_{\bar{Q}} (\eta = 1 | T_{n+1} = 1) \PP( T_{n+1} = 1).
\end{align*}
Therefore, 
\begin{align}\label{eq:power-conditional}
    \alpha \geq \PP_{\bar{Q}} (\eta = 1 | T_{n+1} = 1)  = \PP_{\bar{Q}}\bigl(Y_{n+1} \notin C (X_{n+1} ; O_1, \ldots, O_n ) \mid T_{n+1}=1\bigr).
\end{align}

For any $x \in \mathcal{X}$, let $D_x \subseteq \mathcal{Y}$ be any Lebesgue measurable set with nonzero finite measure and ${U}_x$ be the uniform distribution on $D_x$. Take $\bar{Q}$ to be a distribution such that 
\begin{enumerate}[label=(\roman*)]
    \item the distribution of $T$ is an arbitrary Bernoulli distribution with success probability in $(0,1)$;
    \item the distributions of $X \mid T=0,1$ satisfy the dominance condition \eqref{eq:measure-dominating}, which is that the source distribution ($X|T=0$) dominates the target distribution ($X|T=1$), and are arbitrary in all other aspects
    \item  the distribution of $Y \mid X=x, T=0$ is arbitrary and the distribution of $Y \mid X=x, T=1$ is ${U}_x$.
\end{enumerate}

We take $\bar{P}^0$ to be the distribution that is identical to $\bar{Q}$, except that the distribution of $Y \mid X=x, T=1$ is identical to $Y \mid X=x, T=0$ rather than ${U}_x$ under $\bar{P}^0$. Note that $\bar{P}^0 \in \bar{\mathcal{P}}^0$. Since $\hat{C}$ is trained only on observed data $\left(O_1, \ldots, O_n\right)$ and $\bar{P}^0$ and $\bar{Q}$ imply the same distribution of the observed data point $O=(X,T, (1-T)Y)$, following \eqref{eq:power-conditional}, we have that
$$
\begin{aligned}
&\PP_{\bar{Q}}\bigl(Y_{n+1} \notin C (X_{n+1} ; O_1, \ldots, O_n ) \mid T_{n+1}=1\bigr) \\
=&\int_y \PP_{\bar{P}_0}\bigl(y \notin C (X_{n+1} ; O_1, \ldots, O_n ) \mid T_{n+1}=1\bigr)U_{X_{n+1}}(\mathrm{d} y) \leq \alpha,
\end{aligned}
$$
where the probability is over training data and possible exogenous randomness in $C$. Since $D_x$ is arbitrary, it follows that the integrand is bounded by $\alpha$, namely
\begin{equation}\label{eq:impossibility-first-step}
\PP_{\bar{P}_0}\Bigl(y \notin C (X_{n+1} ; O_1, \ldots, O_n ) \mid T_{n+1}=1 \Bigr) \leq \alpha
\end{equation}
for a.e. $y \in \mathcal{Y}$. The desired result \eqref{eq:finite-marginal-set-infinite} follows by replacing the notations $X_{n+1}$ and $T_{n+1}$ with $X$ and $T$, respectively, and noting that $C\left(x ; O_1, \ldots, O_n\right)=\hat{C}(x)$ by definition.

Note that \eqref{eq:impossibility-first-step} implies
\[
\inf_{y\in\mathcal{Y}} \mathbb{P}_{\bar{P}_0}(y \in \widehat{C}(X)|T = 1) \ge 1 - \alpha,
\]
for any prediction set $\widehat{C}$ that satisfies~\eqref{eq:finite-marginal-set}. Define $L(X) = \inf\{t:\,t\in\widehat{C}(X)\}$ and $U(X) = \sup\{t:\,t\in\widehat{C}(X)\}$. Because $\widehat{C}(X)\subseteq[L(X),\infty)$, we get
\begin{align}
\mathbb{P}_{\bar{P}_0}\Bigl(\inf_{y\in\mathcal{Y}}y \ge L(X)\bigm|T = 1\Bigr) = \inf_{y\in\mathcal{Y}}\,\mathbb{P}_{\bar{P}_0}(y \ge L(X) \bigm|T = 1) = \inf_{y\in\mathcal{Y}}\mathbb{P}_{\bar{P}_0} \bigl(y\in [L(X), \infty) \bigm| T = 1 \bigr) \ge 1 - \alpha. 
\end{align}
Similarly, because $\widehat{C}(X)\subseteq (-\infty, U(X)]$,
\begin{align}
\mathbb{P}_{\bar{P}_0}\Bigl(U(X) \ge \sup_{y\in\mathcal{Y}}y\bigm|T = 1\Bigr) = \inf_{y\in\mathcal{Y}}\mathbb{P}_{\bar{P}_0}\bigl(y \le U(X)\bigm|T = 1\bigr) = \inf_{y\in\mathcal{Y}}\mathbb{P}_{\bar{P}_0}\bigl(y\in (-\infty, U(X)]\bigm|T = 1\bigr) \ge 1 - \alpha.
\end{align}
Hence with $\mathcal{Y} = \mathbb{R}$ here, then any prediction set with valid coverage must have at least one of the end points $\infty$ in absolute value with probability at least $1-\alpha$ and hence must have an infinite diameter with probability at least $1-\alpha$.

\section{Some useful propositions and lemmas}\label{sec:supporting}
\begin{prop}\label{prop:chernozhukov2009improving}(\cite{chernozhukov2009improving}, Proposition 2)
Let the target function $f_{0}: \mathcal{X}^{d} \rightarrow K$ be weakly increasing and measurable in $x .$ Let $\hat{f}: \mathcal{X}^{d} \rightarrow K$ be a measurable function that is an initial estimate of $f_{0}.$ 
\begin{enumerate}
    \item For each ordering $\pi$ of $1, \ldots, d$, the $\pi$-rearranged estimate $\hat{f}_{\pi}^{*}$ is weakly increasing. Moreover, $\hat{f}^{*}$, an average of $\pi$-rearranged estimates, is weakly increasing.
    \item A $\pi$-rearranged estimate $\hat{f}_{\pi}^{*}$ of $\hat{f}$ weakly reduces the estimation error of $\hat{f}$ :
$$
\biggl\{\int_{\mathcal{X}^{d}}\bigl|\hat{f}_{\pi}^{*}(x)-f_{0}(x)\bigr|^{p} d x\biggr\}^{1 / p} \leq\biggl\{\int_{\mathcal{X}^{d}}\bigl|\hat{f}(x)-f_{0}(x)\bigr|^{p} \mathrm{d} x\biggr\}^{1 / p}.
$$
\end{enumerate}
\end{prop}

\subsection{Lemma \ref{lem:expectation-empirical-proc-step} and its proof}\label{pf:lem:expectation-empirical-proc-step}
\begin{lem}\label{lem:expectation-empirical-proc-step}
There exists a universal constant $\mathfrak{C} < \infty$ such that for any functions $s(t,x)\in[-\kappa_0, \kappa_0]$ and $h(x,y)$,
\begin{equation}
    \E \Bigl[ \sup_{\theta} |\mathbb{G}_{n}[ s(t,x)\mathbbm{1}\{h(x,y) \leq \theta  \}]| \Bigr] \leq \mathfrak{C}\kappa_0.
\end{equation}
\end{lem}
\begin{proof}
We have a class of functions $\mathcal{F}=\bigl\{f: f_\theta(w)=s(t,x) \mathbbm{1}{\{ h(x,y) \leq \theta\} }  \bigr\}$. Notice that $\forall \theta \in \mathbb{R}$, we have $|f_\theta (w)| \leq |s(t,x)|$. Therefore, $F(w) := |s(t,x)|$ is an envelope function of $\{ f_\theta(\cdot): \theta \in \mathbb{R} \}$. For any discrete probability measure $Q$, let $\|f(\cdot)\|_{Q}$ denote the empirical $L_2(Q)$ norm where $\|f(\cdot)\|_{Q} := \bigl( \sum_{i=1}^n f^2(x_i) Q(x_i) \bigr)^{1/2} $.

Let $h_i$ denote $h_i:=h(x_i, y_i)$. For any function $f \in \mathcal{F}$ and $\theta_1,\theta_2 \in \mathbb{R}$, 
\begin{align}
        \|f_{\theta_1} - f_{\theta_2} \|_Q &= \Bigl[ \sum_{i=1}^n  s(t_i,x_i)^2 \bigl( \mathbbm{1}{\{ h_i \leq \theta_1 \}} - \mathbbm{1}{\{ h_i \leq \theta_2\} } \bigr)^2 Q(w_i) \Bigr]^{1/2}\\
        &= \frac{ \Bigl[ \sum_{i=1}^n   s(t_i,x_i)^2 \bigl( \mathbbm{1}{\{ h_i \leq \theta_1 \}} - \mathbbm{1}{\{ h_i \leq \theta_2\} } \bigr)^2 Q(w_i) \biggr]^{1/2} }{ \bigl[ \sum_{i=1}^n  s(t_i,x_i)^2 Q(w_i) \bigr]^{1/2}} \times\\
        &\qquad \qquad \qquad \qquad \qquad \qquad \qquad \qquad \qquad \qquad \bigl[ \sum_{i=1}^n s(t_i,x_i)^2 Q(w_i) \bigr]^{1/2}\\
        &= \bigl\| \mathbbm{1}{\{ h_i \leq \theta_1\}} - \mathbbm{1}{\{ h_i \leq \theta_2\} } \bigr\|_{ \widetilde{Q}} \cdot \bigl[ \sum_{i=1}^n s(t_i,x_i)^2 Q(w_i) \bigr]^{1/2}, \label{eq:change-of-measure}
\end{align}
where $\widetilde{Q}$ is the new probability measure defined by 
\begin{equation}\label{eq:new-measure}
\widetilde{Q}(w_i):=  s(t_i,x_i)^2 Q(w_i) /  \sum_{i=1}^n s(t_i,x_i)^2 Q(w_i).
\end{equation}
Using the definition of $F(w)$ as the envelope function of the class $\mathcal{F}$, \eqref{eq:change-of-measure} becomes
\begin{equation}\label{eq:distance-relationship}
        \| f_{\theta_1} - f_{\theta_2} \|_Q = \bigl\| \mathbbm{1}{\{ h_i \leq \theta_1\}} - \mathbbm{1}{\{ h_i \leq \theta_2\} } \bigr\|_{ \widetilde{Q}} \cdot \|F\|_Q.
\end{equation} 
Define a new class of functions $\mathcal{F}_1:= \bigl\{f: f_\theta(z) = \mathbbm{1}\{h(x,y) \leq \theta \} \bigr\}$ with its envelope function $F_1 \equiv 1$. Then  \eqref{eq:distance-relationship} gives a relationship between the covering numbers of the two classes $\mathcal{F}$ and $\mathcal{F}_2$,
\begin{equation}\label{eq:covering-change-of-measure}
N\bigl(\varepsilon \| F \|_{Q}, \mathcal{F}, L_{2}(Q)\bigr) \leq    N \bigl(\varepsilon, \mathcal{F}_1, L_{2}(\widetilde{Q}) \bigr).
\end{equation}

From Chapter 2.6 of \cite{van1996weak}, we know that the VC dimension of the class $\mathcal{F}_1 $ is $1$. We use Theorem 2.6.7 of \cite{van1996weak} with $r=2$ to get a bound on the covering numbers. Specifically, for any probability measure $Q$, there exists a universal constant $\mathfrak{C}$ such that
\begin{equation}\label{eq:covering1}
 N\bigl(\varepsilon, \mathcal{F}_1, L_{2}(Q)\bigr) = N \bigl(\varepsilon \| F_1 \|_{Q}, \mathcal{F}, L_{2}(Q)\bigr) \leq \frac{\mathfrak{C}}{\varepsilon}. 
\end{equation}
Because \eqref{eq:covering1} holds for any probability measure $Q$, we choose $Q$ to be $\widetilde{Q}$ defined in \eqref{eq:new-measure} in order to bound the right hand side of  \eqref{eq:covering-change-of-measure}, 
\begin{equation}\label{eq:covering}
N \bigl(\varepsilon \| F \|_{Q}, \mathcal{F}, L_{2}(Q) \bigr) \leq    N \bigl(\varepsilon, \mathcal{F}_1, L_{2}(\widetilde{Q})\bigr) \leq \frac{\mathfrak{C}}{\varepsilon}.
\end{equation}

Next, we obtain an upper bound of the uniform-entropy 
$$
J(\delta, \mathcal{F}):=\sup _{Q} \int_{0}^{\delta} \sqrt{1+\log N\bigl(\varepsilon\|F \|_{Q}, \mathcal{F}, L_{2}(Q)\bigr)} d \varepsilon,
$$
where the supremum is taken over all discrete probability measures $Q$ with $\|F \|_{Q}>0.$ 
Applying $p=1$ to Theorem 2.14.1 of \cite{van1996weak} gives us
\begin{equation}\label{eq:expectation-bound}
 \E \sup_{ f \in \mathcal{F} } | \mathbb{G}_n f | ~\leq~ J(1, \mathcal{F} ) \|F \|_{Q} ~\lesssim~ \|F \|_{Q} \int_{0}^{1}  \mathfrak{C}^{1/2}\sqrt{\log(1/\varepsilon) } d \varepsilon ~\leq~ \mathfrak{C}_1 \kappa_0, 
\end{equation}
where the second inequality is from \eqref{eq:covering} and $\mathfrak{C}_1$ is a universal constant.
\end{proof}

\subsection{Lemma \ref{lem:covering-product} and its proof}
\begin{lem}\label{lem:covering-product}
For any two bounded function classes $\mathcal{F}_1$ and $\mathcal{F}_2$ with $\sup_{f_1 \in \mathcal{F}_1 } \|f_1 \|_\infty \leq C_1$ and $\sup_{f_2 \in \mathcal{F}_2} \|f_2 \|_\infty \leq C_2$, the following holds for the covering number of the class of functions $\mathcal{F}_1 \times \mathcal{F}_2 := \{f=f_1 f_2, f_1 \in \mathcal{F}_1 \text{ and }f_2 \in \mathcal{F}_2 \}$,
\begin{equation}\label{eq:covering-product}
   N \bigl(\varepsilon, \mathcal{F}_1 \times \mathcal{F}_2, L_2(Q) \bigr) \leq N \bigl(\frac{\varepsilon}{2C_2}, \mathcal{F}_1, L_2(Q) \bigr) N \bigl(\frac{\varepsilon}{2C_1}, \mathcal{F}_2, L_2(Q) \bigr),
\end{equation}
where $Q$ is some probability measure.
\end{lem}

\begin{proof}
For any two positive numbers $\epsilon_1$ and $\epsilon_2$, let $\mathcal{G}=\{g_1,\dots,g_K\} \subseteq \mathcal{F}_1$ be an $\varepsilon_1$-net for $\mathcal{F}_1$ and $\mathcal{H}=\{h_1,\dots,h_L\} \subseteq \mathcal{F}_2$ be an $\varepsilon_2$-net for $\mathcal{F}_2$, where $K=N (\varepsilon_1, \mathcal{F}_1, L_2(Q))$ and $L=N \bigl(\varepsilon_2, \mathcal{F}_2, L_2(Q))$. For any function $f \in \mathcal{F}_1 \times \mathcal{F}_2$ with $f=f_1 f_2$, where $f_1 \in \mathcal{F}_1$ and $f_2 \in \mathcal{F}_2$, let $\widetilde{g} $ and $\widetilde{h}$ be the closest element in $\mathcal{G}$ and $\mathcal{H}$ to $f_1$ and $f_2$ respectively. Then,
\begin{align*}
\|f-\widetilde{g} \widetilde{h} \|_Q &= \|f_1 f_2-\widetilde{g} \widetilde{h} \|_Q\\
&= \| (f_1 - \widetilde{g} )f_2 +  \widetilde{g} (f_2 - \widetilde{h}) \|_Q \\
&\leq  \| (f_1 - \widetilde{g} )f_2\|_Q +  \| \widetilde{g} (f_2 - \widetilde{h}) \|_Q \\
&\leq C_2 \varepsilon_1 +C_1 \varepsilon_2\\
& = \varepsilon, \text{ when we take }\varepsilon_1 =\varepsilon/2C_2 \text{ and } \varepsilon_2=\varepsilon/2C_1.
\end{align*}
This way $\{ g_k h_l: 1 \leq k \leq K, 1 \leq l \leq L\}$ is an $\varepsilon$-net of $\mathcal{F}_1 \times \mathcal{F}_2$ and the inequality \eqref{eq:covering-product} follows.
\end{proof}

\section{Some more simulation results}\label{sec:more-simulations}
\subsection{Absolute residual score}\label{sec:absolute-residuals}
Figure \ref{fig:synthetic_4methods} and \ref{fig:real_4methods} contain results comparing our doubly robust prediction method with different splits of data, i.e. we are comparing doubly robust prediction with full data described in Algorithm \ref{alg:influence-conformal-full} and split data (Algorithm \ref{alg:influence-conformal-split}) with two splits and three splits, and weighted conformal prediction, all under the absolute residual score as described at the start of Section \ref{sec:simulations}. It can be seen that DRP performs similarly across different splits and they all outperform WCP. 

\begin{figure}[H]
    \centering
    \includegraphics[width=\textwidth]{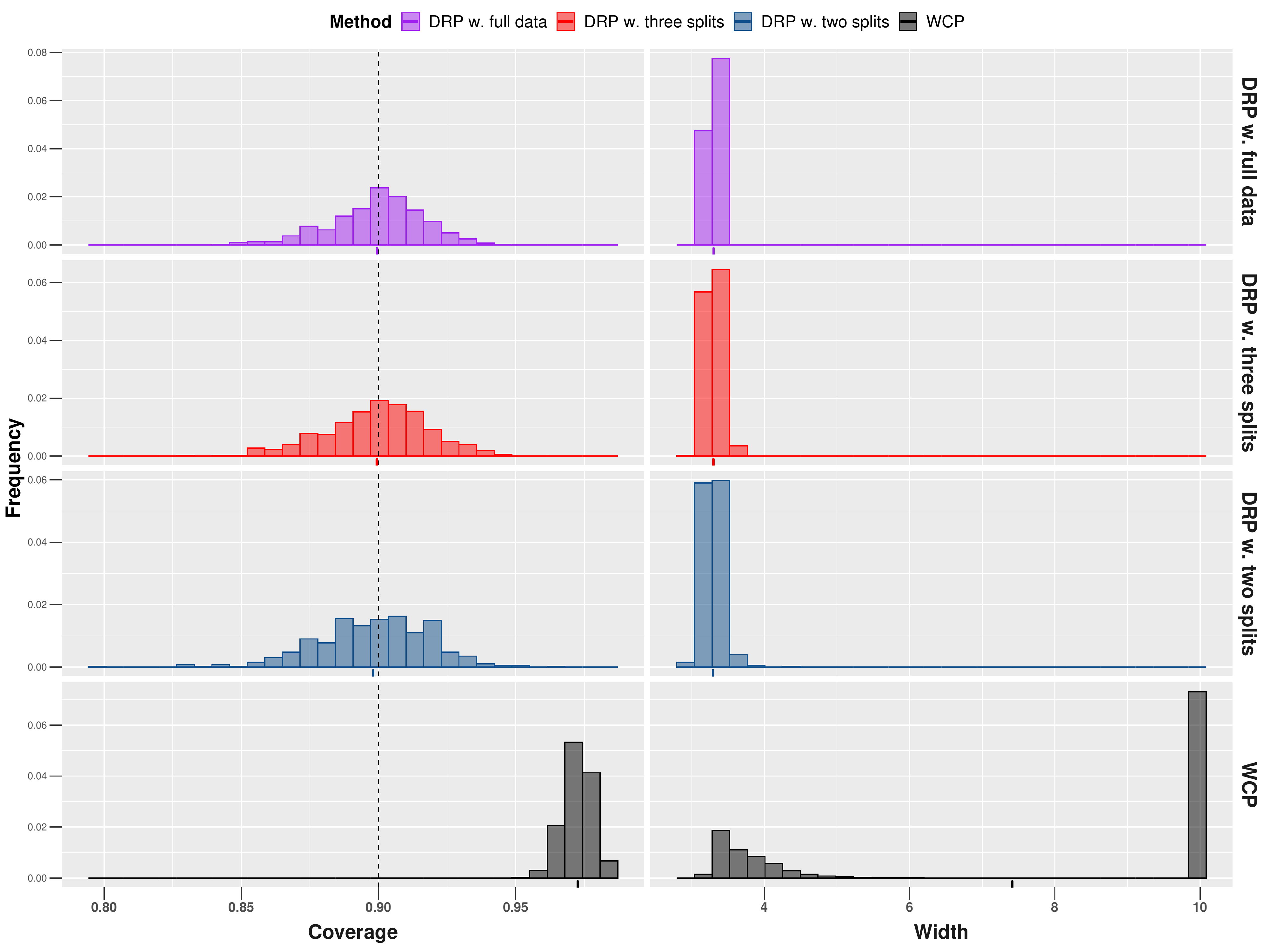}
    \caption{Coverage and width of Double Robust Prediction (DRP) with full data, two splits, three splits and Weighted Conformal Prediction (WCP) on  synthetic data with the absolute residual score. The width is truncated at 10 for WCP.}
    \label{fig:synthetic_4methods}
\end{figure}

\begin{figure}[H]
    \centering
    \includegraphics[width=\textwidth]{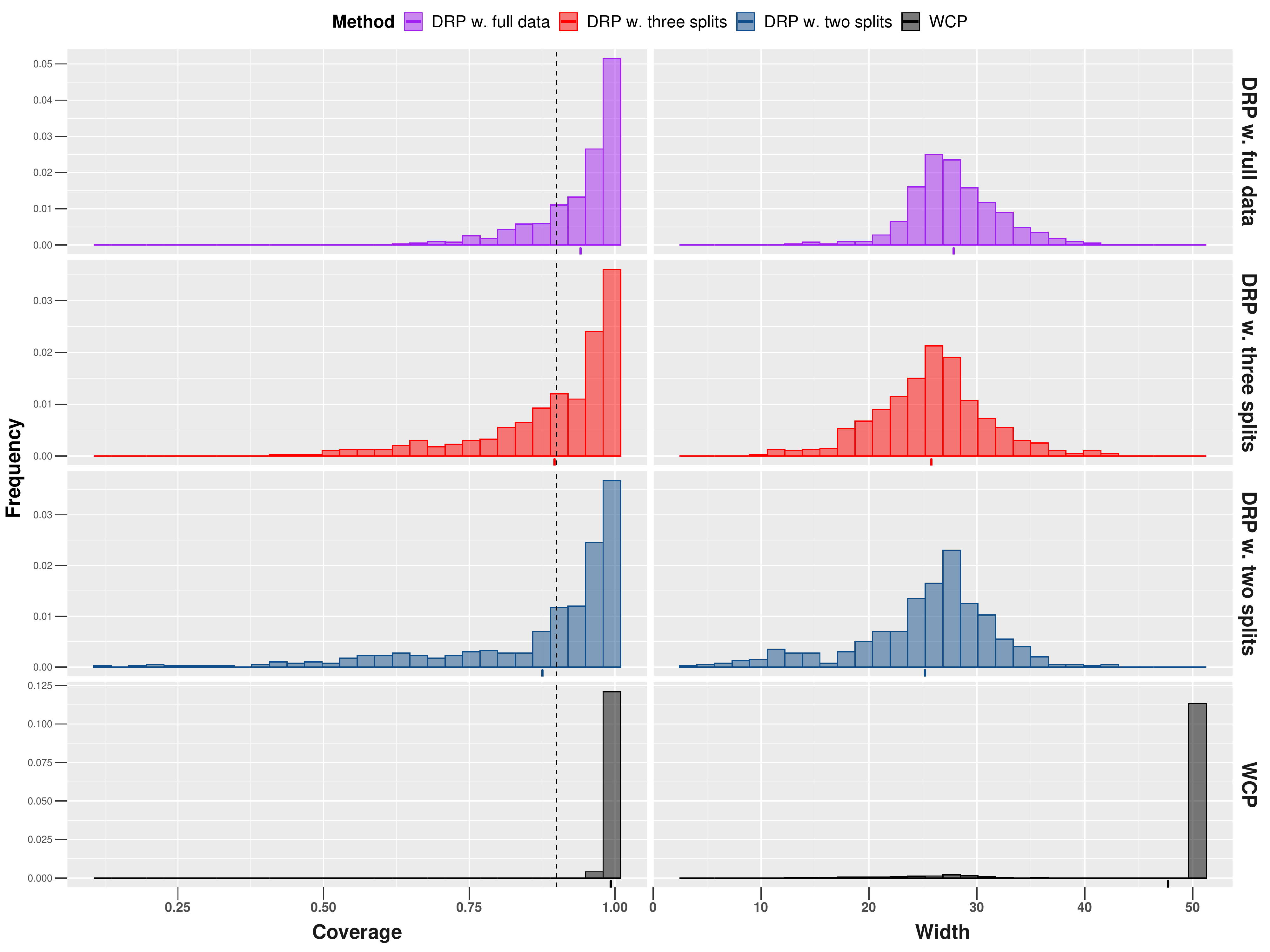}
    \caption{Coverage and width of Doubly Robust Prediction (DRP) with full data, two splits, three splits and Weighted Conformal Prediction (WCP) on  real data with the absolute residual score. The width is truncated at 50 for WCP.}
    \label{fig:real_4methods}
\end{figure}

\subsection{CQR score}\label{sec:cqr-synthetic}
Figure \ref{fig:synthetic_cqr} illustrates DRP and WCP's performance on synthetic data (described in Section \ref{sec:synthetic}) under the CQR score. It is found that all the methods achieve the desired coverage with WCP being overly conservative, i.e. having infinity width over half of the cases.
\begin{figure}[H]
    \centering
    \includegraphics[width=\textwidth]{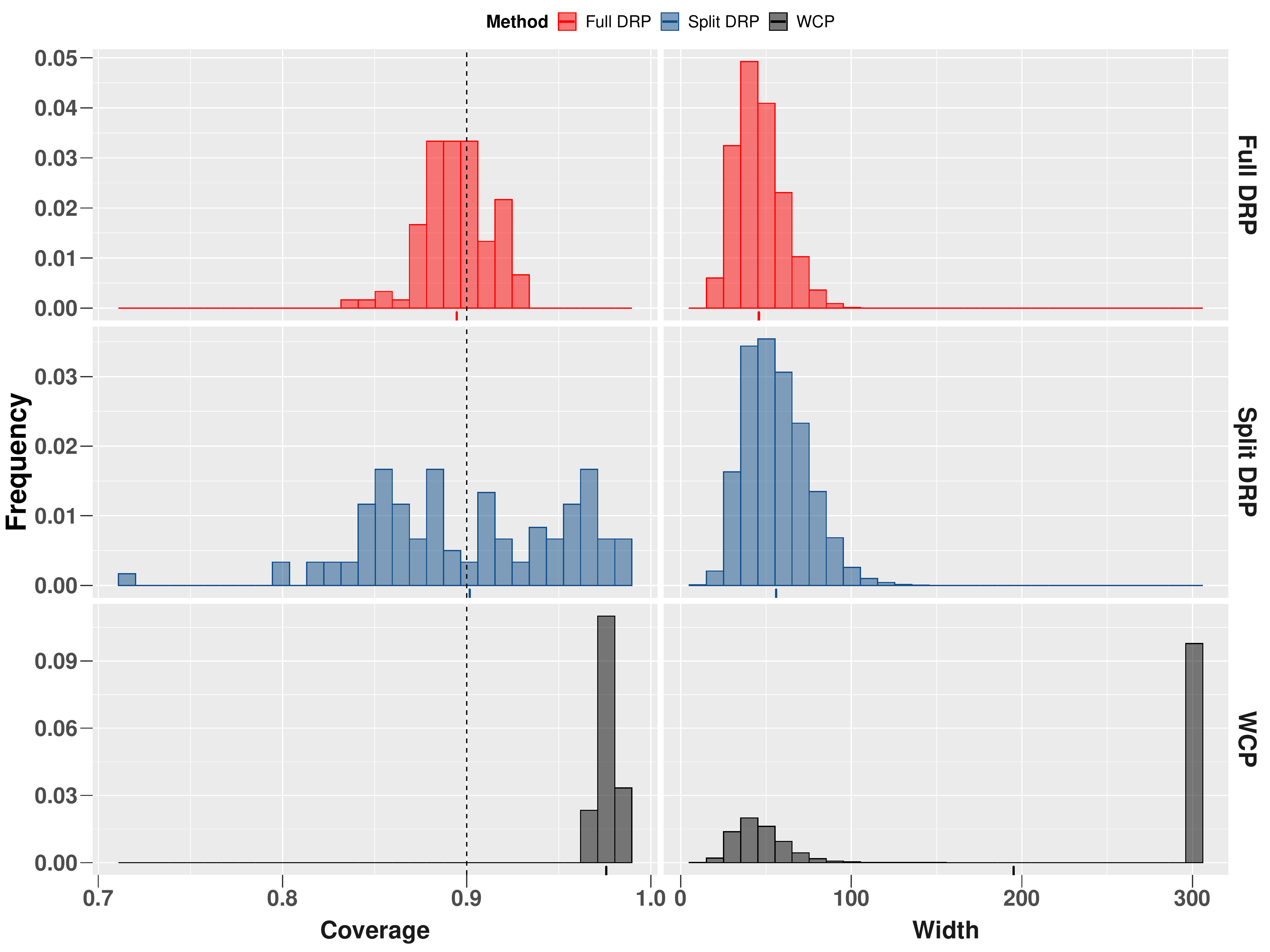}
    \caption{Histograms of coverage and width of Doubly Robust Prediction (DRP) and Weighted Conformal Prediction (WCP) on synthetic data through CQR score, where the width is truncated at 300 for WCP.}
    \label{fig:synthetic_cqr}
\end{figure}

\section{Proof of Theorem~\ref{thm:sensitivity-if}}\label{sec:proof-of-sensitivity-if}
\begin{proof}[Proof of Theorem~\ref{thm:sensitivity-if}]
By definition of the odds ratio function, we have the following expressions:
\begin{align}
&\log \frac{\PP (T=0|X=x,Y=y)}{\PP (T=1|X=x,Y=y)}=\log \frac{\PP(T=0|X=x,Y=0)}{\PP(T=1|X=x,Y=0)}+\gamma^\star(x,y); 
\end{align}
It follows from Bayes' rule that
\begin{align*}
    f(R|X,T=1) = f(R|X,T=0)\frac{\exp^{-\gamma^\star (X,Y)}}{\E \bigl\{\exp^{-\gamma^\star (X,Y) |T=0,X \bigr\} }}.
\end{align*}
And this leads to 
\begin{align}
&\E \Bigl(\mathbbm{1}\{R \leq \theta\}|X, T=1 \Bigr)=\frac{\E [ \mathbbm{1}\{R \leq \theta\} e^{-\gamma^\star (X,Y)}  \big| X,T=0 ]}{\E [ e^{-\gamma^\star (X,Y)} \big| X,T=0 ]} \label{eq:sens-nuisance2}.
\end{align} 
\begin{itemize}
\item If $\eta (x)$ is correct:
This would imply the density ratio is correct, i.e. $$ \exp \bigl\{  - \eta(x) - \gamma^\star(x,y)\bigr\} = \PP(T=1|X=x,Y=y) / \PP(T=0|X=x,Y=y),$$ then we have
\[
\mathbb{E}\left[\mathbbm{1}\{T = 0\}\exp \bigl\{  - \eta(X) - \gamma^\star(X,Y)\bigr\} \big|X = x, Y\right] = \mathbb{P}(T = 0|X = x, Y)\frac{\mathbb{P}(T = 1|X = x,Y)}{\mathbb{P}(T = 0|X = x,Y)} =\mathbb{P}(T = 1|X = x,Y) .
\]
This implies that
\begin{align*}
&\mathbb{E}\left[\mathbbm{1}\{T = 0\}  \exp \bigl\{  - \eta(X) - \gamma^\star(X,Y)\bigr\} \big\{\mathbbm{1}\{R \le \theta\} - m(\theta, X) \big\} \big| \theta   \right]\\ 
&\quad= \mathbb{E}\Bigl[\mathbb{P}(T = 1|X,Y)\big\{\mathbbm{1}\{R \le \theta \} - m(\theta, X)\big\} \big| \theta  \Bigr]\\
&\quad= \mathbb{E}\Bigl[\mathbb{P}(T = 1|X,Y) \bigl\{\mathbb{P}(R \le \theta|X) - m(\theta, X) \bigr\} \big| \theta   \Bigr].
\end{align*}
Similarly, 
\[
\mathbb{E}\Bigl[\mathbbm{1}\{T = 1\}\left\{m(\theta, X) - (1 - \alpha)\right\} \big| \theta  \Bigr] = \mathbb{E}\Bigl[\mathbb{P}(T = 1|X,Y)\bigl\{m(\theta, X) - (1 - \alpha)\bigr\} \big| \theta   \Bigr].
\]
Hence, if $\exp \bigl\{  - \eta(x) - \gamma^\star(x,y)\bigr\}$ is the true density ratio, then
\begin{align*}
\mathbb{E} [\mathrm{IF}(\theta, X, Y, R, T; \eta^\star, m, \gamma^\star)]
&= \E \Bigl[\PP (T = 1|X,Y) \bigl\{ \PP (R \le \theta |X, \theta) - (1 - \alpha) \bigr\} \big| \theta  \Bigr]\\
&= \E \Bigl[\PP (T = 1|X,Y) \bigl\{ \PP (R \le \theta |X, \theta)  \big| \theta \Bigr]- (1 - \alpha) \PP(T = 1)\\
& = \E \Bigl[ \E( \mathbbm{1}\{ T=1\} | X,Y ) \mathbbm{1}\{ R \leq\theta \}  \big| \theta \Bigr]  - (1 - \alpha) \PP(T = 1)\\
& =  \E \Bigl[ \E( 1\{ T=1\} \mathbbm{1} \{ R \leq \theta \} | X,Y, \theta) \Bigr] - (1 - \alpha) \PP(T = 1)\\
& = \E \Bigl[  1\{ T=1\} \mathbbm{1} \{ R \leq \theta \}  \big| \theta \Bigr] -  (1 - \alpha) \PP(T = 1)\\
& =  \PP(T = 1) \Big\{ \E \bigl[  \mathbbm{1} \{ R \leq \theta \} | T=1, \theta \bigr]   -  (1 - \alpha) \Bigr\}.
\end{align*}
Therefore,
\begin{align*}
    \E [\mathrm{IF}(r_\alpha, X, Y, R, T; \eta^\star, m, \gamma^\star)] = 0.
\end{align*}
\item If $m$ is correct: In this case,
\begin{align*}
&\mathbb{E}\left[\mathbbm{1}\{T = 0\}  \exp \bigl\{  - \eta(X) - \gamma^\star(X,Y)\bigr\} \big\{\mathbbm{1}\{R \le \theta\} - m(\theta, X)\big\}  \right]\\ 
&\quad= \mathbb{E}\Bigl\{ \exp \{ - \eta(X) \} \E \Bigl[ \mathbbm{1}(R \leq \theta) \exp\{ - \gamma^\star(X,Y) \} | T=0,X \Bigr]   \Bigr\}\PP(T=0)\\
& \qquad - \mathbb{E}\Bigl\{ \exp \{ - \eta(X) \} \PP(R \leq \theta | X,T=1) \E \Bigl[  \exp\{ - \gamma^\star(X,Y) \} | T=0,X \Bigr] \Bigr\}   \PP(T=0)\\
&\quad= 0 , \text{ by } \eqref{eq:sens-nuisance2}.
\end{align*}
Hence,
\begin{align*}
\mathbb{E}[\mathrm{IF}(\theta , X, Y, R, T; \eta, m^\star, \gamma^\star)] &= \mathbb{E}\Bigl[\mathbbm{1}\{T = 1\}\bigl\{ m^\star(\theta , X) - (1 - \alpha)\bigr\} \big| \theta \Bigr]\\ 
&= \mathbb{E}\Bigl[\mathbbm{1}\{T = 1\} \bigl\{\E \bigl( \mathbbm{1} \{R \le \theta \}|X,T=1, \theta \bigr) - (1 - \alpha)\bigr\} \big| \theta \Bigr]\\
& = \PP(T=1) \Bigl\{ \mathbb{E}\Bigl[ \E \bigl( \mathbbm{1} \{R \le \theta \}|X,T=1, \theta \bigr) | T=1, \theta \Bigr] - (1 - \alpha) \Bigr\}\\
& = \PP(T=1) \Bigl\{ \mathbb{E}\bigl[\mathbbm{1} \{ R \le \theta \}|  T=1, \theta \bigr] - (1 - \alpha) \Bigr\}.
\end{align*}
Therefore,
\begin{align*}
    \mathbb{E} [\mathrm{IF}(r_\alpha, X,Y, R, T; \eta, m^\star, \gamma^\star)] = 0.
\end{align*}
\end{itemize}
Finally, by the double robustness property of $\mathrm{IF}(\cdots)$, it can be easily verified that $\partial_t \E [ \mathrm{IF}(\cdots; \eta_t, m^\star, \gamma^\star)]/\partial \eta_t$ and $ \partial_t \E [ \mathrm{IF}(\cdots; \eta^\star, m_t, \gamma^\star)]/\partial m_t$ for regular parametric submodels $\eta_t$ and $m_t$ are both zero at the truth. And this concludes our claim that $\mathrm{IF}(\cdots)$ is the efficient influence function up to a proportionality constant.
\end{proof}

\end{document}